\documentclass[]{hdsr} 

\usepackage[linesnumbered]{algorithm2e}
\usepackage{textcomp}
\usepackage{thm-restate}
\usepackage[style=apa, natbib]{biblatex} 
\usepackage{amsmath} 
\usepackage{array,multirow} 
\usepackage{graphicx}
\usepackage{tikz}
\usepackage{lipsum}
\usepackage[bottom=1.0in]{geometry} 
\usepackage{color, colortbl}

\newtheorem{definition}{Definition}

\newtheorem{theorem}{Theorem}

\usepackage{xspace}

\DeclareMathOperator*{\argmin}{arg\,min}

\let\vec\mathbf

\allowdisplaybreaks

\usepackage{mathtools}

\DeclarePairedDelimiter\floor{\lfloor}{\rfloor}

\DeclareMathOperator{\ngeo}{\Gamma} 
\DeclareMathOperator{\nchild}{children} 
\DeclareMathOperator{\nparent}{parent} 
\DeclareMathOperator{\nnode}{\gamma} 
\DeclareMathOperator{\nroot}{\nnode_0} 
\DeclareMathOperator{\leaves}{leaves} 

\begin{document}
\newgeometry{bottom=1 in} 


\begin{center}

  \title{The 2020 Census Disclosure Avoidance System TopDown Algorithm}
  \maketitle

  \vspace*{.2in}
  
  \begin{tabular}{cc}
    John M. Abowd,\upstairs{\affilone} 
    Robert Ashmead,\upstairs{\affilone}
    Ryan Cumings-Menon,\upstairs{\affilone}
    Simson Garfinkel,\upstairs{\affilfour}\\
    Micah Heineck,\upstairs{\affilsix}
    Christine Heiss,\upstairs{\affilsix}
    Robert Johns,\upstairs{\affilsix}
    Daniel Kifer,\upstairs{\affilone,\affilthree}\\
    Philip Leclerc,\upstairs{\affilone}
    Ashwin Machanavajjhala,\upstairs{\affilfive, \affilseven}
    Brett Moran,\upstairs{\affilone}
    William Sexton,\upstairs{\affilfour, \affilseven}\\
    Matthew Spence,\upstairs{\affilone}
    Pavel Zhuravlev\upstairs{\affilone}\\   
  \\[0.25ex]
   {\small \upstairs{\affilone} U.S. Census Bureau} \\
   {\small \upstairs{\affilfour} Formerly, U.S. Census Bureau} \\
   {\small \upstairs{\affilfive} Duke University} \\
   {\small \upstairs{\affilthree} Penn State University} \\
   {\small \upstairs{\affilsix} Knexus Research Corporation}\\
   {\small \upstairs{\affilseven} Tumult Labs}
  \end{tabular}
   \emails{
    \upstairs{*}The views expressed in this technical paper are those of the authors and not those of the U.S. Census Bureau or the U.S. Department of Homeland Security. Statistics reported in this paper have DRB clearance number CBDRB-FY-20-DSEP-001. This paper is forthcoming in the \textit{Harvard Data Science Review}.  Final pre-print April 7, 2022.}
  \vspace*{0.3in}
\begin{abstract}
    The Census TopDown Algorithm (TDA) is a disclosure avoidance system using differential privacy for privacy-loss accounting. 
The algorithm ingests the final, edited version of the 2020 Census data and the final tabulation geographic definitions. The algorithm then creates noisy versions of key queries on the data, referred to as measurements, using zero-Concentrated Differential Privacy. Another key aspect of the TDA are invariants, statistics that the Census Bureau has determined, as matter of policy, to exclude from the privacy-loss accounting. The TDA post-processes the measurements together with the invariants to produce a Microdata Detail File (MDF) that contains one record for each person and one record for each housing unit enumerated in the 2020 Census. The MDF is passed to the 2020 Census tabulation system to produce the 2020 Census Redistricting Data (P.L. 94-171) Summary File. This paper describes the mathematics and testing of the TDA for this purpose.
\end{abstract}
\end{center}

\vspace*{0.15in}
\hspace{10pt}
  
\restoregeometry

\section{Introduction}
\label{sec:intro}
Differential privacy \citep{DMNS06} (henceforth DP) is considered the gold standard in privacy-protected data publication---it allows organizations to collect and publish statistics about groups of people while protecting the confidentiality of their individual responses. Initially adopted by the U.S. Census Bureau in 2008 for the OnTheMap product \citep{ashwin08:map,onthemap}, it has since seen development by Google \citep{rappor,prochlo}, Apple \citep{applediffp},  Uber \citep{elasticsensitivity}, and Microsoft \citep{DingKY17}. This paper describes the implementation of the TopDown Algorithm, an algorithm developed within the differential privacy framework for the 2020 Census of Population and Housing in the United States \citep{abowd18kdd}.

The 2020 Census implementation will be among the largest deployments of differential privacy using the trusted-curator (centralized) model and, arguably, will have the highest stakes of any deployed formal privacy system, since decennial census data are used for apportionment, redistricting, allocation of funds, public policy, and research. The TopDown Algorithm (TDA) is the name given to the DP mechanisms, optimization algorithms, and associated post-processing that are used for research and production. TDA generates confidentiality-preserving person- and housing-unit level data called microdata detail files (MDF) with demographic and housing-unit information from the resident populations of the United States and Puerto Rico. We refer to the system that carries out the TDA as the 2020 Census Disclosure Avoidance System (DAS). The DAS is the collection of computer programs that implements statistical disclosure limitation for the census. It replaces the record-swapping system used from 1990 to 2010 for decennial censuses.

This paper provides an overview of the DAS and TDA, discusses the motivation for adopting DP as the privacy-loss accounting framework as compared to other disclosure avoidance frameworks, and summarizes the critical policy considerations that motivated specific aspects of the design of TDA. We focus on the implementation of TDA that the Census Bureau used to release the 2020 Census Redistricting Data (P.L. 94-171) Summary File (redistricting data, hereafter) \citep{CensusFederalRegistryRedistrictingDataFile}. We discuss the utility of the TDA, but leave for future work a discussion of how users can best analyze the released data and how the privacy semantics---the interpretation of DP---are modified in the presence of invariants.

The remaining sections of the paper are organized as follows. Section \ref{sec:whyDP} lays out the rationale for adopting DP as the disclosure avoidance framework for the 2020 Census. Section \ref{sec:policy} discusses the overarching policy and practical constraints that governed the implementation. Section \ref{sec:census_est} discusses the many sources of uncertainty that are inherent in producing census population data. Section \ref{sec:redistricting} provides the schema, invariants (tabulations that are not processed by the differentially private mechanisms) and constraints for the specification of the redistricting data, the first data product  released using the methods described in this paper. Section \ref{sec:mechanism:overview} provides an overview of the DP mechanism implemented in the TDA as well as the relevant mathematical properties. Section \ref{sec:estimation} goes into detail about estimation routines and improvements to the algorithm. Section \ref{sec:utility_experiments} describes the tuning and testing of the DAS over time. Section \ref{sec:utility_experiments2} describes a set of experiments carried out to study the effect of certain design choices on the utility of the results using the 2010 redistricting data. Section \ref{sec:conclusion} concludes.

\section{Differential Privacy: What is it and why use it?}
\label{sec:whyDP}
The decennial census data are used to apportion the House of Representatives, to allocate at least 675 billion dollars of federal funds every year,\footnote{This is the latest estimate produced by the Census Bureau \citep{hotchkiss:phelan:2017}. Independent researchers place the figure at 1.5 trillion dollars \citep{reamer:2020}. The effect of the decennial census data on funding occurs primarily through its effect on the annual Population Estimates Program.} and to redistrict every legislative body in the nation.\footnote{From the text of Public Law 94-171 ``To amend section 141 of title 13, United States Code, to provide for the transmittal to each of the several States of the tabulation of population of that State obtained in each decennial census and desired for the apportionment or districting of the legislative body or bodies of that State.'' \citep{pl94:law}}
The accuracy of those data is extremely important.
The Census Bureau is also tasked with protecting the confidentiality of the respondents and the data they provide.
This dual mandate presents a challenging problem because data accuracy and data privacy are competing objectives \citep{abowdschmutte2019}.
For any given privacy-protection procedure, more accuracy means less privacy, and the theory underlying differential privacy and related research have helped to quantify that tradeoff. We use the term \emph{formally private} to encompass frameworks that include the many variants of differential privacy. Pure differential privacy, approximate differential privacy, and concentrated differential privacy are all examples of formally private methodologies. A disclosure avoidance framework is formally private if it is constructed from randomized mechanisms whose properties do not depend on the realized confidential data nor, ideally, on limitations on the attacker’s information set. A formally private framework must provide an accounting system that quantifies the privacy loss associated with the ensemble of published query answers. TDA uses the formal privacy framework defined by zero-Concentrated Differential Privacy which will be defined below.



There are several reasons why the Census Bureau introduced formally private methodology for disclosure avoidance.
First, the vulnerabilities of traditional disclosure limitation methods are well known among privacy researchers and legacy systems have been attacked in various ways \citep{Dinur:Nissim:2003, dobra2000bounds,barth2012re, sweeney2002k, homer:etal:2008, narayanan2008robust, hansell2006aol, cohen2018linear, kifer2009attacks, wong2007minimality, garfinkel2015identification, fienberg2005preserving, dwork:etal:2017}.
Access to petabyte-scale cloud-based computing resources and software libraries designed to use these resources has increased enormously. At the same time, the amount of commercially available or independently held data on individuals that could be used as auxiliary information in order to re-identify individuals in Census Bureau products has also exploded. Unlike the Census Bureau, which has operated under a strict data confidentiality statute since 1954 \citep{title13}, Internet-based data aggregators like Apple, Facebook, Google, LinkedIn, Microsoft, Uber, Twitter, and many others only recently became subject to tight privacy regulation in the form of the California Consumer Privacy Act \citep{calprivacyact} and the European Union General Data Protection Regulation \citep{EUgeneralDataProtection}. The presence of vast arrays of personal information held by private companies and state actors in combination with software like Gurobi, CPLEX and GLPK designed to solve many 
NP-hard systems of billions of simultaneous equations finally realized the scenario that national statistical agencies have known was a vulnerability since Ivan Fellegi's original paper on data confidentiality \citep{fellegi:1972}.

Formal methods like differential privacy were developed specifically to avoid the vulnerabilities of traditional disclosure limitation methods. 
Furthermore, in contrast to the previous disclosure avoidance methods used for decennial censuses, the exact methodology and parameters of the randomized mechanisms that implement differential privacy are \emph{transparent}, meaning an organization can release the source code and parameters of the mechanisms, but not the actual random numbers used, without compromising the privacy guarantees \citep{DMNS06}.

In addition to transparency, two other key properties make DP methods attractive compared with traditional disclosure avoidance methods. These are the absence of degradation under post-processing and adaptive composition. 
First, \emph{post-processing} \citep{DMNS06, pinq, Dwork:2014:AFD} means that if we run an algorithm ${A}$ with no direct access to the confidential data $X$ on the output of the differentially private mechanism ${M}(X)$, then the composed algorithm ${A(M}(X))$ also satisfies differential privacy with no additional privacy loss. Second, differential privacy has an \emph{adaptive composition} property. If we use a DP mechanism $M_1(X)$ to produce an output $\omega$ and a second DP mechanism $M_2(\omega, X)$, then the composed output is also differentially private and the total privacy-loss of the composed mechanism is a subadditive and smoothly bounded function of the privacy loss of $M_1$ and $M_2$ \citep{Dwork:2014:AFD, murtagh:vadhan:10.1007/978-3-662-49096-9_7}. 

To the best of our knowledge, no traditional disclosure avoidance method satisfies transparency, non-degradation under post-processing, and composition \citep{abowd:schmutte:2015}. 
In particular, the household record-swapping method used in the 1990, 2000 and 2010 Censuses \citep{mckenna:2018} is not transparent and degrades poorly under composition. Its parameters and the details of the swapping algorithm cannot be published without compromising the privacy guarantee.
Aggregation, which was combined with swapping in previous censuses, does not preserve uncertainty about the underlying data under post-processing because it allows re-identification via reconstruction of the microdata \citep{Dinur:Nissim:2003, JASON:2020}. 
DP mechanisms also degrade under composition but in a predictable and smoothly bounded manner, thus avoiding highly accurate database reconstructions if agencies reasonably control their expended privacy-loss budget. DP mechanisms do not degrade under post-processing. 
For the traditional disclosure limitation method of cell suppression \citep{cox:1980,cox1987research}, an arbitrarily large part of the suppression pattern can be unraveled if an attacker knows the value of a single cell that the system assumed was not known, which can occur either as prior knowledge or because some other group did a cell-suppressed release of overlapping data but didn't suppress the same cells. Cell suppression has formal guarantees, but they only hold against an attacker who knows exactly what was published and nothing else relevant. Suppression breaks down in a brittle, discontinuous way when the attacker knows more. Such knowledge may not exist when the data are published, but may enter the public domain at a later point in time. In contrast, formal privacy provides guarantees against general rather than specific classes of attackers. While the protection of formally private systems also degrades with uncoordinated independent data releases, this degradation is smoothly bounded and fully quantified by the cumulative global privacy-loss budget.

An important aspect of formal privacy is the acknowledgment that the attacker's knowledge plays a large role in disclosure avoidance and that assumptions about this knowledge are inherently imperfect. Thus, the definitions used in formal privacy systems attempt to limit the extent to which the protections depend upon the specifics of the attacker's prior information. These limitations usually take the form of worst-case analysis and in that sense can be viewed as a minimax approach to a class of privacy protection problems \citep{wasserman:zhou:2010}. 

Rather than attempting to make statements about the absolute privacy risk, which requires specifying the attacker's information set and strategy, DP shifts the focus to relative privacy guarantees, which hold for arbitrary attacker prior beliefs and strategies. For example, rather than saying Alice's data are not re-identifiable given the output of the DP mechanism, we say that Alice's incremental risk of re-identification is measured by the difference between the output of the DP mechanism that did not include Alice's data as input versus the same one that did. An attacker or user should be able to make a statistical inference about Alice---one that depends on aggregate data---but not an identifying inference---one that is enabled only by the presence of Alice's information as input to the DP mechanism. The mathematical analysis that quantifies these limitations on identifying inferences is called \emph{privacy semantics} \citep{kasiviswanathan2014semantics, balle2020hypothesis, dong2021gaussian, kairouz2015composition, wasserman:zhou:2010}.

\subsection{Basic definitions} \label{subsec:basic_defs}
There are multiple definitions of formal privacy, each having slightly different motivation and mathematical primitives. We use two specific formalizations as part of the TDA. The primary definition is that of approximate differential privacy or $(\epsilon,\delta)$-differential privacy. The parameter $\epsilon$ is commonly referred to as ``the privacy-loss parameter'' for this definition. The parameter $\delta$ represents the amount of departure from pure differential privacy and is typically set at values that are smaller than $1/N$, where $N$ is the number of records in the dataset, because a mechanism that releases a record at random satisfies $(0, 1/N)$ differential privacy.  The second definition is called zero-Concentrated Differential Privacy (zCDP), which implies $(\epsilon,\delta)$-differential privacy and has very robust privacy-loss accounting tools and composition properties. Concentrated differential privacy was introduced by \citet{DBLP:journals/corr/DworkR16} to provide tighter analysis of privacy loss due to composition. We implement the variant defined by \citet{10.1007/978-3-662-53641-4_24,bun2016concentrated}, which can be used to provide a complete privacy-loss accounting framework using the discrete Gaussian distribution for noise addition \citep{NEURIPS2020_b53b3a3d,canonne2020discrete}. 

\begin{definition}\label{def:dp}
    (Pure/Approximate Differential Privacy) Given $\epsilon \geq 0$ and $\delta \in [0,1]$, a randomized mechanism ${M}: \mathcal{X}^{n} \rightarrow \mathcal{Y}$ satisfies $(\boldsymbol\epsilon,\boldsymbol\delta)$\textbf{--differential privacy} if for all $x, x' \in \mathcal{X}^n$ differing on a single entry (denoted $x \sim x'$ hereafter),\footnote{In the literature, datasets that differ on a single entry are called \emph{neighboring} datasets. There are two different versions of neighboring. (1) \emph{unbounded}: two datasets are considered neighbors if one can be obtained from the other by adding or subtracting a single record; (2) \emph{bounded}: two data sets are considered neighbors if one can be obtained the from other by modifying a single record. In this paper we use bounded neighbors because the total population of the U.S. (and of Puerto Rico) is invariant therefore neighboring datasets must have the same total record count.} and all measurable events $E \subset \mathcal{Y}$, we have that $P (M(x)\in E) \le e^{\epsilon}P(M(x')\in E)) + \delta$ \citep{dwork2006our}. \textbf{Pure differential privacy} occurs for the case $\delta=0$ \citep{DMNS06}.
\end{definition}

\begin{definition}\label{def:zcdp}
    (zero-Concentrated Differential Privacy (zCDP)) Given $\rho \geq 0$, a randomized mechanism $M : \mathcal{X}^n \rightarrow \mathcal{Y}$ satisfies \textbf{$\boldsymbol\rho$--zero-concentrated differential privacy} ($\rho$--zCDP) if  for all $x \sim x' \in \mathcal{X}^n$, and all $\alpha \in (1,\infty)$:
    $$
    D_\alpha(M(x)||M(x')) \le \rho \alpha,
    $$
    where $D_\alpha(P||Q) = \frac{1}{\alpha-1} \log \left( \sum_{E\in \mathcal{Y}} P(E)^\alpha Q(E)^{(1-\alpha)} \right)$ is the Renyi divergence of order $\alpha$ of the distribution $P$ from the distribution $Q$  \citep{10.1007/978-3-662-53641-4_24, bun2016concentrated}.
\end{definition}

The reason for specifying two privacy definitions is that, during the course of implementing the TDA, we observed substantial accuracy improvements from adopting a discrete Gaussian mechanism \citep{canonne2020discrete} for the generation of noisy measurements versus the geometric mechanism \citep{ghosh2012universally}. While we quantify our privacy semantics using approximate differential privacy, privacy-loss accounting using the discrete Gaussian mechanism more naturally aligns with zCDP---allocating $\rho$ rather than $\epsilon$. Therefore, we adopted the zCDP privacy-loss accounting for TDA. To permit comparison to other DP implementations and because reasoning about the complete distribution of the privacy-loss random variable rather than its tail area is more difficult to communicate to policymakers, we translate the zCDP privacy parameter $\rho$ into the corresponding $(\epsilon, \delta)$ values for approximate differential privacy. When converting $\rho$ to $(\epsilon, \delta)$, we intentionally did not use the tightest conversion \citep{NEURIPS2020_b53b3a3d,canonne2020discrete}, but instead chose to use a conversion \citep{10.1007/978-3-662-53641-4_24} that overestimates $\delta$ and is interpretable as a upper bound on the probability that a particular level of $\epsilon$ does not hold \citep{asoodeh:etal:2020:9b52f4c81b6f4c50a766fe9675b81066}.



Rather than determining a single $(\epsilon, \delta)$ pair, zCDP produces a continuum of $(\epsilon, \delta)$ pairs each corresponding to a particular value of $\rho$, and all satisfying Definition \ref{def:dp}. The privacy-loss accounting is done using $\rho$. See Figure \ref{fig:curvepy} for an examples of the $(\epsilon, \delta)$ curve based on a zCDP value of $\rho=1$ and $\rho=2.63$ (the final production setting of TDA for the redistricting data). When $\delta = 0$, Definition \ref{def:dp} is that of pure differential privacy \citep{DMNS06}. The curve in Figure \ref{fig:curvepy} is based on the privacy-loss random variable and indicates that zCDP provides privacy protection that weakens as $\delta$ approaches 0, but never results in catastrophic failure, which some approximate differential privacy mechanisms do permit.\footnote{Catastrophic failure means one or more records is released with probability $min(n\delta,1)$, where $n$ is the number of records in the dataset.} 
No single point on Figure \ref{fig:curvepy} summarizes the value of rho completely. In public presentations of the work, the Census Bureau has provided the $(\epsilon, \delta)$ pairs associated with $\delta=10^{-10}$ to facilitate the transition from $(\epsilon, \delta)$ privacy-loss accounting to rho-based privacy-loss accounting. However, a privacy analysis of only this $(\epsilon, \delta)$ pair would lead to a very incomplete analysis of the privacy protections provided by a given setting of $\rho$. We do not attempt to further interpret the $\delta$ parameter in this work.

\begin{figure}[h]
\caption{Example of the $\epsilon$-$\delta$ curve for $\rho=1$ and $\rho=2.63$ }
\centering
\includegraphics[width=0.9\textwidth]{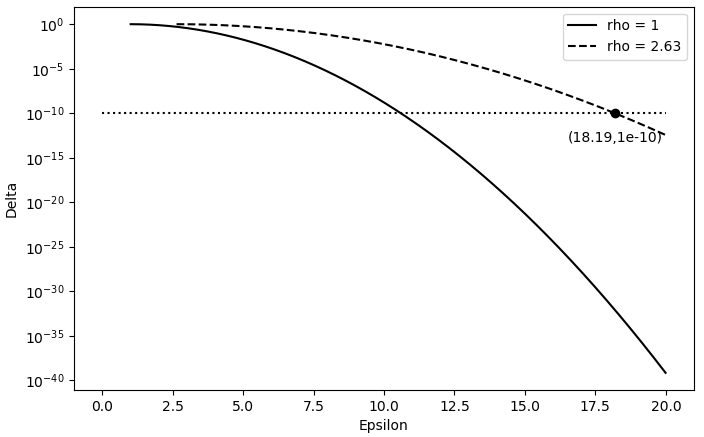}
\label{fig:curvepy}
\end{figure}

Definition \ref{def:dp} describes a property of the output distribution of DP mechanism ${M}$. It characterizes a guarantee that similar input databases will produce approximately the same output distribution. That is, an adversary cannot learn much more about the specific data of an arbitrary individual in the database  from the mechanism's output than if the record for that individual had been removed and replaced with an arbitrary default set of values \citep{kasiviswanathan2014semantics}. 
The translation to confidentiality protection comes from the argument that, if a person's record is first changed to a default set of values, then their confidentiality is protected because their data were never used in any published statistic. The degree to which the DP mechanism's output distribution is close to this private counterfactual depends on the privacy-loss parameter $\rho$ for zCDP or $\epsilon$ for $(\epsilon,\delta)$-DP. The parameter $\delta$ is the amount by which the pure DP guarantees are violated. For zCDP, $\rho$ provides a continuum of $(\epsilon, \delta)$ pairs with these interpretations. DP mechanisms that implement zCDP can be tuned to balance the data accuracy and privacy loss of a particular application. In the case of the 2020 Census and other government statistical applications, setting these values is (or should be) the job of policymakers. TDA was designed to permit policymakers, in particular the Census Bureau's Data Stewardship Executive Policy Committee (DSEP), to set the global privacy-loss parameter and its allocation to each query. The typical goal is to achieve the maximum accuracy for a given privacy-loss parameter ($\rho$ or $\epsilon$). However, DSEP's instructions generally tried to minimize the privacy-loss parameter given a lower bound on utility (see Section \ref{sec:mechanism:overview}).

\subsection{Data curator models}
The \emph{data curator} is the organization that collects and retains the data to be protected, often in a formal database. The 2020 Census application of DP uses the \emph{trusted-curator} model, sometimes called the centralized DP model, meaning that a legally authorized entity (the U.S. Census Bureau) collects and stores the confidential data and is responsible for the release of tabulations (queries) based on those data applying appropriate confidentiality protection techniques on the tabulations. This is in comparison with the \emph{untrusted curator} model, or local DP model, in which confidentiality protection techniques are applied to individual data records as they are collected and thus the curator never observes the confidential data directly, only through the output of the local DP mechanism. We point out the differences between these two models because the DP methods associated with them are necessarily different. An early example of a privacy-protection method in the untrusted curator model is randomized response \citep{warner1965randomized} in which the respondent answers a question truthfully with some predetermined probability, but otherwise supplies a fixed answer (or answers a different question). The advantage of the trusted curator model is that it allows for DP methods that are more statistically efficient, meaning that the results will be more accurate given the same privacy-loss parameter.  

\section{Policy and Practical Constraints}
\label{sec:policy}

In designing disclosure avoidance methods for the 2020 Census, several practical and policy considerations played a large role. Decennial census products are generally tables with counts of individuals, households, group quarters residents, or housing units with certain characteristics. The tabulation systems in place for the 2020 Census assume that the input is an individual or housing-unit record-level dataset in which each row represents a person or a particular housing unit. This type of dataset is often called \emph{microdata}. In order to align with the tabulation system, the design specifications for TDA mandated the production of microdata---a record-level image for each person (or housing unit) that, when tabulated, produces the protected query answers (not their confidential counterparts).

Producing privacy-protected microdata is not a common requirement for disclosure avoidance systems implemented using the centralized DP model, which are often designed to produce an unbiased noisy measurement of the answer to the original confidential query according to the use case for the queries. 
An important implication of the microdata requirement is that it forces the final output of TDA to be non-negative and integral because there is no accommodation  of a negative weight or fractional person in the Census Bureau's microdata tabulation systems.\footnote{Positive tabulation weights greater than unity are typically used in sample surveys, as opposed to full-enumeration censuses, with weights related to the inverse probability of selection into the sample. These tabulation systems could be modified to use negative weights, but a different statistical model for the meaning of the weight would be required.} 
From a practical point of view, this makes complete sense.
Some might find it strange if the Census Bureau reported that there were $-1.3$ persons with characteristic $A$ in geography $B$. Any disclosure avoidance system that allows for changing counts that are zero to non-zero values and enforces non-negativity, including suppression and swapping, suffers from systematic bias \citep{abowd:schmutte:2015}. 

There are two ways that the microdata requirement could be met by the disclosure avoidance system. The differentially private mechanism could directly output data consistent with the microdata requirement. Alternatively, the output of the differentially private mechanism could be post-processed to conform to the microdata requirement. The first approach could theoretically be achieved by using the exponential mechanism \citep{mcsherry2007mechanism}; however, that would require the mechanism to use the fully specified output range schema because the exponential mechanism relies on sampling from that range. For our application, that approach was infeasible because it required enumerating the discrete elements of the output space that simultaneously met the microdata, invariant, edit constraint, and structural zero requirements. Not only would this have required a new theoretical apparatus to solve, but during the course of the development of the DAS the requirements changed several times. Instead, the DAS uses post-processing to impose these requirements hierarchically using a preexisting framework from the operations research literature, specifically total unimodularity \citep[Sections III.1 and III.2]{wolsey:1999}.

There are some noticeable downstream consequences of imposing non-negativity. If the confidential count of characteristic $A$ in geography $B$ is $0$ and the TDA output of that count must be non-negative, then the microdata-based estimator, constrained to use non-negative weights, will be positively biased. This applies to both the post-processing case and the exponential mechanism case above, but note that when using the discrete Gaussian mechanism, the noisy measurements themselves are unbiased. In addition, the tabulation from microdata makes the errors in the disclosure avoidance system data dependent. 
Controlling the algorithmic biases and the data-dependent errors---properties due entirely to post-processing and not the DP mechanisms---is the hardest implementation challenge for TDA \citep{anon:2021}. Inferential methods using a posterior sampling framework and the noisy measurements directly rather than the post-processed tabulations have shown promising results, though there are computational limitations given the scope of the 2020 Census \citep{seeman2020private}.

A second, but related, requirement for the 2020 Census implementation was internal consistency, which is a natural property of tabulating all statistics from the same record-level input data with constant, non-negative weights (all unity for the 2020 Census).
Internal consistency means that a count of interest will be exactly the same if it appears in multiple tables or can be calculated from taking the sum or difference of counts from other tables. Consistency is a consequence of the algorithm producing privacy-protected microdata, but it is possible to have consistency without requiring tabulation from microdata.

A third requirement for the TopDown Algorithm was the implementation of invariants, edit constraints, and structural zeros.
By \emph{invariants} we mean counts that deliberately bypass the DP mechanism and are reported with certainty exactly as they are tabulated from the confidential data.
The most notable invariants for the 2020 Census are the state total resident populations because of their essential function in the apportionment of the House of Representatives. 
\emph{Edit constraints and structural zeros} are rules determined by Census Bureau subject matter experts, then applied to the confidential data as part of the processing that produces the Census Edited File, one of the two inputs to TDA. An example of an edit constraint is the requirement that a mother be at least some number of years older than her oldest natural child.

The privacy loss associated with the invariants is not quantifiable in the standard DP literature, and therefore the inclusion of invariants in the DAS complicates the interpretation of the privacy guarantees. Instead of going into detail here, we save this topic for a separate paper and note that the privacy-loss accounting for the DAS presented in this paper ignores any contributions from the invariants.    

One of the most important and challenging design considerations for the 2020 Census DAS taking into account the multidimensional nature of the utility of the census data. While apportionment is perhaps the most important use case, census data are used for hundreds if not thousands of public, private, and commercial applications ranging from redistricting, to population age-distribution estimation, to characteristics of the housing stock. Therefore, the TDA must produce output that satisfies the same edit constraints, where feasible, with controllable accuracy, in the sense of data utility, for all of its use cases. From an economic point of view, algorithmic design and implementation for TDA can be seen as allocating a scarce resource (the finite information content of the confidential data) between competing use cases (different accuracy measures) and respondent privacy loss \citep{abowdschmutte2019}.

Two excellent examples of competing use cases are redistricting and intercensal population estimates. Redistricting consists of drawing contiguous geographic entities that are mutually exclusive, exhaustive, and divide a political entity into voting districts of approximately equal populations. These districts also must satisfy the anti-discrimination provisions in Section 2 of the 1965 Voting Rights Act. To satisfy these requirements, nonpartisan redistricting experts work with the Census Redistricting and Voting Rights Data Office (part of the Census Bureau) to designate the smallest unit of geography---the atom from which the districts are assembled---and the taxonomy for race and ethnicity categories, which must comply with Office of Management and Budget Statistical Policy Directive 15 \citep{spd15}. Population estimates require annually updated populations for approximately 85,000 politically and statistically defined geographic areas and age by sex by race by ethnicity tables for all counties in the US and municipios in  Puerto Rico. Accuracy for redistricting targets geographic areas that are not yet defined when the disclosure avoidance is applied for statistics on population in 252 categories (voting age by race by ethnicity). Accuracy for population estimates targets one statistic (total population) in predefined geographies ranging in size from a few hundred to millions of people and about 2,000 statistics (sex by age in single years by major race and ethnicity categories) in counties and \emph{municipios}. Accuracy for redistricting is improved primarily by allocating privacy-loss to the lowest levels of geography (block groups and blocks in census terminology) and the most detailed race and ethnicity summaries. Accuracy for population estimates is improved by allocating privacy-loss to more aggregated geographies (counties and tracts) and favoring detail on age over detail on race and ethnicity.

\section{Sources of Uncertainty in the Census Estimation Problem}
\label{sec:census_est}

The use of DP mechanisms for the 2020 Census DAS will not be the first time uncertainty is deliberately introduced into the tabular summaries from a decennial census to enhance disclosure avoidance. And in any census, disclosure avoidance uncertainty is not the only source of uncertainty affecting the published statistics. Historically, the decennial census has used data suppression, blank and replace, swapping, and partial synthetic data for disclosure avoidance in the tabular summaries \citep{mckenna:2018}. In addition, sampling and coarsening along with swapping and partially synthetic data were used for public-use microdata samples \citep{mckenna:2019}. These methods also create uncertainty in inferences using the published data. Suppression creates a nonignorable missing data problem because the suppression rule depends on the size of the statistic being suppressed and on the correlation of the complementary suppressions with primary suppressions. Swapping and partially synthetic data affect inferences because of the randomness used by both procedures \citep{abowd:schmutte:2015}. However, the Census Bureau, like most national statistical agencies, does not provide any quantification of the uncertainty due to suppression, swapping or partial synthesis. Consequently, the effect on inferences from the published data is usually not acknowledged because researchers are unable to quantify it. For example, in the case of traditional swapping, the parameters of the swapping algorithm are secret. In the case of disclosure avoidance systems implementing DP mechanisms, the parameters and distributions of the mechanism's randomness are public and therefore can be quantified. TDA is an inference-valid disclosure avoidance system. Its uncertainty can be quantified and that quantification can be used to guide inferences from the published data; though, this can often be non-trivial.

The Census-taking process itself can be viewed as a random event. Modern incomplete data inference has formalized the many reasons why a realized census does not produce perfect, complete data. Post-enumeration surveys have shown the under-coverage of persons with some characteristics and the over-coverage of persons with other characteristics \citep{hogan:et:al:2013}. Additionally, observed data items can be missing or erroneous for some individuals. Given the observed data, missing data techniques are used to produce the Census Edited File (CEF)---the final version of the confidential census data. Statistically, the CEF is the completed census data from a collection of single-imputation missing data models.\footnote{We use the term ``completed data'' in the sense of \citet{rubin1974,rubin1976} to mean a data set with no missing data because all missing values have been imputed.} In the CEF, there is a record for every housing unit, occupied group quarters resident, and housing unit resident person alive on April 1, 2020 with all tabulation variables containing valid responses on every record. There is an enormous investment in the technology of specifying this ancillary estimator for the completed census data, given the observed census data, even when there is no adjustment for differential net under-coverage. Finally, disclosure avoidance techniques are applied to these completed census data to produce tabular summaries that protect confidentiality. In this paper, we use the CEF as the standard for estimating the accuracy of the DAS output. For better understanding the uncertainty introduced by disclosure avoidance, we compare the DAS output to small-scale simulations that quantify the uncertainty from census operational, measurement and coverage errors for the 2010 Census \citep{bell:schafer:2022}.

\section{Redistricting Data (P.L. 94-171) Summary File}
\label{sec:redistricting}

The Redistricting Data (P.L. 94-171) Summary Files are an essential data product as they are basis for state redistricting, official population counts, and other statutory uses. The tables are therefore sometimes called the ``redistricting data'' tables. Before going into additional details, it is important to have a sense of the geographic hierarchy that the Census Bureau uses in the creation of data products.\footnote{See https://www2.census.gov/geo/pdfs/reference/geodiagram.pdf for additional details.} The primitive element of the tabulation geographic hierarchy is the census block. All other geographic entities (states, tribal census areas, counties, places, etc.) can be defined as aggregations of blocks. The ``spine'' of the tabulation geographic hierarchy consists of the U.S., states (including the District of Columbia), counties, tracts, block groups, and blocks.\footnote{For the Commonwealth of Puerto Rico, the hierarchy is P.R., municipios, tracts, block groups, and blocks.} In each case, the entities at lower levels of geography on the spine aggregate without partitions to entities above them. That is, on the tabulation geographic spine a block belongs in one, and only one, block group; a block group belongs in one, and only one, tract; and so on up the spine. This tabulation geographic spine plays an important role in the TopDown Algorithm. Other geographic aggregations (e.g., Congressional districts, ZIP Code tabulation areas, incorporated places, etc.) are considered ``off-spine'' because they don't fit directly onto the tabulation spine's hierarchy.           
\subsection{Schema}
The  redistricting data tables for 2020 \citep{CensusDataProductsCrosswalk} are
\begin{enumerate}
    \item P1: Race
    \item P2: Hispanic or Latino, and not Hispanic or Latino by race
    \item P3: Race for the population 18 years and over
    \item P4: Hispanic or Latino, and not Hispanic or Latino by race for the population 18 years and over
    \item P5: Group Quarters (GQ) population by major GQ type
    \item H1: Occupancy status.
\end{enumerate}
Each table is produced down to the census block level. In order to tabulate from microdata, it is necessary to consider only a subset of census attributes with mutually exclusive levels. We designate these as the detailed queries for the person and housing-unit data, respectively. For person data the following attributes and levels constitute the detailed query to create the redistricting data tables:

\begin{enumerate}
    \item Block. 5,892,698 levels\footnote{This is the number of blocks in 2020 with at least one housing unit or at least one occupied GQ.}  
    \item Household or Group Quarters Type.  8 levels: household, correctional facilities for adults, juvenile facilities, nursing facilities/skilled-nursing facilities, other institutional facilities, college/university student housing, military quarters, other noninstitutional facilities;
    \item Voting Age. 2 levels: age 17 or younger on April 1, 2020, age 18 or older on April 1, 2020;
    \item Hispanic/Latino origin. 2 levels: Hispanic/Latino, not Hispanic/Latino;
    \item Census race. 63 levels: every combination of Black/African American, American Indian/Native Alaskan, Asian, Native Hawaiian/Pacific Islander, White, and some other race, except ``none of the above.''
\end{enumerate}

\noindent Excluding block, there are 2,016 mutually exclusive, exhaustive combinations of these attribute levels. Including block, there are 11,879,679,168 levels, of which the number of persons in each must be output by the DP mechanisms so that TDA can be used to tabulate each of the redistricting data person-level tables.

For housing-unit data the following attributes and levels constitute the detailed query to create the redistricting data tables:
\begin{enumerate}
    \item Block. 5,892,698 levels; 
    \item Occupancy status. 2 Levels: occupied, vacant.
\end{enumerate}

\noindent Therefore a total of 11,785,396 mutually exclusive, exhaustive combinations of the attribute levels must be output by the DP mechanisms so that TDA can be used to tabulate each of the redistricting data housing-unit tables.  

\subsection{Invariants}
In addition to the measurements associated with the specified tables, the TDA is programmed to bypass the DP mechanisms when instructed by policymakers at the Census Bureau. Queries that bypass DP mechanisms are called \emph{invariants}. Mathematically, an invariant is equivalent to designating infinite privacy loss for the specified tabulation, putting such queries outside the formal guarantees of differential privacy. In addition, invariants challenge the implementation of TDA because they 
may force the post-processing phase of TDA to attempt to solve computationally infeasible problems
(see Section \ref{sec:estimation}). As a consequence, the Census Bureau policymakers have specified invariants only when there was either a Constitutional mandate or an operational constraint in the conduct of the census.

The Constitutional constraint comes from the mandate to apportion the House of Representatives based on the ``actual Enumeration'' of the population on census day. Technically, this constraint only requires the allocation of the seats in the House to be invariant---not subject to changes due to the confidentiality protection system; however, DSEP determined that setting the populations of the fifty states, District of Columbia, and Commonwealth of Puerto Rico invariant was the most transparent and practical way to accomplish this objective. The policymakers accepted the risk associated with releasing these data outside the enhanced protections provided by differential privacy; that is, aggregation is the only disclosure avoidance applied to the state populations.

The operational constraints come from the methods used to keep the Census Bureau's Master Address File (MAF)---the frame for the decennial census and many other surveys---up to date. The MAF is a database of known addresses and features or characteristics of these addresses. The most important feature is \emph{living quarters}, which means that the address designates a domicile capable of housing at least one human being. Fixed domiciles come in two types: housing units, which usually domicile a single household, and group quarters, which usually domicile a group of unrelated individuals.\footnote{The technical definitions are based on whether or not there are separate entrances and shared facilities like kitchens and bathrooms, not on the familial relations among the inhabitants. Like many dichotomies applied to human populations, there is a gray area where housing units and group quarters are difficult to distinguish.} There are also living quarters that are not fixed domiciles such as service-based facilities (soup kitchens, shelters, etc.) and transient locations (tents, campgrounds, etc.). In addition, there are addresses for facilities that are not deemed living quarters at all, such as businesses, highway medians, and fire hydrants. 

Maintaining the MAF is a continuous process involving many operations. A few are salient to the disclosure avoidance system. First, by statute, the Census Bureau conducts the Local Update of Census Addresses (LUCA) late in the census decade. The LUCA operation enlists partners at all levels of government---states, recognized tribes, counties, and municipalities---to review the addresses in the MAF. In support of this operation, the Census Bureau publicly distributes certain summaries of the MAF, among these public tabulations are the number of addresses in each block that represent housing units and the number that represent group quarters of different types. While the addresses associated with these housing units and group quarters remain confidential, the counts are used to guide local partners to areas they believe may be incomplete. In these areas, the local partners supply candidate addresses, which are used to update the MAF. Second, throughout the decade, the Census Bureau cooperates with local expert demographers to update the location and types of group quarters facilities in an area. Since group quarters are in scope for the American Community Survey, this cooperation also affects that survey. Critically, the local demographers use public information on facilities like prisons, nursing homes, dormitories, and barracks to facilitate timely updates of the MAF group quarters data. For operational reasons, then, the Census Bureau treats the count of physical addresses on a block, their designation as housing units or group quarters, and the type of group quarters facility as public information. Note that while unoccupied (vacant) housing units are enumerated in the census, vacant group quarters are not. Hence, TDA treats housing units and occupied group quarters as invariant. The Census Bureau policymakers interpret this decision as requiring the disclosure avoidance to be consistent with public information. This decision means that a prison, or other group quarters, will be properly labeled in census tabulations if it is properly labeled in the confidential data, but the number and characteristics of the population living in those facilities or in housing units is protected by the DP mechanisms.

The complete list of invariants implemented in TDA is:

\begin{enumerate}
    \item total population for each state, the District of Columbia, and Puerto Rico;
    \item the count of housing units in each block, but not the population living in these housing units;
    \item the count and type of occupied group quarters in each block, but not the population living in these group quarters.
\end{enumerate}

\subsection{Edit constraints and structural zeros}
Edit constraints and structural zeros occur when a particular combination of values of the characteristics in the confidential data is deemed impossible a priori; that is, before any data are collected. For example, if the  relationship of person two to person one (also called the householder) is ``natural child,'' then person two cannot have age greater than person one's age.
Edit constraints and structural zeros modify the sample space of the final output data. In the case of the DAS, this would be the post-processed results (MDF) but not the intermediate noisy measurements.

Edit constraints are conceptually different from invariants; however, when examining the full constraint set in TDA, they are mathematically indistinguishable. Nevertheless, it is important to maintain the distinction. Edit constraints and structural zeros are imposed on the confidential data during the process of creating the Census Edited File. They are never violated in the input data to TDA by design; therefore, they must also be respected in the output data. Invariants, on the other hand, are statistics computed from the input data that deliberately skip the differential privacy mechanism. They are passed to the post-processing without noise injection. Hence, the distinction is that edit constraints and structural zeros define the sample space for the census. They are features of its design, not part of the information collected from respondents. Edit constraints and structural zeros are public information known before data collection while invariants are statistics collected in the census data that bypass the confidentiality protection system by policy.\footnote{Invariants complicate the standard privacy semantics. Edit constraints and structural zeros do not. We save this discussion for future work.}

The complete edit specification for the 2020 Census Edited File is not a public document, nor is the complete edit specification for the 2010 Census Edited File. The version of TDA that processed the redistricting data product contains only a few edit constraints and structural zeros. The most important of these are listed here:

\begin{enumerate}
    \item for the race variable, the combination ``none of the above'' is a structural zero;
    \item for occupancy status, there are no unoccupied group quarters.
\end{enumerate}

The race constraint is derived from OMB Statistical Policy Directive 15 \citep{spd15}, which requires the use of the categories listed in the redistricting data schema except for ``some other race,'' which was Congressionally mandated \citep{pl111:law}. Because a respondent cannot select ``none of the above,'' there are $2^6-1$ allowable combinations of the six race categories, not $2^6$. There are unoccupied group quarters on enumeration day, but they are removed from the MAF universe used in the Census Edited File. 

The structural zero on unoccupied group quarters, the occupied group quarters invariant, and the housing unit invariant interact in creating the TDA constraint set. If there is an occupied group quarters facility of a particular type in the relevant geographic area, then at least one person must be assigned in post-processing to that facility. If there is a housing unit in the particular geography, however, there is no similar requirement because occupancy status is not invariant for housing units. 
However, the number of householders (person one on the questionnaire) cannot be greater than the number of housing units. Aside from this, 
TDA treats the data for persons and housing units separately. Hence, occupancy status is protected without reference to the person data. This interaction is deliberate. Joining the housing unit data to the person data and then applying a differentially private mechanism to the joined data is a much harder problem because the sensitivity of queries (see Definition \ref{def:sensitivity}) about the composition of a household, including how many persons live in the housing unit, is the maximum allowable household size, not one. Join queries like this will be processed by DP mechanisms implemented in algorithms that do not post-process to microdata. These data products will be released  at a later date as part of the Detailed Demographic and Housing Characteristics (Detailed-DHC) data release. Users may be tempted to divide the total household population in a particular geography from table P1 by the total number of occupied housing units in that geography from table H1. For well-populated geographies this answer will be useful; however, for sparsely populated geographies it will not be reliable. The join query DP mechanisms implemented in the Detailed-DHC are designed to produce reliable data for average household size.

\section{Mechanism Overviews}
\label{sec:mechanism:overview}
\subsection{The discrete Gaussian mechanism}
The TDA and the block-by-block algorithm, which we describe in Appendix \ref{subsec:block_by_block}, rely on DP mechanisms responsible for generating the noisy measurements. For the TDA, we use the discrete Gaussian mechanism, based on the discrete Gaussian distribution (Definition \ref{def:discrete_gaussian}). The multivariate discrete Gaussian mechanism is described in Theorem \ref{thrm:mdg_cdp}. Setting $a=1$ yields the univariate case \citep{NEURIPS2020_b53b3a3d,canonne2020discrete}.
Let $\mathbb{Z}$ represent the space of integers and $\mathcal{X}^{n}$ a dataset of $n$ records where $\mathcal{X}$ is the sample space or data universe. Theorem \ref{thrm:mdg_cdp} uses the concept of query sensitivity provided in Definition \ref{def:sensitivity}. Query sensitivity is important because it affects the scale of the noise distribution for a fixed privacy-loss value.

\begin{definition}\label{def:sensitivity}
    ($L_p$ Query Sensitivity)
    For $p \in (1,2, \ldots)$, the \textbf{$L_p$ sensitivity} of the query $q:\mathcal{X}^{n}\rightarrow\mathbb{R}^a$ 
    for the dataset with sample space $\mathcal{X}^{n}$ is defined as 
    $$
    \max_{x \sim x' \in \mathcal{X}^n} ( \sum_{j=1}^{a}  |q_j(x) - q_j(x')|^p )^{\frac{1}{p}}.
    $$
\end{definition}

\begin{definition}\label{def:discrete_gaussian}
(Discrete Gaussian Distribution) Let $\mu,\sigma \in \mathbb{R}$ and $\sigma>0$. The \textbf{discrete Gaussian distribution}, denoted by $\mathcal{N}_{\mathbb{Z}}(\mu, \sigma^2)$, with location parameter $\mu$ and scale parameter $\sigma$ has probability distribution
$$
P(X=x) = \frac{\exp(-(x-\mu)^2/2\sigma^2)}{\sum_{y \in \mathbb{Z}} \exp(-(y-\mu)^2/2\sigma^2)} \; \forall \, x \in \mathbb{Z}.
$$
\end{definition}


\begin{theorem}\label{thrm:mdg_cdp}
(Multivariate Discrete Gaussian Mechanism; Theorem 14 \cite{canonne2020discrete} ) Let $\sigma_1, \cdots, \sigma_a > 0$ and $\rho \geq 0$. For a query $q:\mathcal{X}^n \rightarrow \mathbb{Z}^a$, suppose the rescaled query $q'=(\frac{q_1}{\sigma_1}, \dots, \frac{q_a}{\sigma_a})$ has $L_2$ sensitivity at most $\sqrt{2\rho}$. Define a randomized algorithm $M: \mathcal{X}^n \rightarrow \mathbb{Z}^a$ by $M(x) = q(x)+y$, where $ y_j \sim \mathcal{N}_{\mathbb{Z}}(0, \sigma^2_j)$ independently for $j=1, \cdots, a$.  Then $M$ satisfies $\rho$--zCDP.
\end{theorem}

The discrete Gaussian mechanism was not always the DP mechanism in the TDA. In earlier work we used the geometric mechanism, which like the discrete Gaussian mechanism produces errors with integer values but uses the double-geometric rather than the discrete Gaussian distribution. The discrete Gaussian mechanism was ultimately chosen based on empirical tests of the accuracy of the TDA using each mechanisms with comparable privacy-loss budgets. The discrete Gaussian distribution has smaller tail probabilities than the double-geometric distribution, and 
this provably reduces the worst-case errors associated with the microdata requirement \citep{anon:2021}.

\subsection{Notation}
The census-taking process results in a person-level dataset of $n$ records with $k$ characteristics (e.g., age, race, ethnicity) including geography. Each characteristic may take exactly one of a finite number of values, which define the sample space, $\vec{\chi}$. An equivalent definition exists for the sample space of housing units. The sample space for persons is defined by enumerating all combinations of the values of every person characteristic measured in the census. Thus, $c=\text{Cardinality}(\vec{\chi})$ is the number of rows in $\vec{\chi}$, and represents the total possible number of combinations of the characteristics of the persons. For the 2020 Census Redistricting Data (P.L. 94-171) Summary File, $c = 11,879,679,168$
for the person-level data and $c = 11,785,396$ for the housing-unit-level data. Of course, these calculations show the dimensionality of these sample spaces for a single person or unit. There are combinatorially more feasible databases satisfying these schemas that are also consistent with all invariants.

Geographic levels play an important role in the methodology. It is useful to consider the cardinality of the sample space disregarding geography. Let the sample space without geography be represented by $\vec{\chi^*}$ and its cardinality by $c^*=\text{Cardinality}(\vec{\chi^*})$. For the redistricting data case, $c^* = 2,016$ for the person-level data and $c^* = 2$ for the housing unit level data.         

Instead of representing the data as an $n \times k$ person-level matrix (also called a \emph{table} in the database literature), another way to represent the data is as a fully saturated contingency table (also called a \emph{histogram} in the privacy-preserving data publication literature) where every cell corresponds to a possible record type as defined by the schema and its value is the number of records of that type. Rather than using a multi-dimensional array, such as a contingency table, it is notationally convenient to consider instead the flattened table, a length $c$ vector flattened in any predefined ordering that spans the schema. Let $\mathbf{x}$ represent the contingency table vector.    

A set of linear queries on $\mathbf{x}$ is represented by an $(a \times c)$ matrix $\mathbf{Q}$ and the vector of query answers of length $a$ is obtained by matrix multiplication of the query matrix and the contingency table vector: $\mathbf{Q} \mathbf{x}$. Let $\mathbf{\widetilde{M}}$ represent a differentially private random vector of answers to a set of linear queries. Measurements have the form $\mathbf{\widetilde{M}} = \mathbf{Q} \mathbf{x} + \mathbf{y}$, where $\mathbf{y}$ is a vector of independent random variables. Specifically in the context of the TDA, we assume $\mathbf{y} \sim \mathcal{N}_{\mathbb{Z}}(0, \sigma^2_j)$, where $j=1, \cdots, a$, a discrete Gaussian distribution with potentially unique variances $\sigma^2_j$ for each of the $a$ noisy measurements.  In the TDA, to turn the set of $\sigma^2_j$ values into a $\rho$ value, we  use Theorem \ref{thrm:mdg_cdp}. See Appendix \ref{sec:approx_dp_proof} for additional details.

A key feature of census data and tabulation products is geography. We use the symbol $\ngeo$ to represent a geographic hierarchy (US, state, county, tract, block group, block for the United States; PR, muncipio, tract, block group, block for Puerto Rico.): a directed, rooted tree with $\nroot$ as its root. Given a node $\nnode\in \ngeo$, we let $\nparent(\nnode)$ represent its parent geographic node in the hierarchy and $\nchild(\nnode)$ represents the set of geographic nodes that are its children. We let $\leaves(\nnode)$ represent all of the leaves under node $\nnode$ and $\leaves(\ngeo)$ to represent all of the leaves. In particular $\leaves(\nroot)=\leaves(\ngeo)$. We refer to the set of all nodes at a fixed distance from $\nroot$ as a \emph{geographic level}, which in the case of the census may be one of the US, state, county, tract, block group, or block.\footnote{There are additional tabulation geographies such as Census-designated places and  American Indian-Alaska Native-Native Hawaiian tribal areas all of which can be constructed as aggregations of blocks.}


Let $\mathbf{x}_{\nnode}$ be the contingency table vector for a specific geographic node $\nnode$. By definition, all cells of the vector not associated with the given geography are $0$. Therefore it is convenient to view $\mathbf{x}_{\nnode}$ as a length $c^*$ vector, disregarding the geographic dimension of the contingency table. Similarly, let $\mathbf{\widetilde{M}}_{\nnode} = \mathbf{Q} \mathbf{x}_{\nnode} + \mathbf{y}$ be a geography specific version of the differentially private random vector of answers to a set of linear queries. For the purposes of the TDA, it is necessary to group all the children of a node together to jointly estimate the solution. Let $\mathbf{x}_{\text{children}(\nnode)}$ be the contingency table vector obtained by stacking $\mathbf{x}_{\text{child}}
\text{ for every}
\text{ child} \in \text{children}(\nnode)$, making it a $(c^* \times |\text{children}(\nnode)|)$ length vector where $|\text{children}(\nnode)|$ is the number of children of node $\nnode$. Let $\mathbf{\widetilde{M}}_{\text{children}(\nnode)}$ represent a the similarly stacked version for the differentially private measurements.        

The TopDown Algorithm processes the noisy measurments from the DP mechanism and the invariants into an estimate of the contingency table vector. We use $\hat{\bf{x}}$ or equivalently $\hat{\bf{x}}_{\nnode}$, $\forall \nnode \in \ngeo$, to denote the output of TDA.

Depending on the specifics of the invariants and edit constraints, both equality and inequality constraints may be imposed within TDA. Specifically, in the redistricting data use case, inequality constraints are due to invariants imposed on the counts of housing units and occupied group quarters because these put bounds on the number of people that must be in that geographic unit. For example, for the invariant on the number of occupied group quarters facilities of a certain type by block, there must be at least one person per group quarters type in that block because vacant group quarters facilities are out-of-scope for the census. We denote the constraints on the algorithm output as a result of the invariants as a set of equality and inequality constraints:  
\begin{align*}
    \mathbf{C}^{eq} \hat{\mathbf{x}} &= \mathbf{c}^{eq}\\
    \mathbf{C}^{u} \hat{\mathbf{x}} &\leq \mathbf{c}^{u},
    \end{align*}
where $\mathbf{c}^{eq}$ and $\mathbf{c}^{u}$ are calculated from $\mathbf{x}$ and possibly additional housing unit data. For example, total state populations are invariant; therefore, the state population for each state is an element of $\mathbf{c}^{eq}$ and the rows of $\mathbf{C}^{eq}$ force the elements of $\hat{\mathbf{x}}$ that correspond to components of the state population to sum to the invariant total. Another example is group quarters populations in any geographic unit. Group quarters must be occupied; therefore, if $k$ occupied group quarters facilities of a given type are present in a particular geographic unit, this generates a lower bound constraint of $k$ that is an element of $\mathbf{c}^{u}$, and the rows of $\mathbf{C}^{u}$ force the elements of $\hat{\mathbf{x}}$ that correspond to components of the group quarters population of that type to sum to at least $k$.\footnote{We specify all inequality constraints as upper bounds in this paper and implement lower bounds by negating the inequality.}

\subsection{Problem statement and utility criteria}

The census ultimately produces a set of tabulations that can be characterized as a query matrix $\mathbf{A}$ multiplied by the contingency table vector. Without any disclosure avoidance applied, the tabulations would be calculated as $\mathbf{A} \mathbf{x}$. While not the \emph{true value} as this term is defined in classical statistics---the unknown population quantity that is the estimand for a census or survey---because of the uncertainty described in Section \ref{sec:census_est}, it is treated as the reference value for purposes of disclosure avoidance because it is the result of answering the query directly from the confidential data. Therefore, the goal of the DAS is to implement a sequence of algorithms that gives optimal tabular estimates  $\mathbf{A} \mathbf{\hat{x}}$ relative to $\mathbf{A} \mathbf{x}$  according to some metric for a given privacy-loss budget. 

The utility of the algorithms developed for the redistricting data was first and foremost designed to meet stringent fitness-for-use accuracy targets for the redistricting and Voting Rights Act use cases. This is a challenging task because the geographic units of interest are the voting districts that will be drawn after the data are published. Hence, the target geographies cannot be specified in advance. Fortunately, these voting districts divide political entities like states, counties, incorporated places, school districts, and tribal areas whose geographies are pre-specified. Therefore, error metrics that optimize accuracy within these predefined political entities help properly control the metrics for their equal-population future voting districts. The collections of blocks in the voting districts are not arbitrary. They cover the political entity and form aggregates of approximately equal total population. When the political entity has a low population, there are usually only a few voting districts or none at all. While a complete description of the criteria for forming legislative districts is beyond the scope of this paper, more details about the compactness, contiguity, political boundary, and community of interest requirements are provided by the nonpartisan organization that represents states' interests in the development of redistricting data, the \citet{ncsl:2021}.

We performed explicit experimental parameter tuning
using 2010 Census data
to ensure that the largest racial/ethnic group (as defined by the OMB in Statistical Policy Directive 15) in off-spine geographies (e.g., minor civil divisions, Census places, American Indian Reservations, etc.) with populations of at least $500$ people, expressed as a percentage of the total population, was within $\pm 5$ percentage points of the enumerated percentage at least $95\%$ of the time.\footnote{For practical reasons related to the clock time required to run TDA and the license limitations on the number of simultaneous numerical optimization instances we could run (2048), we estimated the $95\%$ compliance using geographic variation within the same geounit type (county, block group, minor civil division, AIAN tribal area, etc.). At the conclusion of these experiments, we confirmed compliance with the $95\%$ target using 25 independent runs of TDA with the April 2021 parameter settings as discussed in Section \ref{sec:utility_experiments}.} We also examined the same metric expressed as a percentage of the voting-age population and found that controlling the ratio to the total population also controlled the error in the percentage of the largest race/ethnicity group in the voting-age population. This utility goal was developed in consultation with stakeholders and is analyzed extensively in \citet{wright:irimata:2021}. 

Independently, we also evaluated the performance of TDA by examining the absolute and relative error of many statistics by quantiles of the count in the 2010 CEF. These error metrics were used for total population, voting-age population, number of races, population in each OMB race category, and population in each ethnicity category within OMB race category. These quantile statistics were used primarily to determine the interactions of changes in the privacy-loss budget allocation among the queries, which cannot be done ex ante because, while the errors from the DP mechanism are independent of the data, the errors in post-processing are not.

Our accuracy assessments included many metrics \citep{DASMetricsOverviewApril28, DASMetricsOverviewNov16} using the 2010 Census data that were developed by demographers at the Census Bureau or suggested by external users. 
They included
many different error measures such as mean absolute error, mean absolute percent error, percent difference thresholds, outliers, as well as different characteristics of interest relevant to redistricting (total populations, total population 18 years and over, occupancy rates, Hispanic or Latino origin population, etc.) across many different geographic levels.


The existing policy and scientific literature provide very little guidance on managing a privacy-loss budget to trade-off accuracy on multiple dimensions. When tuning the TDA, our experiments were designed to find the smallest privacy-loss budget that met specified accuracy goals. In particular, we used a tight standard for redistricting data accuracy in the smallest voting districts that the Department of Justice Voting Section provided as examples of Section 2 scrutiny under that law. Finding the minimum privacy-loss budget that could achieve this goal illustrated where other accuracy objectives deteriorate---in particular, statistics at the tract and county level, which were down-weighted by the redistricting accuracy measure. Other subject matter experts, internal and external, then laid out their accuracy goals for these other statistics. Internal statistical disclosure limitation experts, including some of the authors of this paper, demonstrated that these accuracy goals---redistricting and general demographic---could be achieved with a modest privacy-loss budget at the block level. The iterative experimental tuning process is described in more detail in Section \ref{sec:utility_experiments}.

The statistical optimization problem can be summarized as: given accuracy targets based on $\mathbf{A} \mathbf{\hat{x}}$ relative to  $\mathbf{A} \mathbf{x}$, choose query matrices $\mathbf{Q}^* = \left\{ {\mathbf{Q}}_1, {\mathbf{Q}}_2 \cdots, {\mathbf{Q}}_m \right\}$ with corresponding privacy-loss budgets 
$\rho^* =$ \{$ \rho_1$, $\rho_2$,  $\cdots, \rho_m \}$ 
and estimator 
$$\hat{\mathbf{x}} = g \left(  \mathbf{M}^* ,\mathbf{Q}^*, \rho^*, \mathbf{C}^{eq}, \mathbf{c}^{eq}, \mathbf{C}^{u}, \mathbf{c}^{u} \right) $$
that meets the accuracy targets such that the total privacy-loss budget 
$\sum_{i=1}^{m} \rho_i$
is minimized, where $\hat{\mathbf{x}} \in \mathcal{Z}^{0+}_d$, meaning the estimated vector is of length $d$ and has non-negative integer elements, and $\mathbf{C}^{eq} \hat{\mathbf{x}} = \mathbf{c}^{eq}$, and $\mathbf{C}^{u} \hat{\mathbf{x}} \leq \mathbf{c}^{u}$, meaning the solution satisfies the constraints. Here we denote ${\mathbf{M}}^*= \left\{\mathbf{\widetilde{\mathbf{M}}}_1, \mathbf{\widetilde{\mathbf{M}}}_2, \cdots, \mathbf{\widetilde{\mathbf{M}}}_m \right\}$ as the differentially private vector of noisy measurements corresponding to $\mathbf{Q}^*$ and 
$\rho^*$. Note the three major elements of the algorithm design:
 
 \begin{enumerate}
 \item choosing which query matrices to use differentially private noisy measurements;
 \item choosing the accuracy of each of the noisy measurements; and
 \item optimally combining the noisy measurement information to produce a non-negative integer estimate of $\mathbf{x}$ that satisfies all constraints. 
 \end{enumerate}

While these algorithm elements should work hand-in-hand, we note that there could be many variations of (3) using the same measurements that are a product of elements (1) and (2) that produce results with good utility. Not only are there many objective functions that can be used for optimization, but many variations of methods to produce non-negative integers satisfying a set of constraints. Over the course of this work, we experimented with each of these elements in order to build the best algorithm for key redistricting and demographic use-cases.       

\subsection{Queries}
\label{subsec:queries}

While it would be possible to ask a linear query $\mathbf{Q}$ with generic integer elements, we restrict ourselves to queries that consists of only binary elements. Furthermore, while mathematically it is often convenient to write a single query matrix for the entire problem or for a single geography, in practice we think of sets of queries stacked together into a single query matrix and define them in terms of dimensions of the schema. In order to explain the specifics of the queries used in TDA, as listed in, for example, Table \ref{table:rho_allocation_persons_production}, we describe here the structure of the complete query matrix, which is partitioned by level of the geographic hierarchy and marginal query group within each geographic level. Privacy-loss budget ($\rho$) allocations are made to the geographic level and then to the marginal query group within each level. When these partitions are stacked vertically within geographic level and concatenated horizontally across the geographic hierarchy, they generate the full linear query matrix.

We use the term \textit{marginal query group} to refer to a matrix with rows given by linear queries such that the product with the data vector, $\mathbf{x}$, provides the estimates for a given marginal of the data.\footnote{DAS queries can be more completely characterized as being constructed by two steps. (a) Defining a list of ``recodes'' of individual schema attributes that could drop or skip some levels of the attribute and/or aggregate or union some attributes together. In general, these recodes-of-attributes are then mutually exclusive but not necessarily exhaustive; and (b) taking a Cartesian product represented via Kronecker products in sparse matrices of these recodes.} For example, in the redistricting data, we might ask the census race marginal query group and the voting-age marginal query group for a given geography.  By the census race query group, we mean the 63 total population counts in each of the OMB-designated race categories for the given geography.  By the voting-age query group, we mean the 2 total population counts for 17 or younger and 18 and older. We could also ask the query group census race by voting age which would be 126 counts of each of the levels of OMB-designated race crossed by voting age. In addition to individual or crossed dimensions of the schema, we also consider the cross of all dimensions of the schema 
as well as the collapse of all dimensions, which we give special names: the \emph{detailed} query and the \emph{total} query, respectively. Lastly, we also consider Cartesian products of mutually exclusive collapses of a schema attribute as a query. For example, we could collapse the 63-level census race attribute into a 3-level attribute with levels: 1 race selected, 2 races selected, and 3 or more races selected. A query strategy may also include Cartesian products of collapsed categories.

The individual marginal query groups we consider have two important properties. They are exhaustive in terms of the levels (we do not query 18 and older without also querying 17 and younger) and the levels are mutually exclusive. Along with the use of binary elements, these properties ensure that the $L_2$ sensitivity of a single marginal query group is at most and $\sqrt{2}$. $L_p$ sensitivity for a single query is defined in Definition \ref{def:sensitivity}. In the query group it means for any two possible neighboring databases, where neighboring is defined as the result of changing the value of a single record in the database (which maintains the total number of records), the maximum over all the queries in the group is at most $\sqrt{2}$. Similarly to the sensitivity of a single query (Definition \ref{def:sensitivity}), the sensitivity of a query group is important because it affects the scale of the noise needed for a fixed privacy-loss or $\rho$ value. The larger the sensitivity, the larger the scale of the noise that is needed in the algorithm. By construction, all queries in a query group receive the same precision (a function of the $\rho$ allocation and the sensitivity) for their DP mechanism noisy measurements, though the precision given to each query group may differ both within and between geographies.

\subsection{The TopDown Algorithm}
\label{subsec:TDA}

\label{tda-motivation} The motivation for the TDA is to leverage the tree-like hierarchical structure of census geographies in order to produce estimates of the contingency table vector that are more precise than algorithms focused on directly and independently estimating block-level vectors then aggregating\footnote{See Appendix \ref{subsec:block_by_block} for more information.} in a scalable manner.
The TDA can be thought of as two phases, a measurements phase and an estimation phase.\footnote{In other contexts, the estimation phase has sometimes been called \emph{post-processing}.} Table \ref{tab:TDA_summary} presents a schematic summary of the algorithm.

\begin{table}[]
    \centering
        \caption{The TopDown Algorithm Summary}
        \label{tab:TDA_summary}
    \begin{tabular}{|p{14.5cm}|}
        \hline
        \emph{Measurement Phase}\\
         \hline
         \begin{tabular}{l}
            (1) For $\text{level } i \in \{\text{US}, \text{state}, \text{county}, \text{tract}, \text{block group}, \text{block}\}$\\
            \hspace{0.5cm} (a) Determine the privacy-loss budget for the level $i$;\\
            \hspace{0.5cm} (b) Take differentially private noisy measurements $\widetilde{\mathbf{M}}_{\gamma}$ for all nodes in level $i$.\\
         \end{tabular}
         \\
         \hline
         \emph{Estimation Phase}\\
         \hline
         \begin{tabular}{l}
              (1) For the US root node $\gamma_0$ estimate the contingency table vector $\mathbf{x}_{\gamma_0}$ by\\
              \hspace{0.5cm} (a) Estimating a non-negative solution $\tilde{\mathbf{x}}_{\gamma_0}$ from the set of differentially private\\ 
              \hspace{0.5cm} noisy measurements $\widetilde{\mathbf{M}}_{\gamma_0}$, invariants, and edit constraints at the US level;\\
              \hspace{0.5cm} (b) Estimating a non-negative integer solution $\hat{\mathbf{x}}_{\gamma_0}$ from $\tilde{\mathbf{x}}_{\gamma_0}$ by controlled rounding. \\
              (2) For $\text{level } i \in \{\text{state}, \text{county}, \text{tract}, \text{block group}, \text{block}\}$, let $P_i$ represent the set of\\
              distinct parents among all nodes at level $i$. For each parent node $\gamma\in P_i$, estimate\\
              the joint contingency table vector $\mathbf{x}_{\text{children}(\gamma)}$ by\\
              \hspace{0.5cm} (a) Estimating a non-negative solution $\tilde{\mathbf{x}}_{\text{children}(\gamma)}$ from the set of differentially private\\
              \hspace{0.5cm} noisy measurements $\widetilde{\mathbf{M}}_{\text{children}(\gamma)}$, invariants, edit constraints, and aggregate constraints\\
              \hspace{0.5cm}  enforcing that the children sum to $\hat{\mathbf{x}}_{\gamma}$;\\
              \hspace{0.5cm} (b) Estimating a non-negative integer solution $\hat{\mathbf{x}}_{\text{children}(\gamma)}$ from $\tilde{\mathbf{x}}_{\text{children}(\gamma)}$ by\\
              \hspace{0.5cm} controlled rounding. 
         \end{tabular}
    \\
    \hline
    \end{tabular}
\end{table}

First, in the measurement phase, sets of differentially private noisy measurements are taken about key queries at each of the geographic levels in the hierarchy. Because queries at different geographic levels (say block and tract) will involve the same individuals, the total privacy loss will be the summation of the privacy loss at the individual geographic levels. Therefore, the distribution of the total privacy-loss budget across levels of the geography as well as the queries within a geographic level must be chosen with care. In the estimation phase, the goal is to produce a non-negative integer estimate of the contingency table vector from the differentially private noisy measurements.

As the name implies, estimation is carried out sequentially, starting with the root geography (US). The algorithm estimates the non-negative integer US-specific contingency table vector from the set of US-level differentially private queries on the data, maintaining any equality or inequality constraints implied by the invariants or due to edit constraints and structural zeros. Due to the complexity of the problems, this is done in a two-step manner. First, we estimate a constrained non-negative weighted least squares solution from the US-level differentially private queries that respects all constraints implied by the invariants. Second, we perform a controlled rounding that finds a nearby non-negative solution that also respects the constraints. See Section \ref{subsec:estimation_routines} for more detail. Once the solution is found, the US-level estimated contingency table vector is considered fixed and used as an additional constraint on the next subproblem in the sequence.

The algorithm then estimates the state-level contingency table vectors, enforcing an aggregation constraint to the US-level contingency table vector and using  the state-level differentially private measurements on the data. This is done using a similar two-step process: first, solving a weighted non-negative least squares problem; then, solving a controlled rounding problem. However, in addition to any constraints that are the product of invariants or edit specifications, an aggregation constraint is added to both subproblems to ensure that the estimated state-level contingency table vectors sum to the US-level estimated contingency table vector.    

Fixing the estimated state-level contingency table vectors, the estimation steps are repeated within each state in order to estimate each of the county-level contingency table vectors, enforcing aggregation constraints for each state. The algorithm then proceeds iteratively to additional geographic levels (tract, block group, and block) until the contingency table vector is estimated for each block. In each step, the differentially private measurements taken at that level are used along with constraints due to invariants, edit specifications and aggregation.

\section{Estimation Routines and Improvements}
\label{sec:estimation}
\subsection{Basic estimation routines}
\label{subsec:estimation_routines}

In an ideal scenario, the contingency table vector $\mathbf{x}$ for all blocks could be estimated in its entirety simultaneously, minimizing error relative to all differentially private measurements at once and respecting invariant constraints; however, the scale of such a problem in the decennial census is several orders of magnitude too large given current computing limitations. Therefore, one way to view  TDA is that it breaks the problem down into smaller,  more computationally tractable sub-problems. Even after doing so, the size of the problems remains quite large. The desired output of the algorithm is a non-negative integer vector $\hat{\mathbf{x}}$ and therefore an integer programming approach would be the natural fit. Unfortunately, the size and complexity of the sub-problems is still too large to directly solve as an integer program. 
Therefore, we further break down each optimization step into two sub-problems: 1) a constrained weighted least squares optimization that finds an estimated non-negative continuous-valued contingency table vector; and 2) a controlled rounding problem that solves a simpler integer program than in the direct method and yields a non-negative integer contingency table vector. For simplicity, we refer to these as the least squares optimization problem and the rounder optimization problem, respectively.

In Sections \ref{subsec:basic_l2} and \ref{subsec:basic_l1}, we present the estimators for the least squares and rounder optimization problems, respectively, that we initially developed. Then, we provide the details of enhancements to these basic estimators using more complex estimators that find a solution by performing a series of optimizations and constraining passes.  We detail these more complex estimators in Sections \ref{subsec:l2_multi_pass} - \ref{subsec:l1_multi_pass}.

\subsubsection{The basic least squares estimator} \label{subsec:basic_l2}
There are two versions of the least squares estimation problem depending upon whether the subproblem includes a single node (as in Step (1)(a) of the estimation phase of TDA) or multiple nodes (as in Step (2)(a) of the estimation phase of TDA). Let $\mathbf{W}$ be the diagonal matrix with the inverse variances of the differentially private random variables along its diagonal. Then, the least squares estimator for a single node $\gamma$ is

\begin{align} \label{l2_opt_1}
    \tilde{\mathbf{x}}_{\gamma} \gets & \argmin_{\mathbf{x}_{\gamma}} (\mathbf{Q} \mathbf{x}_{\gamma} - \widetilde{\mathbf{M}}_{\gamma})^\top \mathbf{W} (\mathbf{Q} \mathbf{x}_{\gamma} - \widetilde{\mathbf{M}}_{\gamma}) \\
    & \text{subject to:} \nonumber\\
    & \mathbf{x}_{\gamma} \ge \mathbf{0}  \text{; and} \nonumber \\
    & \mathbf{C}^{eq} \mathbf{x}_{\gamma} = \mathbf{c}^{eq} \text{; and} \nonumber \\
    & \mathbf{C}^{u} \mathbf{x}_{\gamma} \leq \mathbf{c}^{u}. \nonumber 
    \end{align}

Now, let $\mathbf{W}$ be the diagonal matrix with the inverse variance of the differentially private random variables for each of the children of node $\gamma$ along its diagonal. While choices of $\mathbf{W}$ other than inverse variance were explored (e.g., constant, square-root of variance), both theoretically (according to standard weighted least squares theory) and in practice, the inverse variance provided the best performance overall. Let $\mathbf{I}_{c^*}$ represent a identity matrix of size $c^*$. Then the least squares estimator for the joint solution for the children of $\gamma$ is

\begin{align} \label{l2_opt_2} 
    \tilde{\mathbf{x}}_{\text{children}(\gamma)} \gets & \argmin_{\mathbf{x}_{\text{children}(\gamma)}} (\mathbf{Q} \mathbf{x}_{\text{children}(\gamma)} - \widetilde{\mathbf{M}}_{\text{children}(\gamma)})^\top \mathbf{W} (\mathbf{Q} \mathbf{x}_{\text{children}(\gamma)} - \widetilde{\mathbf{M}}_{\text{children}(\gamma)}) \\
    & \text{subject to:} \nonumber\\
    & \mathbf{x}_{\gamma} \ge \mathbf{0}  \text{; and} \nonumber \\
    & \mathbf{C}^{eq} \mathbf{x}_{\text{children}(\gamma)} = \mathbf{c}^{eq} \text{; and} \nonumber\\
    & \mathbf{C}^{u} \mathbf{x}_{\text{children}(\gamma)} \leq \mathbf{c}^{u} \text{; and} \nonumber\\
    & \begin{bmatrix} 
        \mathbf{I}_{c^*} & \mathbf{I}_{c^*} & \cdots & \mathbf{I}_{c^*}  \\
      \end{bmatrix} \mathbf{x}_{\text{children}(\gamma)} = \hat{\mathbf{x}}_{\gamma}. \nonumber
    \end{align}

In (\ref{l2_opt_1}) and (\ref{l2_opt_2}) we use generic notation for $\mathbf{Q}, \mathbf{C}^{eq}, \mathbf{c}^{eq}, \mathbf{C}^{u}$, and $ \mathbf{c}^{u}$ rather than node specific subscripts in an effort to simplify notation, though it should be noted that the values of $\mathbf{Q}, \mathbf{C}^{eq}, \mathbf{c}^{eq}, \mathbf{C}^{u}$, and $ \mathbf{c}^{u}$ will be different depending on the specific child  or parent node.     

\subsubsection{The basic rounder estimator} \label{subsec:basic_l1}
Similar to the least squares optimization problem, we specify two versions of the solution depending on whether we are estimating over a single or multiple nodes. Let $\floor*{\mathbf{x}}$ be the floor function (applied element-wise to a vector).  Then the rounder estimator for a single node $\gamma$ is    

\begin{align}
     \hat{\mathbf{x}}_{\gamma} \gets & \floor*{\tilde{\mathbf{x}}_{\gamma}} + \mathbf{\hat{y}} \\
     & \mathbf{\hat{y}} = \argmin_{\mathbf{y}} \mathbf{1}^\top | \tilde{\mathbf{x}}_{\gamma} - (\floor*{\tilde{\mathbf{x}}_{\gamma}} + \mathbf{y}) | \nonumber \\
     & \text{subject to:} \nonumber\\
     & y_i \in \{0,1\} \text{ for } y_i \text{ elements of } \mathbf{y} \text{; and} \nonumber \\
     & \mathbf{C}^{eq} \left( \floor*{\tilde{\mathbf{x}}_{\gamma}} + \mathbf{y} \right) = \mathbf{c}^{eq} \text{; and} \nonumber \\
     & \mathbf{C}^{u} \left( \floor*{\tilde{\mathbf{x}}_{\gamma}} + \mathbf{y} \right) \leq \mathbf{c}^{u}. \nonumber 
\end{align}

\noindent Let $\mathbf{I}_{c^*}$ represent a diagonal matrix of size $c^*$. Then the rounder estimator for the joint solution for the children of $\gamma$ is

\begin{align}
     \hat{\mathbf{x}}_{\text{children}(\gamma)} \gets & \floor*{\tilde{\mathbf{x}}_{\text{children}(\gamma)}} + \mathbf{\hat{y}} \\
     & \mathbf{\hat{y}} = \argmin_{\mathbf{y}}  \mathbf{1}^\top|\tilde{\mathbf{x}}_{\text{children}(\gamma)} - (\floor*{\tilde{\mathbf{x}}_{\text{children}(\gamma)}} + \mathbf{y}) | \nonumber \\
     & \text{subject to:} \nonumber\\
     & y_i \in \{0,1\} \text{ for } y_i \text{ elements of } \mathbf{y}  \text{; and} \nonumber \\
     & \mathbf{C}^{eq} \left( \floor*{\tilde{\mathbf{x}}_{\text{children}(\gamma)}} + \mathbf{y} \right) = \mathbf{c}^{eq} \text{; and} \nonumber \\
     & \mathbf{C}^{u} \left( \floor*{\tilde{\mathbf{x}}_{\text{children}(\gamma)}} + \mathbf{y} \right) \leq \mathbf{c}^{u}. \nonumber \\
     & \begin{bmatrix} 
        \mathbf{I}_{c^*} & \mathbf{I}_{c^*} & \cdots & \mathbf{I}_{c^*}  \\
      \end{bmatrix} \left( \floor*{\tilde{\mathbf{x}}_{\text{children}(\gamma)}} + \mathbf{y} \right) = \hat{\mathbf{x}}_{\gamma}. \nonumber
\end{align}
\noindent DAS implements the rounder problem by subtracting the floor and minimizing the distance over the fractional part of the target vector.
This problem is feasible and efficiently solvable if it has the total unimodularity (TUM) property (see Section \ref{subsec:l2_multi_pass}). 


\subsection{Algorithmic improvements: Multi-pass least squares estimator}
\label{subsec:l2_multi_pass}

While we found that the basic least squares estimator worked well overall, we observed a concerning trend in the bias of the estimation of geographic totals. Even when these queries were asked with very high precision, when combined in an optimization with a query with a large number of components (e.g., the detailed query), the resulting estimated data vector would be less accurate on the target query than would be expected. We know that this problem is due to the constraint that all cell counts must be non-negative. We now understand that there is an inherent \emph{uncertainty principle} limiting the accuracy of inequality-constrained least squares. The maximum expected error in either the detailed query or its total depends upon both the privacy-loss budget allocation and the total number of queries \citep{anon:2021}.\footnote{It is now understood that this problem is a direct consequence of requiring TDA to produce microdata. It is not due to the use of DP mechanisms \citep{anon:2021}. The result is related to the implementation difficulties studied in \citet{Balcer:Vadhan:2019}.} Managing this accuracy trade-off between detailed and total queries, which is due entirely to the post-processing requirement to produce privacy-protected microdata, motivated the algorithmic improvements discussed here.

To improve on the single-pass non-linear least squares estimator, we developed a \emph{multi-pass} estimator. Instead of attempting to find a data vector solution that minimizes errors for all noisy queries at once, we use multiple passes of the optimizer to estimate increasingly more detailed tabulations of the data vector. On each pass, after estimating one or more margins of the contingency-table vector, we constrain the following passes to find a solution that maintains the previously estimated margins. For example, for the redistricting data persons schema, we might first estimate and constrain the total population counts in a first pass. Then, estimate the detailed query in a second pass subject to the constraint that its margin equals the total population count from the first pass. Multi-pass processing, therefore, controls the algorithmic error due to the non-negativity constraints by prioritizing the queries based on the expected size of the total query (total population before voting-age population, etc.).

The input of the least squares multi-pass methods described here is a collection of marginal queries and, for each such query group, the vector of DP random measurements for the query group. Let $k\in \{1, \dots, K\}$ represent the passes such that $\tilde{\mathbf{x}}_{\text{children}(\gamma)}^{(k)}$ is the estimated data vector after pass $k$, $\mathbf{Q}^{(k)}$ is the stacked set of query groups used for pass $k$, and $\widetilde{\mathbf{M}}_{\text{children}(\gamma)}^{(k)}$ is the corresponding DP random measurements. Let $\mathbf{B}^{(k)}$ be the matrix formed by stacking 
query group matrices 
used for pass $k$; 
it
will be used for constraints in passes greater than $k$. For brevity, we omit the single node version of the estimator and focus on the joint solution for the children of $\gamma$ which is given by  $\tilde{\mathbf{x}}_{\text{children}(\gamma)}^{(K)}$, where the optimization routine is given by (\ref{l2_opt_multipass}).

\begin{align} \label{l2_opt_multipass} 
    \tilde{\mathbf{x}}_{\text{children}(\gamma)}^{(k)} \gets
    & \argmin_{\mathbf{x}_{\text{children}(\gamma)}} (\mathbf{Q}^{(k)} \mathbf{x}_{\text{children}(\gamma)} - \widetilde{\mathbf{M}}_{\text{children}(\gamma)}^{(k)})^\top \mathbf{W} (\mathbf{Q}^{(k)} \mathbf{x}_{\text{children}(\gamma)} - \widetilde{\mathbf{M}}_{\text{children}(\gamma)}^{(k)})  \\
    & \text{subject to:} \nonumber\\
    & \mathbf{x}_{\text{children}(\gamma)} \ge \mathbf{0}  \text{; and} \nonumber \\
    & \mathbf{C}^{eq} \mathbf{x}_{\text{children}(\gamma)} = \mathbf{c}^{eq} \text{; and} \nonumber\\
    & \mathbf{C}^{u} \mathbf{x}_{\text{children}(\gamma)} \leq \mathbf{c}^{u} \text{; and} \nonumber\\
    & \begin{bmatrix} 
        \mathbf{I}_{c^*} & \mathbf{I}_{c^*} & \cdots & \mathbf{I}_{c^*}  \\
      \end{bmatrix} \mathbf{x}_{\text{children}(\gamma)} = \hat{\mathbf{x}}_{\gamma}  \nonumber \\
    & \text{If } k > 1 :  \nonumber \\
    & \; \; \; \lvert \mathbf{B}^{(j)} \left( \mathbf{x}_{\text{children}(\gamma)} - \tilde{\mathbf{x}}_{\text{children}(\gamma)}^{(j)} \right) \rvert \le  \hat{\tau}^{(j)} \mathbf{1}\; \forall \; j \in \{1, \dots , k-1\}  \nonumber
    \end{align}
The parameter $\hat{\tau}^{(j)}$ represents a small positive constant to help ensure feasibility of future solutions. Rather than specifying this number, we solve the following secondary optimization problem to estimate it at the end of each pass $k$:

\begin{align} \label{l2_opt_multipass_tol}
    \hat{\tau}^{(k)} &\gets \min_{\tau, \mathbf{x}_{\text{children}(\gamma)}} \tau  \\
    & \text{subject to:} \nonumber\\
    & \lvert \mathbf{B}^{(j)} \left( \mathbf{x}_{\text{children}(\gamma)} - \tilde{\mathbf{x}}_{\text{children}(\gamma)}^{(j)} \right) \rvert \le \tau \mathbf{1} \; \forall \; j \in\{1, \cdots , k\}  \nonumber \\
    & \mathbf{x}_{\text{children}(\gamma)} \ge \mathbf{0}  \text{; and} \nonumber \\
    & \mathbf{C}^{eq} \mathbf{x}_{\text{children}(\gamma)} = \mathbf{c}^{eq} \text{; and} \nonumber\\
    & \mathbf{C}^{u} \mathbf{x}_{\text{children}(\gamma)} \leq \mathbf{c}^{u} \text{; and} \nonumber\\
    & \begin{bmatrix} 
        \mathbf{I}_{c^*} & \mathbf{I}_{c^*} & \cdots & \mathbf{I}_{c^*}  \\
      \end{bmatrix} \mathbf{x}_{\text{children}(\gamma)} = \hat{\mathbf{x}}_{\gamma}.  \nonumber 
    \end{align}

Note that the basic least squares estimator can be recovered by choosing $K=1$ and including all query groups for which we have DP random measurements. We have found it generally advantageous to first estimate query groups with elements that are unlikely to be zero using an estimator with negligible bias. Constraining the estimates in subsequent passes to have marginals that are consistent with these earlier query group estimates has the additional effect of reducing the positive bias on cell counts with true values that are zero. Our current estimation strategy limits the positive bias of query group estimates in early passes by omitting query groups from the objective function that are likely to have many true answers that are zero. For example, choosing the first pass to only include query groups that are higher-order marginals (e.g. total, voting-age). However, note that by excluding DP random measurements, we are not including all available information. Note that the choice of query groups does not depend on the confidential data.

\subsection{Algorithmic improvements: Multi-query multi-pass rounder}
\label{subsec:l1_multi_pass}

The basic rounder estimator finds a nearby integer solution by minimizing the distance from the real-valued solution with respect to the individual cell values while still respecting constraints. We found that, as the size of the optimization problems grew, even though in the rounder solution individual cells were forced to be close to their real-valued counterparts, the summation of many cells (e.g., a variable margin) could still change by a large amount. Changes to the individual detailed cell values could add up to make a large change in, for example, the voting-age population. To improve this behavior, we developed an estimation routine that optimizes based on the distance from the real-valued solutions to the queries used in the least squares estimator. Like the least squares multi-pass estimator, we also allowed for multiple passes of the rounder at increasing levels of detail.  

Because the non-negative least squares estimator is solved numerically, the amount of mass (accumulated fractional parts of the least squares solution) the rounder was responsible for estimating increased as privacy-loss budgets increased. That is, as we increased the privacy-loss budget, we wanted to ensure monotone convergence to perfect accuracy. However, when using the original, basic single-pass rounder, as the privacy-loss budget increased many more cell estimates from the least squares problem were 0.999 (not 1.0). As this happened, because the rounder is responsible for estimating a total amount of mass that it is equal to the sum of the fractional parts of the least squares estimator, the rounder became increasingly important to solution quality the larger the privacy-loss budget. Rounding using only the detailed query answers is clearly an inferior estimator of most queries of interest compared to using all queries from the least squares solution. Hence, basic rounding caused very troubling behavior: specifically, an intermediate region in which increased privacy-loss budget resulted in poorer performance on utility metrics. When we implemented multi-pass rounding, we also added the ability to specify additional queries as part of the rounder's objective function. TUM implies restrictions on these extra queries. Without violating TUM, we inserted a single additional set of hierarchically related queries (a `tree of queries') in this way. This modification of the objective function solved the non-monotonicity problem.


Let $k\in \{1, \dots, K\}$ represent the passes such that $\hat{\mathbf{x}}_{\text{children}(\gamma)}^{(k)}$ is the estimated data vector after pass $k$,\footnote{Note that the number of passes for the rounder can be different from the number of passes for least squares.} and $\mathbf{Q}^{(k)}$ is the stacked set of query groups used for pass $k$ (which can be different from the query groups used in the  multi-pass least squares estimator).
Meanwhile, $\tilde{\mathbf{x}}_{\text{children}(\gamma)}$ is the vector produced by the multi-pass least squares estimator.
Let $\mathbf{B}^{(k)}$ be the matrix for the stacked query groups that will be the used as constraints in passes greater than $k$. For brevity, we omit the single node version of the estimator, and focus on the joint solution for the children of $\gamma$. The solution is denoted by $\hat{\mathbf{x}}_{\text{children}(\gamma)}^{(K)}$, the terminal value in the following optimization letting $k=1, \dots, K$: 

\begin{align} \label{l1_multiplass}
     \hat{\mathbf{x}}_{\text{children}(\gamma)}^{(k)} \gets & \floor*{\tilde{\mathbf{x}}_{\text{children}(\gamma)}} + \mathbf{\hat{y}}^{(k)} \\ \nonumber
     & \mathbf{\hat{y}}^{(k)} = \argmin_{\mathbf{y}} \mathbf{1}^\top | \mathbf{Q}^{(k)}  
     \left( \tilde{\mathbf{x}}_{\text{children}(\gamma)} - \left(
     \floor*{\tilde{\mathbf{x}}_{\text{children}(\gamma)}}
     + \mathbf{y} \right)  \right) | \nonumber \\
     & \text{subject to:} \nonumber\\
     & y_i \in \{0,1\} \text{ for } y_i \text{ elements of } \mathbf{y}  \text{; and} \nonumber \\
     & \mathbf{C}^{eq} \left( \floor*{\tilde{\mathbf{x}}_{\text{children}(\gamma)}} + \mathbf{y} \right) = \mathbf{c}^{eq} \text{; and} \nonumber \\
     & \mathbf{C}^{u} \left( \floor*{\tilde{\mathbf{x}}_{\text{children}(\gamma)}} + \mathbf{y} \right) \leq \mathbf{c}^{u}. \nonumber \\
     & \begin{bmatrix} 
        \mathbf{I}_{c^*} & \mathbf{I}_{c^*} & \cdots & \mathbf{I}_{c^*}  \\
      \end{bmatrix} \left( \floor*{\tilde{\mathbf{x}}_{\text{children}(\gamma)}} + \mathbf{y} \right) = \hat{\mathbf{x}}_{\gamma} \nonumber \\
    & \text{If } k > 1 :  \nonumber \\
    & \; \; \; \mathbf{Q}^{(j)} \left( \floor*{\tilde{\mathbf{x}}_{\text{children}(\gamma)}} + \mathbf{y} \right) = \mathbf{Q}^{(j)} \hat{\mathbf{x}}_{\text{children}(\gamma)}^{(j)} \; \forall \; j \in \{1, \dots , k-1\}. \nonumber 
\end{align}

We note that the existence of a solution for the multi-pass rounder is not always guaranteed and depends heavily on the choice of the query groups used in each pass and their ordering. In particular, to guarantee polynomial time solution of the rounder estimator, the mathematical constraints must exhibit total unimodularity (TUM), the condition that every square non-singular submatrix is unimodular (integer and has determinant +1 or -1). TUM does two things. First, if a continuous solution exists to the non-negative least squares problem problem, then TUM implies than an integer solution exists to the rounder problem. Second, TUM implies that finding the rounder solution is tractable. TUM does not, unfortunately, guarantee that the continuous non-negative least squares problem has a solution in the first place. Because the mathematical constraints in the rounder estimator are composed of interactions between the invariants, edit constraints, and hierarchical consistency constraints, ensuring TUM requires limiting invariants and edit constraints as much as possible.

\subsection{Algorithmic improvements: AIAN spine}
\label{subsec:aian_spine}
American Indian and Alaska Native (AIAN) areas tend to be far off the standard tabulation spine and, therefore, are at risk of poor accuracy for a given privacy-loss budget.  For example, the largest AIAN tribal area is the Navajo nation which includes subsets of many counties within the states of New Mexico, Arizona and Utah. To improve the accuracy of the AIAN areas, we define what we call the \emph{AIAN spine}, which  subdivides each state with any AIAN tribal areas into two geographic entities: the union of the legally defined AIAN areas within the respective state and the remaining portion of the state.\footnote{The AIAN spine was developed in collaboration with geographers who specialize in the different tribal area definitions \citep{AIANspine}. We used all Census AIAN and native Hawaiian areas except state-designated tribal statistical areas (SDTSAs) and tribal-designated statistical areas (TDSAs). } Counties, tracts, and block groups are also subdivided as necessary to be fully nested within either the AIAN branch or the non-AIAN branch of the state's hierarchy. Note that AIAN areas, like all tabulation geographic entities, are composed of tabulation blocks; therefore the creation of this spine did not require the census blocks to be subdivided.

In the 2010 Census geography definitions, there were 36 states with one or more AIAN area on the spine. That number increased to 37 in 2020. Functionally, these new geographic units are treated as state equivalents in the TDA.  Meaning there are 87 (88 in 2020) state-equivalent geographies within the US used in TDA.

Creating the AIAN spine created a subtle, but important, modification to the way the least squares and rounder problems were solved. In states with AIAN areas on the spine, the TDA hierarchy branches at the state level. This means that at the state level, the AIAN total population and the non-AIAN total population are solved with privacy-loss budget. Only their sum is invariant. For states without AIAN areas on the spine total population is invariant and there is no need to estimate AIAN and non-AIAN components.


\subsection{Algorithmic improvements: Optimized spine}
\label{subsec:optspine}
\emph{Off-spine} entities are those geographic entities that are not directly estimated as part of the TDA spine. Using the standard census tabulation spine (US, state, county, tract, block group, block), many legally-defined off-spine entities are far away from the tabulation spine. The relevant distance measure is the \emph{minimum} number of on-spine entities added to or subtracted from one another in order to define the off-spine entity, which we call the \emph{off-spine distance}. There is an inverse relation between off-spine distance and count accuracy for off-spine entities. As an entity's off-spine distance increases, the accuracy of its estimated characteristics decreases due to the accumulation of independent random noise from the DP mechanism applied to each of these components. Our approach to optimizing the geographic spine is discussed in detail in \citet{AIANspine}. In this section, we give the provide the main details.     

In order to improve the accuracy of legally-defined off-spine entities, we designed an algorithm to re-configure the standard tabulation spine for use within the TDA. This was done in two ways.  First, we created \emph{custom block groups}, which are re-groupings of the tabulation blocks such that off-spine entities are closer to the spine and group quarters facilities of each type are isolated from their surrounding housing-unit only areas. Isolating group quarters in this manner improves the accuracy of estimated characteristics in both the group quarters and surrounding housing-unit only blocks. Second, we bypass parent geographies with only a single child, combining the privacy-loss budget for the child with that of the parent, which avoids wasting privacy-loss budget on redundant measurements and yields more accurate estimates. The accuracy gain occurs because the algorithm proceeds in a top-down fashion; hence, in the case of a single child geography all measurements must exactly equal those of the parent. Additional privacy-loss budget used on the child after the parental estimation is wasted. This optimization step effectively collapses the parent and child into a single node when the parent's fan-out is one and uses the combined privacy-loss budget of the parent and child to take DP measurements once for this node.  

We applied spine optimization after defining the AIAN spine. Hence, the AIAN areas and the balance of the state were optimized in separate branches of the hierarchy. To implement spine optimization, the Geography Division provided a list of the major sub-county governmental entities within each state. For portions of strong Minor Civil Division (MCD) states outside of AIAN areas, these entities were MCDs. In all other areas, the major sub-county entities were a combination of incorporated and Census-designated places. As with states, counties and tracts using the AIAN spine, these sub-state entities were divided into the part within an AIAN area and the balance of the sub-state entity. Call these sub-state entities the targeted off-spine entities. We minimized the off-spine distance for the targeted off-spine entities using a greedy algorithm that created custom block groups to reduce the number of on-spine entities required to build the targeted off-spine entities. The custom block groups replaced tabulation block groups within TDA; however, the official tabulation block groups are used for all reported statistics. That is, custom block groups are an implementation detail for TDA, not a replacement for tabulation block groups. The spine optimization algorithm could also have created custom tract groups, but this feature was not implemented for the redistricting data.

Table \ref{table:geo_entity_counts} shows selected statistics for geographic entities at each level of the heirarchy after applying the spine optimization to the AIAN spine respectively for 2010 and 2020. The number of entities for state, county, tract and block group differ from the number of tabulation tabulation geographic entities because the AIAN spine subdivides geographies with AIAN areas. Effectively, this moves tabulation counties, tracts and block groups slightly off spine in order to move AIAN areas, MCDs, and census places closer to the spine.  


\begin{table}[H]
\caption{Selected Statistics for Geographic Entity Counts over the Hierarchy}
\begin{tabular}{|p{4.2cm}| p{4.2cm} | p{4.2cm} |}
\hline
 & 2010 & 2020 \\
\hline
States with AIAN areas on the Spine (FIPS Codes)* & 01, 02, 04, 06, 08, 09, 12, 13, 15, 16, 19, 20, 22, 23, 25, 26, 27, 28,   30, 31, 32, 35, 36, 37, 38, 40, 41, 44, 45, 46, 48, 49, 51, 53, 55, 56 & 01, 02, 04, 06, 08, 09, 12, 13, 15, 16, 18, 19, 20, 22, 23, 25, 26, 27,   28, 30, 31, 32, 35, 36, 37, 38, 40, 41, 45, 46, 47, 48, 49, 51, 53, 55, 56 \\
\hline
Entities State Level  & 87       & 88      \\
\hline
Entities County Level & 3,488 & 3,496 \\
\hline
Entities Tract Level & 73,180 & 84,589 \\
\hline
Tabulation Block Groups & 217,740 & 239,780 \\
\hline
Custom Block Groups** & 402,020  & 409,548  \\
\hline
Tabulation Blocks & 11,078,297 & 8,132,968 \\
\hline
Tabulation Blocks with Occupied GQs & 112,956 & 126,723 \\
\hline
Tabulation Blocks with Housing Units \textgreater 0 & 6,379,963 & 5,879,049 \\
\hline
Tabulation Blocks with Potential Positive Population & 6,398,202 & 5,892,698 \\
\hline
\multicolumn{3}{|p{12.6cm}|}{\footnotesize{*Rhode Island has an AIAN area in 2020; however, it is not included in the list of states with AIAN areas on the spine because there are no housing units or occupied GQs in this AIAN area. Therefore, the blocks within this AIAN area are not included in the spine at all because they have a zero probability, a priori, of containing positive resident population.  

**Custom block groups (CBG) used within the TDA differ from tabulation block groups. These differences improve accuracy for off-spine geographies like places and minor civil divisions. The use of CBG for measurement and post-processing within TDA does not impact how the resulting data are tabulated. All 2020 Census data products will be tabulated using the official tabulation block groups as defined by the Census Bureau's Geography Division. }}\\
\hline
\end{tabular}
\label{table:geo_entity_counts}
\end{table}
{}

\subsection{Computing and technology requirements of the DAS}
\label{subsec:implementation_details}
The DAS requires a specialized environment to run given the size and complexity of the data being processed and the computing requirements of the algorithm. Here we describe some of the specifics of those computing and technology requirements. In Section \ref{sec:utility_experiments} we discuss the final production settings used in the TDA for the redistricting data. The DAS, as implemented for the redistricting data, is written in Python version 3.7 and uses Gurobi 9.1 for numerical optimization. The system is executed on Amazon Web Services (AWS) using Elastic Map Reduce (EMR) version 6.2 and Apache Spark version 3.0.1. The DAS is run on EMR clusters built from AWS r5.24xlarge virtual machines, in which each VM has 96 cores and 768GiB of RAM. The typical cluster size and spark configuration is: 1 EMR \emph{master} node and between 8 and 30 \emph{worker} (core plus task) nodes, with Spark deployed in client mode. The mix of Spark executors, memory settings, and cores-per-executor used by TDA varies based on workload, but for a typical redistricting data run, was often set to 24 cores per executor, 256 GiB driver memory, 64 GiB executor memory, and 128 GiB memory overhead.

The TDA implementation within the DAS consists of two primary parts: an underlying framework and the specialization of the framework for the implementation in 2020 Census.

The \textit{DAS framework} is a set of 22 Python source files (2,814 lines of code and 1,125 lines of unit tests) that provides a generic framework for implementing a batch data processing system. The framework provides for explicit modules that read a configuration file and then use the information in that file to initialize the remainder of the application,  read the input data (the \emph{reader}), perform the differential privacy computation (the \emph{engine}), validate the results of the computation (the \emph{validator}), and write out the results (the \emph{writer}). The actual reader, engine, validator, and writer are implemented by their own Python classes that are also specified in the configuration file. Each module can in turn read additional parameters from the configuration file. Thus, the behavior of the entire system is determined by the configuration file, making it easy to use the framework for both algorithmic experimentation as well as production by specifying different configuration files.

The TDA for the 2020 Census is written within the DAS Framework. It consists of 461 Python source files (78,428 lines of program files and 11,513 lines of unit tests). The system includes 10 readers, 20 writers, and multiple engines. Only one reader, writer, and engine are used for any given run. To document the provenance for each TDA run, the DAS framework automatically determines from the configuration file the specific Python files required and saves this information in XML and in Portable Document Format (PDF) in the \textit{DAS Certificate}. The system also creates an archive containing the complete source code run. All TDA run data and metadata are archived in the AWS Simple Storage Service (S3) with vintage tags.

All DP mechanisms implemented in TDA require a source of random numbers. We estimate that the 2020 Census will require roughly 90TB of random bytes to protect the person and housing unit tables. We generate random bits in the AWS environment using the Intel RDRAND instruction mixed with bits from/dev/urandom \citep{garfinkel:leclerc:2020}. Because we used an implementation of the  \citet{NEURIPS2020_b53b3a3d,canonne2020discrete} exact sampling approach for the discrete Gaussian mechanism, and allocated privacy-loss budgets using rational numbers, our implementation avoids the vulnerabilities often associated with floating-point random numbers as noted by \citet{mironov:2012}. 


\section{Tuning and Testing the DAS Over Time}
\label{sec:utility_experiments}

\subsection{Tuning the DAS and the release of demonstration data products}
The 2010 Census data represent the best source of real data for conducting utility experiments relevant to the 2020 census without using the 2020 data. Starting in October 2019, the Census Bureau released a series of demonstration data products based on the 2010 data \citep{2020_das_development}.  Each set consisted of privacy-protected microdata files (PPMF) generated from the microdata detail file (MDF) that supported at least the redistricting data tables. When released, each demonstration data product represented the current state of DAS development. 

It is not possible to tune a formally private disclosure avoidance system by directly using the confidential data it is designed to protect. Doing so causes the values of the tuning parameters---the privacy-loss budget allocations and the query strategies---to leak information about the confidential data. This is the reason that the TDA was not tuned using the 2020 Census Edited File. It is also the reason why quality assurance of the MDF is a difficult conceptual task. Historically, the disclosure avoidance methods used for the decennial census were tuned using data from the same census. Record swaps were rejected if the resulting swapped data did not meet particular accuracy standards. For historical reasons and because of errors associated with the microdata requirement, quality assurance for the 2020 DAS MDF does not strictly adhere to the tenets of formal privacy. The MDF receives human review examining it for errors that look unusual, surprising, or unacceptable compared to the 2020 Census Edited File. Notification of such errors may be escalated up the organizational hierarchy for review and decision-making. As noted in Section \ref{sec:conclusion}, there were no such errors flagged in the review of the production redistricting data; however, this situation highlights the need for formally private techniques for quality assurance.

The demonstration data product dated April 28, 2021 was the final pre-production release. This PPMF consisted of person- and housing unit-level microdata with two global privacy-loss budgets: $\rho=1.095$ ($1.05$ for persons and $0.045$ for housing units) or $\epsilon = 11.14$ and $\rho=0.1885$ ($0.185$ for persons and $0.0035$ for housing units) or $\epsilon = 4.36$.\footnote{The October 2019 and May 2020 demonstration data products were based on the full suite of tables planned for the redistricting, demographic and housing characteristics data. They were implemented using versions of TDA based on the discrete geometric mechanism with pure DP ($\delta=0$). In these releases, the global privacy-loss budget was determined using $\epsilon$-allocation---4 for persons, 2 for housing units---which sums to a total of 6. The September 2020, November 2020, April 2021, and production versions of TDA were implemented using versions of TDA based on the discrete Gaussian mechanism. For these releases, including the production release, privacy-loss budget was determined using $\rho$-allocation in the zCDP framework. When $\rho$ allocations are re-expressed in terms of $\epsilon$, the conversion uses the numerical approximation discussed in Section \ref{sec:mechanism:overview} evaluated a $\delta=10^{-10}$. In public summaries, we added $\epsilon$ for persons and housing units to facilitate comparisons to earlier releases. When using zCDP, a more accurate conversion of the global privacy-loss budget to $(\epsilon,\delta)$ form can be obtained by summing the $\rho$-allocation of persons and housing units, then converting to $\epsilon$ at $\delta=10^{-10}$. For the earlier releases, global privacy-loss budgets were $\rho = 0.1885$ (November 2020), $\rho = 0.1885$ (September 2020), $\epsilon =6$ (May 2020), and $\epsilon =6$ (October 2019). As noted in Section \ref{sec:mechanism:overview}, zCDP is summarized by the value of $\rho$, which represents a continuum of $(\epsilon,\delta)$ pairs.}

Table \ref{table:rho_allocation_geography} gives the proportion of the $\rho$-budget per geographic level for both persons and housing units in the April 2021 release. For the estimation of the persons tables, the block-level queries received approximately half of the total privacy-loss budget. For housing units, block-group-level queries received the predominant share of the privacy-loss budget.        

\begin{table}[H]
\caption{Per Geographic Level $\rho$ Allocation Proportions for Persons and Housing Units, April 2021 Release}
\begin{tabular}{|l|r|r|}
\hline
Geographic Level & Person $\rho$ Proportions & Housing Units $\rho$ Proportions \\
\hline
US                     & 51/1024 & 1/1024                            \\
State                  & 153/1024  & 1/1024                          \\
County                 & 78/1024     & 18/1024                       \\
Tract                  & 51/1024    & 75/1024                      \\
Custom Block Group*    & 172/1024   & 906/1024                      \\
Block                  & 519/1024 & 23/1024         \\
\hline
\multicolumn{3}{|p{12cm}|}{\footnotesize{*The custom block groups used within TDA differ from tabulation block groups.}}\\
\hline
\end{tabular} \label{table:rho_allocation_geography}
\end{table}

Tables \ref{table:rho_allocation_persons} and \ref{table:rho_allocation_household} show the proportions of  the total $\rho$-budget allocations per query within each geographic level for the persons and households tables respectively. These should be interpreted in conjunction with Table \ref{table:rho_allocation_geography} in order to calculate the proportion of the total budget which a level-specific query represents.  For example, the County TOTAL query represents $342/1024 \times 78/1024$, or approximately $2.5\%$ of the total privacy-loss budget for persons.\footnote{In this section, the use of query names in all capitals, like TOTAL, refers to the definitions in Tables \ref{table:rho_allocation_persons} and \ref{table:rho_allocation_household}. DETAILED means HHGQ $\times$ VOTINGAGE $\times$ HISPANIC $\times$ CENRACE.}

The set of queries and the values used per geographic level and per query were determined by a set of internal experiments that took place over a period of approximately five months, beginning in December 2020 and ending in early April 2021. These internal experiments estimated the minimum global privacy-loss budget and its allocation to specific queries in order to ensure fitness-for-use in the redistricting application. See \citet{wright:irimata:2021} for details. These allocations were reflected in the April 2021 demonstration data product. Subsequent analysis by internal and external stakeholders addressed accuracy criteria for other, primarily demographic, uses. These allocations are reflected in the final production settings.   

The initial, systematic tuning phase specified a quantitative metric targeted at the redistricting use case: for at least 95\% of geographic units, the largest of a selected set of CENRACE $\times$ HISPANIC populations, corresponding to the major race and ethnicity groups in Statistical Policy Directive 15, should not change by more than $\pm 5$ percentage points when expressed as a percentage of the geographic unit's total population compared to its enumerated percentage of the total population in that geographic unit. Due to computational constraints of creating many independent samples, geographic variability was used for evaluating the 95\% criterion. While not a direct substitute for evaluating variability across many runs, we have observed that these metrics are very consistent from run to run, making them useful in their own right.  
 

The TDA has error distributions that are data-dependent in a fashion that cannot be expressed in closed form and depends on confidential data values that cannot be published. Therefore, we performed the fitting to the redistricting use-case iteratively rather than using the ex ante properties of zCDP. We used six candidate query strategies (queries for which formally private measurements are taken), created by defining two dimensions. First, query strategies varied by either using measurements that closely conformed to those in the publication specifications for the redistricting data product, or by collapsing/excluding certain marginal queries to try to take advantage of dependence among queries. Second, because both formal privacy and certain measures of error in TDA's estimation are especially sensitive to cells with small counts, we assigned budgets to queries in one of three ways: proportional to the square of query cardinality, inversely proportional to the square of query cardinality, or uniformly assigned. Crossing these two dimensions generated six candidate strategies. Optimization passes were ordered heuristically, with queries assigned large relative budgets typically isolated in their own passes. In the strategies with almost all privacy-loss budget assigned to the DETAILED query (proportional to the square of query cardinality), in particular, queries were optimized simultaneously, in a single pass.

Strategies that assigned weight according to the inverse-square of query cardinality were quickly ruled out as noncompetitive. For the remaining four strategies, we iterated experimentally, geographic level by geographic level starting at the top with the US. First, using binary search, for each geographic level, we identified the approximate assignment of $\rho$ to that geographic level's queries needed to satisfy the redistricting criterion, when the other geographic levels were assigned infinite budget. We then proceeded to account for dependence between geographic levels. We set US and state to their single-level fitted $\rho$ values, and assigned the remaining levels infinite budget. We then increased the US and state $\rho$ budgets proportionally until the criterion was satisfied. Next, we set the county budget to its single-level fitted $\rho$, and further increased US, state, and county budgets proportionally until the criterion was satisfied. We proceeded down the hierarchy until all geographic levels had been assigned a $\rho$ value. This procedure led to the initial identification of a total $\rho$ of $1.095$ as satisfying the redistricting criterion. While it may not be obvious from this description, because the redistricting criterion requires statistics for off-spine entities, every experimental run of TDA processed data down to the block level. This is the reason tuning took several months to accomplish.

The April 2021 demonstration data product used the $L_2$ and $L_1$ multi-pass components. Specifically, at the US level, we used a single pass for both the $L_2$ and $L_1$ phases with all the queries in Table \ref{table:rho_allocation_persons} used for both phases. The state, county, tract, and block group geographic levels used two passes each.  For both the $L_2$ and $L_1$ phases, the first pass consisted of the TOTAL query used both as the noisy estimate and constraint. The second pass consisted of the remaining queries in  Table \ref{table:rho_allocation_persons} for both phases. Finally, at the block level, a single pass was used for both phases with all the queries from Table \ref{table:rho_allocation_persons} for both phases. The choice of the number of passes and the specific queries used in each pass were tested as part of the redistricting experiments, though we do not show any comparisons here with alternatives. 

\subsection{Redistricting data production settings} \label{subsec:prod_settings}

\begin{table}[H] 
\caption{Per Query $\rho$ Allocation Proportions by Geographic Level for Persons, April 2021 Release}
\begin{tabular}{|p{4.2cm} |p{1.25cm}p{1.5cm}p{1.5cm}p{1.25cm}p{1.5cm}p{1.5cm}|}
\hline
Query & \multicolumn{6}{c|}{Per Query $\rho$ Allocation Proportions by Geographic Level}\\
\hline
& US & State      & County   & Tract    & CBG* & Block    \\
\hline
TOTAL (1 cell)  &          & 678/1024** & 342/1024 & 1/1024   & 572/1024                 & 1/1024   \\
CENRACE (63 cells)         & 2/1024   & 1/1024     & 1/1024   & 2/1024   & 1/1024                   & 2/1024   \\
HISPANIC (2 cells)         & 1/1024   & 1/1024     & 1/1024   & 1/1024   & 1/1024                   & 1/1024   \\
VOTINGAGE (2 cells)        & 1/1024   & 1/1024     & 1/1024   & 1/1024   & 1/1024                   & 1/1024   \\
HHINSTLEVELS (3 cells)     & 1/1024   & 1/1024     & 1/1024   & 1/1024   & 1/1024                   & 1/1024   \\
HHGQ (8 cells)             & 1/1024   & 1/1024     & 1/1024   & 1/1024   & 1/1024                   & 1/1024   \\
HISPANIC$\times$CENRACE (126 cells)     & 5/1024   & 2/1024     & 3/1024   & 5/1024   & 3/1024                   & 5/1024   \\
VOTINGAGE$\times$CENRACE (126 cells)    & 5/1024   & 2/1024     & 3/1024   & 5/1024   & 3/1024                   & 5/1024   \\
VOTINGAGE$\times$HISPANIC (4 cells)     & 1/1024   & 1/1024     & 1/1024   & 1/1024   & 1/1024                   & 1/1024   \\
VOTINGAGE$\times$HISPANIC $\times$CENRACE (252 cells)    & 17/1024  & 6/1024     & 11/1024  & 17/1024  & 8/1024                   & 17/1024  \\
HHGQ$\times$VOTINGAGE $\times$HISPANIC $\times$CENRACE (2,016 cells) & 990/1024 & 330/1024   & 659/1024 & 989/1024 & 432/1024                 & 989/1024 \\
\hline
\multicolumn{7}{|p{14.5cm}|}{\footnotesize{*Custom block groups (CBG) differ from tabulation block groups and are only used by the TDA.}}\\
\multicolumn{7}{|p{14.5cm}|}{\footnotesize{**The TOTAL query (total population) is held invariant at the state level. This $\rho$ allocation assigned to TOTAL at the state level is the amount assigned to the state-level queries for the total population of all AIAN tribal areas within the state and for the total population of the remainder of the state, for the 36 states that include AIAN tribal areas.}}\\
\hline
\end{tabular} \label{table:rho_allocation_persons}
\end{table}

\begin{table}[h!] 
\caption{Per Query $\rho$ Allocation Proportions by Geographic Level for Housing Units, April 2021 Release}
\begin{tabular}{|p{4.2cm} |p{1.25cm}p{1.5cm}p{1.5cm}p{1.25cm}p{1.5cm}p{1.5cm}|}
\hline
Query & \multicolumn{6}{l|}{Per Query $\rho$ Allocation Proportions by Geographic Level}\\
\hline
& US & State      & County   & Tract    & CBG* & Block    \\
\hline
Occupancy Status (2 cells) & 1 & 1 & 1 & 1 &1 & 1\\
\hline
\multicolumn{7}{|p{14.5cm}|}{\footnotesize{*Custom block groups (CBG) differ from tabulation block groups and are only used by the TDA.}}\\
\hline
\end{tabular} \label{table:rho_allocation_household}
\end{table}

For production run of the 2020 Census Redistricting Data (P.L. 94-171) Summary File, a total privacy-loss budget of $\rho = 2.63$ was used with $\rho=2.56$ used for person tables and $\rho=0.07$ for housing-units tables, considerable increases from the larger of the April 2021 settings. Production settings of the geographic-level $\rho$ proportions are given in Table \ref{table:rho_allocation_geography_production} and the query $\rho$ proportions by geographic level are given in Tables \ref{table:rho_allocation_persons_production} and \ref{table:rho_allocation_household_production} for persons and housing units, respectively. Note, in particular, the shift in the $\rho$ allocation for the person-level data away from the block level and onto the custom block group and tract level. This shift allowed the location confidentiality protection of respondents, which is primarily controlled by the allocation to the block level, to improve while meeting even tighter accuracy targets for HISPANIC $\times$ CENRACE at the tract and block-group level. This reallocation relied upon the efficiency gain from the spine optimization. The production settings used the same estimation strategy ($L_2$ and $L_1$ multi-pass components) and multi-pass orderings as the April 2021 release. 

\begin{table}[h!]
\caption{Per Geographic Level $\rho$ Allocation Proportions for Persons and Housing Units, Production Settings}
\begin{tabular}{|l|l|l|}
\hline
Geographic Level & Person $\rho$ Proportions & Housing Units $\rho$ Proportions \\
\hline
US                     & 104/4099 & 1/205                            \\
State                  & 1440/4099  & 1/205                          \\
County                 & 447/4099     & 7/82                       \\
Tract                  & 687/4099    & 364/1025                      \\
Custom Block Group* & 1256/4099   & 1759/4100                      \\
Block                  & 165/4099 & 99/820         \\
\hline
\multicolumn{3}{|p{12.5cm}|}{\footnotesize{*Custom block groups differ from tabulation block groups and are only used by the TDA.}}\\
\hline
\end{tabular} \label{table:rho_allocation_geography_production}
\end{table}

\begin{table}[H] 
\caption{Per Query $\rho$ Allocation Proportions by Geographic Level for Persons, Production Settings}
\begin{tabular}{|p{4.2cm} |p{1.25cm}p{1.5cm}p{1.5cm}p{1.5cm}p{1.5cm}p{1.5cm}|}
\hline
Query & \multicolumn{6}{c|}{Per Query $\rho$ Allocation Proportions by Geographic Level}\\
\hline
& US & State      & County   & Tract    & CBG* & Block    \\
\hline
TOTAL (1 cell)  &          & 3773/4097 & 3126/4097 & 1567/4102   & 1705/4099                 & 5/4097   \\
CENRACE (63 cells)        & 52/4097   & 6/4097     & 10/4097   & 4/2051   & 3/4099                   & 9/4097   \\
HISPANIC (2 cells)        & 26/4097   & 6/4097     & 10/4097   & 5/4102   & 3/4099                   & 5/4097   \\
VOTINGAGE (2 cells)        & 26/4097   & 6/4097     & 10/4097   & 5/4102   & 3/4099                   & 5/4097   \\
HHINSTLEVELS (3 cells)    & 26/4097   & 6/4097     & 10/4097   & 5/4102   & 3/4099                   & 5/4097   \\
HHGQ (8 cells)            & 26/4097   & 6/4097     & 10/4097   & 5/4102   & 3/4099                  & 5/4097   \\
HISPANIC$\times$CENRACE (126 cells)    & 130/4097   & 12/4097     & 28/4097   & 1933/4102   & 1055/4099                  & 21/4097  \\
VOTINGAGE$\times$CENRACE (126 cells)  & 130/4097   & 12/4097     & 28/4097   & 10/2051   & 9/4099                   & 21/4097   \\
VOTINGAGE$\times$HISPANIC (4 cells)   & 26/4097   & 6/4097     & 10/4097   & 5/4102   & 3/4099                   & 5/4097  \\
VOTINGAGE$\times$HISPANIC $\times$CENRACE (252 cells)  & 26/241  & 2/241     & 101/4097  & 67/4102  & 24/4099                   & 71/4097  \\
HHGQ$\times$VOTINGAGE $\times$HISPANIC$\times$CENRACE (2,016 cells) & 189/241 & 230/4097   & 754/4097 & 241/2051 & 1288/4099                & 3945/4097 \\
\hline
\multicolumn{7}{|p{14.5cm}|}{\footnotesize{*Custom block groups (CBG) differ from tabulation block groups and are only used by the TDA.}}\\
\hline
\end{tabular} \label{table:rho_allocation_persons_production}
\end{table}

\begin{table}[H] 
\caption{Per Query $\rho$ Allocation Proportions by Geographic Level for Housing Units, Production Settings}
\begin{tabular}{|p{4.2cm} |p{1.25cm}p{1.5cm}p{1.5cm}p{1.25cm}p{1.5cm}p{1.5cm}|}
\hline
Query & \multicolumn{6}{c|}{Per Query $\rho$ Allocation Proportions by Geographic Level}\\
\hline
& US & State      & County   & Tract    & CBG* & Block    \\
\hline
Occupancy Status (2 cells) & 1 & 1 & 1 & 1 &1 & 1\\
\hline
\multicolumn{7}{|p{14.5cm}|}{\footnotesize{*Custom block groups (CBG) differ from tabulation block groups and are only used by the TDA.}}\\
\hline
\end{tabular} \label{table:rho_allocation_household_production}
\end{table}

\subsection{Accuracy comparisons over time}

There are hundreds, if not thousands, of possible accuracy metrics that could be shown in order to assess the utility of the TDA for the redistricting and general demographic use cases. Here, we show just a few. The complete set of detailed statistical metrics for the redistricting use case can be found in \citet{wright:irimata:2021}. We summarize here. First, for Census-designated places and target off-spine entities, including off-spine entities in both the AIAN and non-AIAN branches of the hierarchy, the production settings of TDA achieve the $\pm 5$ percentage point target at least 95\% of the time in the population interval 200 to 249 (far below the original target of 500). For block groups the target is achieved in the interval 450 to 499 (again below the original target of 500). At the production settings, the redistricting data can be used reliably for Congressional, state legislative districts, and every one of the 73 redistricting use cases that the Department of Justice (DoJ) Voting Section supplied---a sample that included many low population legislative areas. Reliability means the difference between the ratio of the largest demographic group population to the total population of an area based on the TDA and the corresponding ratio of the largest demographic group population to the total population of the area based on the official 2010 Census Redistricting Data (P.L. 94-171) Summary File is less than or equal to 5 percentage points at least 95\% of the time. Furthermore, for specific redistricting plans, majority-minority districts near the 50\% threshold differ by only a handful of persons and are as likely to flip over as under 50\%. The Voting Rights Act Section 2 analysis for the vast majority of Congressional, state legislative, and DoJ-provided plans is unchanged compared to the results obtained from the official 2010 Census Redistricting Data (P.L. 94-171) Summary File. The complete detailed metrics for general demographic uses cases for all demonstration versions of the DAS and for the production settings run on the 2010 Census data can be found at \citet{2020_das_development}.

Figures \ref{fig:improv_countyunder1000} through \ref{fig:improv_TractHisp} are good representations of the improvement in the TDA over time including those found in the tuned April 2021 release and the final production settings. The figures illustrate four different accuracy metrics; however, it is important to remember, that the statistical optimization problems in Section \ref{sec:estimation} focus on $L_1$ and $L_2$ count accuracy. Mean absolute error is the $L_1$ error in the indicated statistic for the indicated geography averaged for that statistic over all geographic units in the figure's universe.
We note that in the figures, mean errors are calculated over geographic units (e.g., counties in Figure \ref{fig:improv_countyunder1000}) for a single realization of the TDA. The improvements in accuracy over time in these figures were due to improvements within the algorithm, tuning strategies, and increases of the privacy-loss budget.
In comparing the different versions of TDA in a single figure, the November 2020 and April 2021 $(\rho = 0.1885)$  represent algorithmic improvements for the redistricting data using the same total privacy-loss budget. The April 2021 $(\rho = 0.1885)$ compared to April 2021 $(\rho = 1.095)$ show the effects of increasing the privacy-loss budget for the same algorithms and budget allocations. It is also apparent from comparing the graphs that every query does not improve as the experiments progressed. For example, the mean absolute error of the total population in all incorporated places increased between the November 2020 and April 2021 runs at $\rho=0.1885$ (see Figure \ref{fig:improv_PlaceMAPE}) but returned to its November 2020 value at the increased privacy-loss budget of $\rho=1.095$. Such movements happened because algorithmic improvement targeted many metrics allocating a fixed privacy-loss budget causing tradeoffs whereas increasing the privacy-loss budget using the same algorithms increased accuracy across many different metrics.

\begin{figure}[H]
\caption{Mean Absolute Error of the County Total Population among the Least Populous Counties (Population Under 1,000) by Demonstration Data Product Vintage}
\centering
\includegraphics[width=\textwidth]{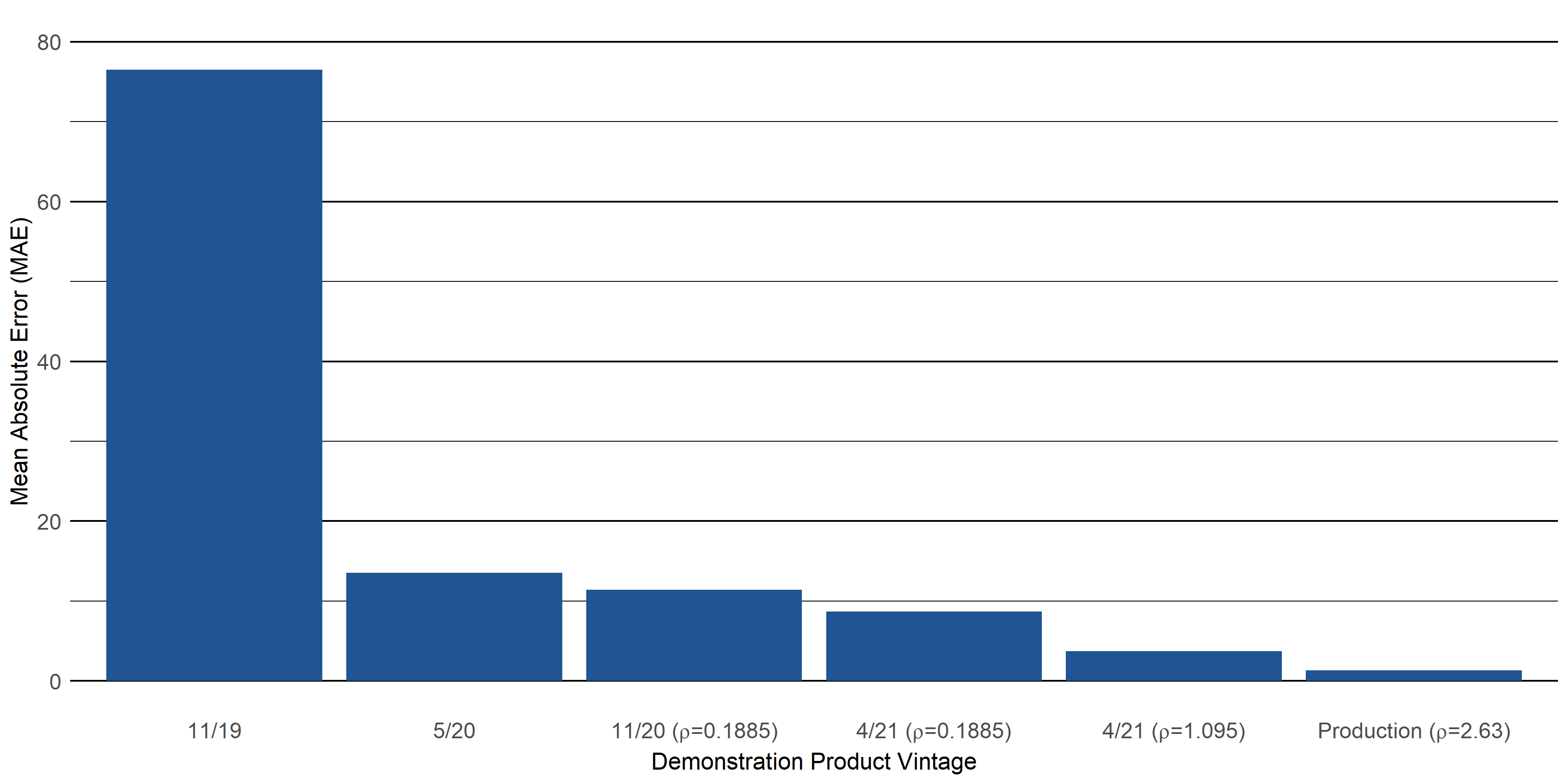}
\label{fig:improv_countyunder1000}
\end{figure}

\begin{figure}[H]
\caption{Mean Absolute Error of the Total Population for Federal American Indian Reservation/Off-Reservation Trust Lands by Demonstration Data Product Vintage}
\centering
\includegraphics[width=\textwidth]{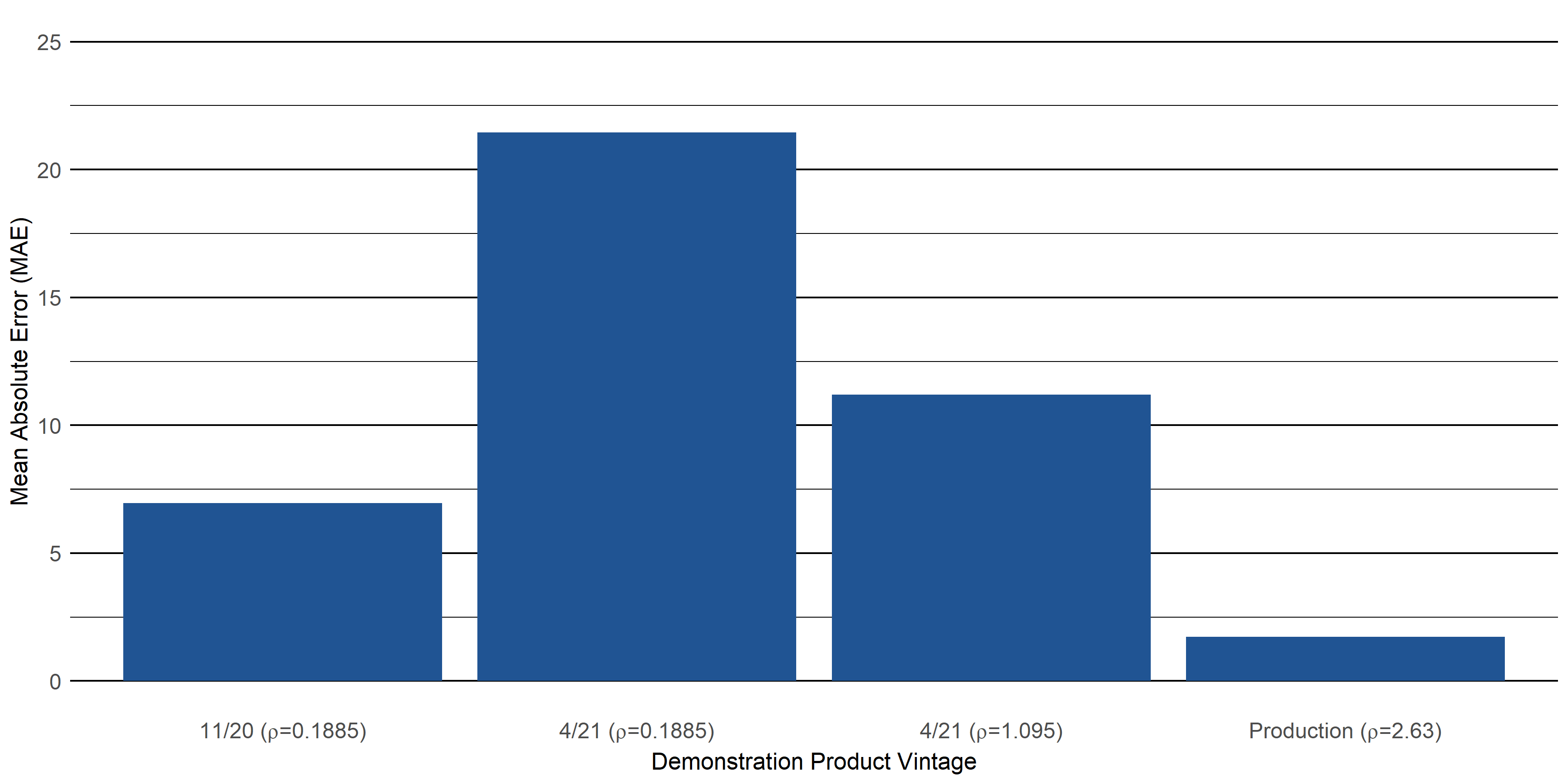}
\label{fig:improv_fedairmae}
\end{figure}

\begin{figure}[H]
\caption{Mean Absolute Error of the Total Population among All Incorporated Places by Demonstration Data Product Vintage}
\centering
\includegraphics[width=\textwidth]{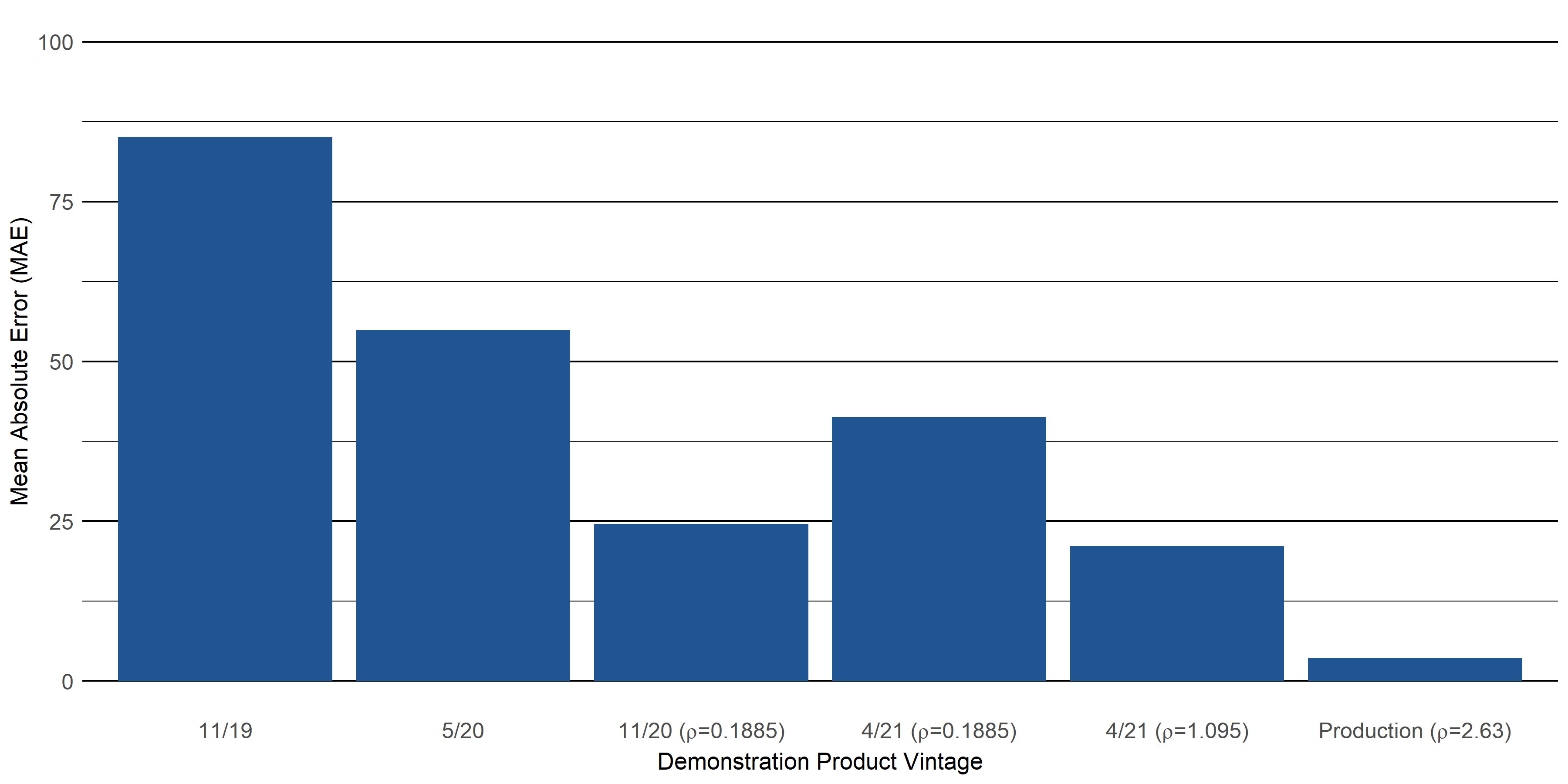}
\label{fig:improv_PlaceMAPE}
\end{figure}

\begin{figure}[H]
\caption{Mean Absolute Error of the Total Population among Tracts for Hispanic x Race Alone Populations by Demonstration Data Product Vintage}
\centering
\includegraphics[width=\textwidth]{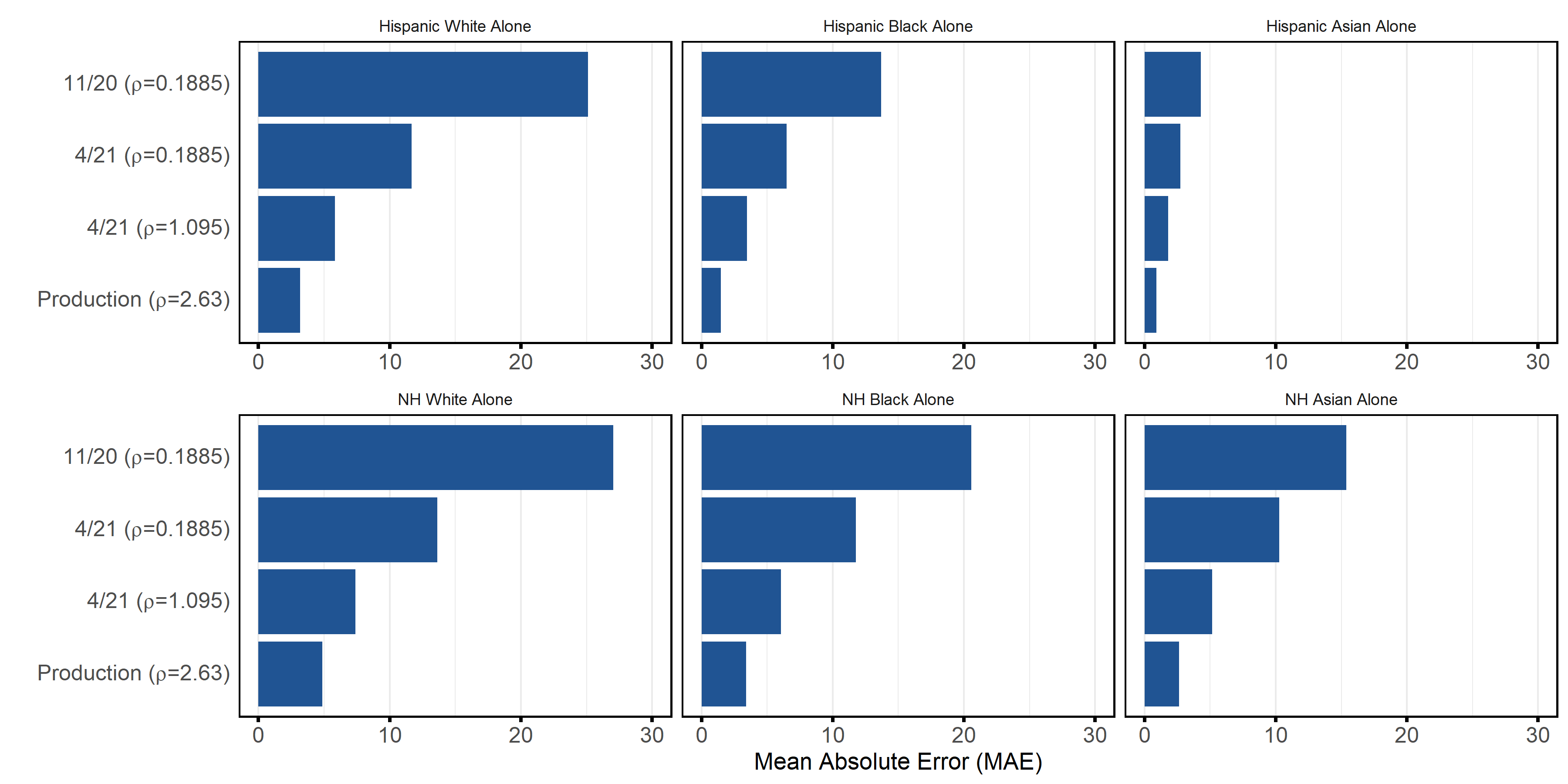}
\label{fig:improv_TractHisp}
\end{figure}

\section{Utility Experiments}
\label{sec:utility_experiments2}

To further study the effect of certain design choices on the utility of the results, we conducted an additional series of experiments after the data release using the 2010 redistricting data. We considered the final production settings (see section \ref{subsec:prod_settings} for full details) as a baseline and completed experimental runs to study the marginal effects of 1) randomness 2) increasing $\rho$, and 3) the multi-pass least squares estimator. These experiments consisted of a total of 4 independent runs of the TDA with 3 different settings. In each case, only the person-level data were used (and not housing units). The first two runs used the same production settings ($\rho = 2.56)$, and we refer to these runs as \emph{baseline1} and \emph{baseline2}. \emph{Baseline1} is the output from the June 8, 2021 PPMF exactly as published whereas \emph{baseline2} uses fresh randomness but is otherwise identical. The third run increased the value of $\rho$ by $10\%$ to $\rho = 2.81$.  We refer to this run as \emph{inc\_rho}. The fourth run removed the $L_2$ multi-pass component used at the state, county, tract, and block-group levels, instead using a single $L_2$ pass. We refer to this run as \emph{l2\_onepass}.

Tables \ref{experiment_tb_total_L1_error1} and \ref{experiment_tb_total_L1_error2} show the absolute error across selected queries averaged across the geographic level of interest. The absolute error for a given query can be represented $\sum_{i=1}^{m}|\hat{\theta_i} - \theta_i|$ where $i=1, \ldots, m$ are the levels in the query\footnote{e.g., $1$ level for TOTAL, $2$ levels for VOTINGAGE, 63 for CENRACE, and $2016$ for DETAILED} with $\hat{\theta_i}$ and $\theta_i$ representing the TDA estimate and the CEF value respectively. In Table \ref{experiment_tb_total_L1_error1} we first observe that there is no error in the TOTAL query at the US and state levels for any of the runs. This is because those values are held invariant. In general, the average error tends to decrease when moving from the US geographic level to the block level, especially for queries with many levels. This effect is largely due to the increasing sparsity of those queries at the lower levels of geography.

When comparing the different runs in Tables \ref{experiment_tb_total_L1_error1} and \ref{experiment_tb_total_L1_error2}, we see that the average errors for \emph{baseline1} and \emph{baseline2} are very consistent with each other, especially at the lower levels of geography where we are averaging over a large number of geographic units. While ideally, average errors could be established by averaging over many runs of the TDA, a single run is extremely computationally expensive. These results show that in evaluating this type of average errors, it is not necessary to look across many runs. We note however, that errors averaged across geographies should not be used as a direct substitute for the run-to-run variability in an individual geography's query error. This is because the underlying error distributions are not necessarily identical across geographic units. We are actively researching better methods for summarizing the error distributions for the published TDA statistics. 

The effect of increasing the overall $\rho$ value by ten percent on the metrics in Tables \ref{experiment_tb_total_L1_error1} and \ref{experiment_tb_total_L1_error2} is noticeable, but small (roughly a $3\%$ reduction in error). We see that using multiple passes instead of one pass consistently improved the errors across many queries. 

\begin{table}[H] 
\caption{Absolute Error in Total, Voting-age, and Hispanic Queries Averaged Across Geographic Units by Geographic Level  }
\label{experiment_tb_total_L1_error1}
\begin{tabular}{|ll|rrrr|}
\hline
Geographic Level & Query & baseline1 & baseline2 & inc\_rho & l2\_onepass  \\
\hline
US & TOTAL & 0.00 & 0.00 & 0.00 & 0.00 \\
State & TOTAL & 0.00 & 0.00 & 0.00 & 0.00  \\
County & TOTAL & 1.75 & 1.75 & 1.71 & 1.73  \\
Tract & TOTAL & 1.94 & 1.93 & 1.84 & 1.88  \\
Block Group & TOTAL & 16.00 & 16.01 & 15.42 & 15.99  \\
Block & TOTAL & 4.85 & 4.85 & 4.69 & 4.85  \\
\hline
US & VOTINGAGE & 70.00 & 92.00 & 64.00 & 112.00  \\
State & VOTINGAGE & 29.69 & 25.57 & 24.08 & 27.14 \\
County & VOTINGAGE & 19.67 & 19.82 & 19.04 & 19.62  \\
Tract & VOTINGAGE & 15.00 & 15.10 & 14.47 & 15.10 \\
Block Group & VOTINGAGE & 21.35 & 21.40 & 20.57 & 21.34  \\
Block & VOTINGAGE & 6.23 & 6.23 & 6.01 & 6.23  \\
\hline
US & HISPANIC & 28.00 & 24.00 & 18.00 & 14.00  \\
State & HISPANIC & 24.24 & 28.00 & 27.25 & 27.02  \\
County & HISPANIC & 20.29 & 20.81 & 19.48 & 20.13  \\
Tract & HISPANIC & 8.11 & 8.12 & 7.81 & 8.06  \\
Block Group & HISPANIC & 20.24 & 20.27 & 19.56 & 20.21 \\
Block & HISPANIC & 5.82 & 5.82 & 5.63 & 5.82 \\
\hline
\end{tabular}
\end{table}

\begin{table}[H] 
\caption{Absolute Error in Queries Averaged Across Geographic Units by Geographic Level  }
\label{experiment_tb_total_L1_error2}
\begin{tabular}{|ll|rrrr|}
\hline
Geographic Level & Query  & baseline1 & baseline2 & inc\_rho & l2\_onepass  \\
\hline
US & CENRACE & 698.00 & 842.00 & 666.00 & 666.00  \\
State & CENRACE & 393.45 & 392.55 & 384.75 & 399.76  \\
County & CENRACE & 126.42 & 126.45 & 122.00 & 127.25  \\
Tract & CENRACE & 35.04 & 35.02 & 33.83 & 35.04 \\
Block Group & CENRACE & 41.36 & 41.39 & 40.17 & 41.39  \\
Block & CENRACE & 8.04 & 8.05 & 7.80 & 8.04 \\
\hline
US & HISPANICxCENRACE & 882.00 & 1054.00 & 930.00 & 890.00  \\
State & HISPANICxCENRACE & 562.00 & 558.31 & 552.75 & 563.53  \\
County & HISPANICxCENRACE & 164.59 & 164.69 & 159.38 & 165.03  \\
Tract & HISPANICxCENRACE & 43.64 & 43.65 & 42.15 & 43.67  \\
Block Group & HISPANICxCENRACE & 49.48 & 49.50 & 48.03 & 49.45  \\
Block & HISPANICxCENRACE & 8.94 & 8.94 & 8.67 & 8.94  \\
\hline
US & DETAILED* & 4456.00 & 4332.00 & 4176.00 & 4212.00 \\
State & DETAILED & 1568.27 & 1573.18 & 1507.14 & 1572.35  \\
County & DETAILED & 294.45 & 293.71 & 284.47 & 294.50  \\
Tract & DETAILED & 95.33 & 95.29 & 91.96 & 95.29  \\
Block Group & DETAILED & 70.18 & 70.18 & 68.03 & 70.14  \\
Block & DETAILED & 10.93 & 10.93 & 10.59 & 10.93  \\
\hline
\end{tabular}
*DETAILED is shorthand for HHGQ×VOTINGAGE×HISPANIC×CENRACE
\end{table}

Table \ref{table_dist_error_county} shows the error distribution of the TOTAL query for each of the experimental runs at the county level grouping by the underlying CEF total population. The table shows both the mean absolute error as well as specific signed error quantiles for each of the runs and total population groupings. In general at the county level, there is a small increase in the average error as the total population increases. The error distributions are fairly symmetric and centered around zero for all but the largest of counties. In those counties there appears to be slightly larger likelihood of having a negative error; however, the errors are quite small relative to the overall population of the county. Comparing the different runs, there is not much variability in these metrics between our baseline runs or between the baseline runs and \emph{inc\_rho} or \emph{l2\_onepass}.

Considering comparable metrics at the block level (Table \ref{table_dist_error_block}), we see that overall, the differences between the error distributions by total population are more pronounced. Blocks with small populations are on average slightly overestimated while blocks with larger populations are on average slightly underestimated by the TDA. This is directly tied to the non-negative query requirement and much more pronounced at the block level because there are many blocks with very small total populations. Comparing the different runs, we see that \emph{inc\_rho} had consistently lower errors than the others at this geographic level.   

\begin{table}[H] 
\caption{Error Distribution in Total Population Query Across Counties by CEF population}
\label{table_dist_error_county}
\begin{tabular}{|c|r|l|rrrr|}
\hline
Total Population & Count & Metric & baseline1 & baseline2 & inc\_rho & l2\_onepass   \\
\hline
\multirow{8}{*}{0 to 1,000} & \multirow{8}{*}{35} & mean L1 & 1.314 & 0.943 & 1.314 & 1.343\\
 &  & q(0.005) & -3 & -2 & -3 & -3\\
 &  & q(0.025) & -3 & -2 & -3 & -3  \\
 &  & q(0.25) & -1 & -1 & -1 & 0 \\
 &  & q(0.5) & 0 & 0 & 1 & 0   \\
 &  & q(0.75) & 1 & 1 & 1 & 2 \\
 &  & q(0.975) & 3 & 3 & 3 & 4   \\
 &  & q(0.995) & 3 & 3 & 3 & 4   \\
 \hline
\multirow{8}{*}{1,001 to 9,999} & \multirow{8}{*}{663} & mean L1 & 1.611 & 1.528 & 1.596 & 1.558  \\
 &  & q(0.005) & -5 & -5 & -6 & -5 \\
 &  & q(0.025)  & -4 & -4 & -4 & -4   \\
 &  & q(0.25)  & -1 & -1 & -1 & -1 \\
 &  & q(0.5) & 0 & 0 & 0 & 0 \\
 &  & q(0.75) & 1 & 1 & 1 & 2  \\
 &  & q(0.975) & 4 & 5 & 4 & 4   \\
 &  & q(0.995) & 6 & 6 & 5 & 6  \\
 \hline
\multirow{8}{*}{10k to 99,999} & \multirow{8}{*}{1867} & mean L1 & 1.77 & 1.819 & 1.736 & 1.774 \\
 &  & q(0.005) & -6 & -6 & -6 & -6  \\
 &  & q(0.025)  & -4 & -4 & -4 & -4  \\
 &  & q(0.25)  & -1 & -2 & -2 & -2  \\
 &  & q(0.5) & 0 & 0 & 0 & 0  \\
 &  & q(0.75) & 2 & 2 & 1 & 2  \\
 &  & q(0.975) & 5 & 5 & 4 & 4\\
 &  & q(0.995) & 6 & 6 & 6 & 6 \\
 \hline
\multirow{8}{*}{100k to 999,999} & \multirow{8}{*}{539} & mean L1 & 1.883 & 1.818 & 1.753 & 1.761 \\
 &  & q(0.005) & -8 & -7 & -6 & -6 \\
 &  & q(0.025)  & -5 & -5 & -4 & -5  \\
 &  & q(0.25)  & -2 & -2 & -2 & -2  \\
 &  & q(0.5) & 0 & 0 & 0 & 0   \\
 &  & q(0.75) & 1 & 1 & 1 & 1  \\
 &  & q(0.975) & 4 & 5 & 4 & 4   \\
 &  & q(0.995) & 5 & 6 & 5 & 5  \\
 \hline
\multirow{8}{*}{1,000k+} & \multirow{8}{*}{39} & mean L1 & 2.026 & 1.923 & 2.256 & 2.333 \\
 &  & q(0.005) & -9 & -10 & -7 & -9   \\
 &  & q(0.025)  & -9 & -10 & -7 & -9   \\
 &  & q(0.25)  & -3 & -2 & -2 & -3  \\
 &  & q(0.5) & -1 & -1 & 0 & -2  \\
 &  & q(0.75) & 0 & 1 & 1 & 0   \\
 &  & q(0.975) & 4 & 4 & 7 & 4  \\
 &  & q(0.995) & 4 & 4 & 7 & 4 \\
 \hline
\end{tabular}

{\small{ Note that q(x) represents the xth quartile.} }
\end{table}

\begin{table}[H] 
\caption{Error Distribution in Total Population Query Across Blocks by CEF population}
\label{table_dist_error_block}
\begin{tabular}{|c|r|l|rrrr|}
\hline
Total Population & Count & metric & baseline1 & baseline2 & inc\_rho & l2\_onepass \\
\hline
\multirow{8}{*}{0} & \multirow{8}{*}{191,175} & l1\_mean & 2.746 & 2.745 & 2.672 & 2.75 \\
 &  & q(0.005) & 0 & 0 & 0 & 0 \\
 &  & q(0.025) & 0 & 0 & 0 & 0 \\
 &  & q(0.25) & 0 & 0 & 0 & 0 \\
 &  & q(0.5) & 2 & 2 & 2 & 2 \\
 &  & q(0.75) & 4 & 4 & 4 & 4 \\
 &  & q(0.975) & 11 & 11 & 11 & 11 \\
 &  & q(0.995) & 17 & 17 & 16 & 17 \\
 \hline
\multirow{8}{*}{1 to 9} & \multirow{8}{*}{1,823,665} & l1\_mean & 3.332 & 3.333 & 3.224 & 3.332 \\
 &  & q(0.005) & -7 & -7 & -7 & -7 \\
 &  & q(0.025) & -5 & -5 & -5 & -5 \\
 &  & q(0.25) & -1 & -1 & -1 & -1 \\
 &  & q(0.5) & 1 & 1 & 1 & 1 \\
 &  & q(0.75) & 4 & 4 & 4 & 4 \\
 &  & q(0.975) & 12 & 12 & 11 & 12 \\
 &  & q(0.995) & 17 & 17 & 17 & 17 \\
 \hline
\multirow{8}{*}{10 to 99} & \multirow{8}{*}{3,666,266} & l1\_mean & 4.824 & 4.828 & 4.655 & 4.825 \\
 &  & q(0.005) & -17 & -17 & -17 & -17 \\
 &  & q(0.025) & -12 & -12 & -12 & -12 \\
 &  & q(0.25) & -4 & -4 & -4 & -4 \\
 &  & q(0.5) & 0 & 0 & 0 & 0 \\
 &  & q(0.75) & 4 & 4 & 4 & 4 \\
 &  & q(0.975) & 13 & 13 & 13 & 13 \\
 &  & q(0.995) & 19 & 19 & 19 & 19 \\
 \hline
\multirow{8}{*}{100 to 999} & \multirow{8}{*}{707,291} & l1\_mean & 9.153 & 9.158 & 8.851 & 9.177 \\
 &  & q(0.005) & -43 & -43 & -42 & -43 \\
 &  & q(0.025) & -30 & -30 & -29 & -31 \\
 &  & q(0.25) & -12 & -12 & -12 & -12 \\
 &  & q(0.5) & -5 & -5 & -5 & -5 \\
 &  & q(0.75) & 1 & 1 & 1 & 1 \\
 &  & q(0.975)  & 13 & 13 & 12 & 13 \\
 &  & q(0.995) & 20 & 20 & 19 & 20 \\
 \hline
\multirow{8}{*}{1,000+} & \multirow{8}{*}{9,805} & l1\_mean & 27.124 & 27.28 & 26.436 & 27.378 \\
 &  & q(0.005) & -88 & -89 & -83 & -88 \\
 &  & q(0.025) & -71 & -71 & -69 & -72 \\
 &  & q(0.25) & -41 & -41 & -40 & -41 \\
 &  & q(0.5) & -26 & -26 & -25 & -26 \\
 &  & q(0.75) & -8 & -8 & -8 & -8 \\
 &  & q(0.975) & 7 & 8 & 7 & 7 \\
 &  & q(0.995) & 18 & 19 & 17 & 18 \\
 \hline
\end{tabular}

{\small{ Note that q(x) represents the xth quartile.} }
\end{table}

Table \ref{blau_error_table} shows the signed error in the total population query after grouping geographies into quintiles by their the Blau index \citep{blau1977inequality} for each of the five TDA runs. The Blau index is a measure of group heterogeneity and was calculated per geography over eight race-ethnicity levels: Hispanic, non-Hispanic White alone, non-Hispanic Black or African American alone, non-Hispanic American Indian and Alaska Native alone, non-Hispanic Asian alone, non-Hispanic Native Hawaiian and Other Pacific Islander alone, non-Hispanic Some Other Race alone, and non-Hispanic Two or More Races. A larger Blau index indicates more heterogeneity. The Blau index quintile groupings are based on the 2010 CEF. The errors by Blau index quintile were one of the main motivations for developing the multi-pass estimators. Notice the improvement in the two baseline metrics compared with \emph{l2\_onepass} at the county, tract, and block group. The first pass in the multi-pass estimator, as used in the production setting, locks in the total population estimate and is less prone to the side effects of enforcing non-negativity. The small error values associated with the increase privacy-loss budget can also be seen in this table, though the effect appears to be small around the production value of $\rho$.

\begin{table}[H] 
\caption{Signed Count Error in Population Total Query Averaged Across Geographic Units by Geographic Level and Blau Index Quintile}
\label{blau_error_table}
\begin{tabular}{|l|l|rrrr|}
\hline
Geographic Level & Quintile  & baseline1 & baseline2 & inc\_rho & l2\_onepass \\
\hline
\multirow{5}{*}{County}  & 1 & 0.17 & 0.14 & 0.14 & 0.39 \\
 & 2 & 0.01 & 0.06 & -0.09 & 0.04 \\
 & 3 & 0.01 & -0.01 & -0.13 & -0.04 \\
 & 4 & -0.04 & -0.04 & 0.11 & -0.11 \\
 & 5 & -0.14 & -0.15 & -0.03 & -0.28 \\
\hline
\multirow{5}{*}{Tract} & 1 & 0.07 & 0.07 & 0.11 & 0.34 \\
 & 2 & 0.03 & 0.03 & 0.02 & 0.08 \\
 & 3 & -0.02 & -0.02 & -0.01 & -0.02 \\
 & 4 & -0.05 & -0.04 & -0.07 & -0.13 \\
 & 5 & -0.04 & -0.04 & -0.05 & -0.26 \\
\hline
\multirow{5}{*}{Block Group} & 1 & 1.69 & 1.70 & 1.66 & 1.81 \\
 & 2 & 1.55 & 1.55 & 1.53 & 1.65 \\
& 3 & 0.69 & 0.69 & 0.62 & 0.67 \\
 & 4 & -0.90 & -0.90 & -0.77 & -0.89 \\
& 5 & -2.91 & -2.91 & -2.92 & -3.10 \\
\hline
\multirow{4}{*}{Block} & 1 and 2 & 1.55 & 1.55 & 1.50 & 1.55 \\
 & 3 & -0.22 & -0.22 & -0.21 & -0.21 \\
 & 4 & -1.07 & -1.07 & -1.05 & -1.07 \\
 & 5 & -1.67 & -1.67 & -1.62 & -1.67 \\
\hline
\end{tabular}
\end{table}

The utility comparisons in this section have focused on mean absolute error and the distribution of signed errors in the DAS. These comparisons permit assessment of the factors contributing to uncertainty in the published data caused by the disclosure avoidance by comparison to the 2010 Census publications.  We now compare the uncertainty in the production settings of the DAS with other sources of uncertainty in the 2010 Census using \citet{DASteam:uncertainty:2022} and \citet{bell:schafer:2022}, which are the sources for statistics below.  

At the block level the mean absolute error in the total population is 4.8 persons and 90\% of the blocks with housing units or GQ population have errors between -11 and 10. The most populous blocks have asymmetric errors in total population. For example, in blocks with total population above 3,250, the mean absolute error is 22.3 persons and 90\% of the blocks have errors between -80 and 3. The simulation studies exclude the group quarters population for technical reasons, but they still provide interesting comparisons. The first simulation (Schafer) measures ``natural variability in population counts over repeated realizations  of  the  census.'' That is, the simulation holds the ``true population,'' which is latent, constant at the observed 2010 enumeration and simulates repeated censuses with operational and coverage error rates observed in 2010.  In this simulation block level population in 2010 has mean absolute error of 1.5 persons and 90\% of the blocks have errors between -4 and 4. For blocks with populations of 1,000 or more, the mean absolute error is 14 persons and 90\% of the blocks have errors between -31 and 30. The second simulation (Bell) measures variation in the ``coverage error,'' the difference between the dual system estimate (DSE) of the true population and the 2010 Census enumerated population using models for the components of the DSE. The mean absolute error in the block population from this simulation is 5.4 persons and 90\% of the blocks have errors between -11 and 12. At the block level, the DAS contributes uncertainty to population counts that appears to be comparable to the uncertainty contributed by census operational, measurement and coverage errors.

The picture is very different for more aggregated geographies. For all counties, the mean absolute error in the DAS is 1.75 and 90\% of counties have errors between -4 and 4. In the Schafer simulation, the mean absolute error in the county population is 117 and 90\%  of counties have errors between -248 and 230. In the Bell simulation, the mean absolute county population error is 964 and 90\% of counties have errors between -1,841 and 2,048. The uncertainty in county populations from the DAS is trivial compared to the estimated uncertainty contributed by operational, measurement and coverage errors. The explanation for this result is straightforward. The DAS was designed to control the error from disclosure limitation at all levels of geography, whereas operational, measurement and coverage errors largely accumulate as the population in the geographic area increases. For this reason, block-level uncertainty caused by the DAS can provide confidentiality protection of the respondent's location without contributing materially to the uncertainty in populations of larger geographic units like MCDs, towns, cities, American Indian Reservations, and counties.

\section{Conclusion}
\label{sec:conclusion}
The development of the TopDown Algorithm as the core disclosure avoidance technology for the publications from the 2020 Census began in 2016. Initial work focused on implementing pure differential privacy using the Laplace mechanism. This preliminary work culminated in the release, in April 2019, of the  Redistricting Data (P.L. 94-171) Summary File from the 2018 End-to-End Census Test. The initial implementation demonstrated feasibility. Since the use case for the End-to-End Test data product was limited to the correctness of the dissemination format and dictionary, this initial release received little public scrutiny. 

Beginning with the October 2019 release of the first DAS demonstration data product, the \emph{Harvard Data Science Review} Symposium that same month, and the Committee on National Statistics (CNSTAT) expert workshop in December 2019, there was intensive internal and public scrutiny of the 2020 Census DAS and TDA. The dialogue continued through the May 2020 demonstration data product release, which was prepared just as the Census Bureau announced the suspension of 2020 Census operations due to the COVID-19 pandemic. As the operations of the 2020 Census resumed in July 2020, many adjustments of the operational and publication timetables occurred. The development of the DAS and the refinement of TDA were limited to the Redistricting Data Summary File in an effort to accelerate the production of those data in light of the collection delays that the pandemic caused. The September 2020, November 2020, and April 2021 demonstration data products were limited to the tables in the redistricting data. These products documented the successive refinements of TDA. 

During the one-month public comment period following the April 2021 demonstration data release, the Census Bureau received 48 detailed comments as well as extensive recommendations from the Census Scientific Advisory Committee and the National Advisory Committee. Those recommendation along with the recommendations from all internal stakeholders within the Census Bureau guided the deliberations of the Data Stewardship Executive Policy Committee. DSEP met for a total of 4.5 hours from June 3 through June 8 to assess the privacy guarantees and accuracy of the TDA as applied to the 2010 Census Edited File. DSEP's instructions were incorporated into the final production parameters for the 2020 Census implementation.

On June 25, 2021, the production DAS, implementing TDA as DSEP instructed, started. On June 26, the 2020 Microdata Detail File for the redistricting data was released to the Demographic Programs Directorate for quality assessment. On July 17, the MDF passed the full data quality assessment and was delivered to the 2020 Census tabulation system for publication. On August 12, 2021 the Census Bureau released the 2020 Census Redistricting Data  (P.L. 94-171) Summary File \citep{pl94:2020} and the final redistricting data demonstration product based on the same production code and parameters but using the 2010 Census data \citep{2020_das_development}. In September 2021, the full production code base for the redistricting data was also released including all parameter values and supporting tables \citep{pl94:2020:codebase}.  

Further data products, especially the Demographic and Housing Characteristics File, must now be processed through the DAS. Tuning is already underway for DHC. The Census Bureau expects to announce the schedule for DHC demonstration products and official publication in early 2022.


\appendix
\section{The Block-by-block Algorithm}
\label{subsec:block_by_block}

A simple and straightforward method of creating a differentially private estimate of $\mathbf{x}$ is to work directly with the leaves of the 
geographic spine (a directed, rooted tree), which are the census blocks.
We call this method the \emph{Block-by-block Algorithm}. Table \ref{tab:block_by_block_summary} presents a schematic diagram of this algorithm.

The first phase of the block-by-block algorithm is to take differentially private measurements of key queries for each block using the block-level contingency table vectors. The second phase is to, for each block, combine the information from the differentially private measurements to generate a non-negative integer estimate of each block's contingency table vector. While there are many estimators for this second phase, one method that we deployed to give non-negative integer solutions involves first finding a non-negative least squares solution and then using a controlled rounding method to produce the integer solutions for each block-level statistic.

The block-by-block algorithm's estimates of the contingency table vector at higher geographic levels are the aggregates of the block-level contingency table vector estimates. This method can be used in parallel across the millions of blocks and, since blocks are disjoint, the entire $\rho$ value can be applied to each block's measurements.\footnote{We make an appropriate sensitivity adjustment because we assumed that the total number of records in the database is fixed.} For example, if we take a set of $\rho=1$ zCDP measurements for each block, the total $\rho$ budget would still be 1.

The block-by-block algorithm has two major limitations. First, it is difficult to control the accuracy of estimated queries above the block level. This is because aggregate geography estimates are simply sums of the the block-level estimates. The differentially private random error of the block-level measurements are independent of one another. Therefore, the 
squared error
of the aggregated estimates is at least as large as the summation of the block-level variances for all the blocks used in a particular aggregation.\footnote{Nonnegativity constraints 
further introduce bias.
This is particularly important when summing over block-level estimates because they have the sparsest histograms.} For example, if it was important to finely tune the accuracy of the nationwide race-ethnicity totals, this could not be done directly. It would only be possible through adjusting the accuracy of the block-level race-ethnicity totals, which would require extremely tight block-level estimates since 5,892,698 blocks can be inhabited in the 2020 Census (see Table \ref{table:geo_entity_counts}).  The second limitation of the block-by-block algorithm is the inability to enforce constraints at higher levels of geography without also enforcing them at the block level (and intermediary levels). For example, imagine that the state-level total populations were designated as invariant (but not sub-geographies). The way the block-by-block algorithm would enforce those constraints would also imply that the total populations at the block, block group, tract, and county levels must also be invariant.  

\begin{table}[]
    \centering
        \caption{The Block-by-block Algorithm Summary}
        \label{tab:block_by_block_summary}
    \begin{tabular}{|p{14.5cm}|}
        \hline
        \emph{Measurement Phase}\\
         \hline
         \begin{tabular}{l}
            (1) Using the total $\rho$ budget, take differentially private noisy measurements $\tilde{M}_{\gamma}$\\
            for each node $\gamma$ at the block level.\\
         \end{tabular}
         \\
         \hline
         \emph{Estimation Phase}\\
         \hline
         \begin{tabular}{l}
              (1) For each block node $\gamma$, estimate the contingency table vector $\mathbf{x}_{\gamma}$ by\\
              \hspace{0.5cm} (a) Estimating a non-negative solution $\tilde{\mathbf{x}}_{\gamma}$ from differentially private \\
              \hspace{0.5cm} noisy measurements $\tilde{M}_{\gamma}$ and constraints;\\ 
              \hspace{0.5cm} (b) Estimating a non-negative integer solution $\hat{\mathbf{x}}_{\gamma}$ from $\tilde{\mathbf{x}}_{\gamma}$ by controlled rounding.\\
         \end{tabular}
    \\
    \hline
    \end{tabular}
\end{table}

\subsection{Block-by-block versus TDA comparison}
Consider a simple illustrative example of the block-by-block mechanism and how it compares to the TDA in terms of utility. Suppose that the only attribute is total population from the block level to the national level. Suppose also, for the sake of this example, that there are no invariants and negative estimates are permitted. Using a total $\rho$-budget of $\frac{1}{2}$, we could use the continuous Gaussian mechanism in the block-by-block algorithm to generate measurements of the block-level total populations with variance $\sigma^2=1$ for each of the $5,892,698$ blocks for the 2020 Census. This would produce extremely accurate differentially private data at the block level, but what are the implications for aggregations of blocks?

Block groups, tracts, counties, states, and the US are all aggregations of blocks. Since the differentially private measurements add independent noise to each block population, the variance of estimates of the aggregates is the sum of the block-level variances. For the US, the estimated total population is unbiased, but, in this most extreme case, the variance of the estimated total population is $5,892,698$. Clearly, this is an unacceptably large mean squared error.  Now consider the comparable TDA-inspired approach. Dividing the $\rho$-budget of $\frac{1}{2}$ into six equal parts of $\frac{1}{12}$ each for the US, state, county, tract, block group, and block levels and taking a differentially private measurement of the total population. Again, all the estimated populations are unbiased; however, allocating $\rho$ across all levels of the hierarchy results in a variance of $\sigma^2=6$ for each measurements. Compared with the block-by-block variance for the total population at the national level, there is a huge mean squared error improvement. At lower geographic levels except for blocks and block groups with fewer than six blocks, the measurements generated from the TDA-inspired approach will also be more precise than block-by-block even without borrowing strength between estimates at different levels of the hierarchy. Early experimental implementations of TDA using public 1940 Census data have shown that TDA produced more accurate estimates even at the lowest-geographic level (enumeration district) than the block-by-block algorithm \citep{AbowdCSAC2018}.

\section{The TDA Satisfies $(\epsilon,\delta)$ Approximate Differential Privacy}
\label{sec:approx_dp_proof}

As discussed in Section \ref{subsec:basic_defs}, a mechanism that is $\rho$-zCDP also satisfies $(\epsilon,\delta)$--approximate differential privacy. When converting $\rho$ to $(\epsilon, \delta)$, we intentionally did not use the tightest conversion \citep{NEURIPS2020_b53b3a3d,canonne2020discrete}, but instead chose to use a conversion \citep{10.1007/978-3-662-53641-4_24} that overestimates $\delta$ and is interpretable as a upper bound on the probability that a particular level of $\epsilon$ does not hold \citep{asoodeh:etal:2020:9b52f4c81b6f4c50a766fe9675b81066}. This appendix shows how we derive the correct $\rho$ values for the individual noisy queries given a target overall $\rho$ privacy-loss budget. That overall $\rho$ value is then converted to an $\epsilon$ value given $\delta = 10^{-10}$.

Consider the univariate queries in the TDA defined as $q_{ijkl}$ where $i$ specifies the geographic level (US, state, etc.), $j$ specifies the node within the level $i$, $k$ indexes the marginal query set (e.g., CENRACE) within the node, and $l$ indexes the individual cells of that marginal query set (within $i,j,k$). 

The query strategy for the TDA consists of two parts: which queries are asked and what privacy-loss budget should be used for each. In early versions of the TDA, the query strategy was symmetric over geographic levels and nodes meaning that the same set of queries would be asked for every geography and level; however, we have since generalized the strategy to allow for different query sets to be asked for different levels. We require that the same set of queries be asked for every node within a level. Similar to the generalization of which queries are asked for each level, we generalize the privacy-loss budget allocations to allow for differences between and within levels.    

We denote the relative precision of a query by $c_{ij} \times d_{ijk}$ where $c_{ij} \in (0,1]$ gives the $\rho$ proportion values for a specific node within a geographic level and $d_{ijk} \in (0,1]$ the $\rho$ proportion for the $k$th query within the specific node. Given the values $\{c_{ij}\}$ and $\{d_{ijk}\}$, the goal is to find a parameter $(\psi)$, which we call the \emph{global scale}, such that asking noisy queries with the discrete Gaussian mechanism with parameter
$$
\sigma_{ijkl}^2 = \left( \frac{\psi^2}{c_{ij} d_{ijk}}  \right)
$$
for queries $\{q_{ijkl}\}$ results in $\rho$-zCDP.  

We impose the following requirements for
the structure of $\{c_{ij}\}$ and $\{d_{ijk}\}$.  First, we require that for all $i,j$, $\sum_{k = 1, \cdots K_{ij} } d_{ijk} = 1$. That is within a level and node, the relative split of the precision across queries always totals to 1. Second, let $P(1, Ij)$ represent the indices of the directed path between the root level (US) and block $j$. Then, we require that $\sum_{ij \in P(1, Ij)} c_{ij} = 1, \forall\, j$, meaning that if we add up the proportion values $c_{ij}$ for a given block and all nodes it is nested within, up to the US, we get 1.  

\begin{theorem} \label{gen_thrm}
The set of queries $\left\{q_{ijkl}\right\}$
processed by the discrete Gaussian mechanism with parameters $\sigma_{ijkl}^2 = \left( \frac{\psi^2}{c_{ij} d_{ijk}} \right)$
where
\begin{enumerate}
    \item All queries used by TDA are marginals
    \item $c_{ij} \ge 0; d_{ijk} \ge 0 \; \forall \; i,j,k$ \footnote{Functionally, if $c_{ij}$ or $d_{ijk}$ is $0$, then the query is not asked.}
    \item For all $i,j$ $\sum_{k = 1, \cdots, K_{ij} } d_{ijk} = 1$
    \item $\sum_{ij \in P(1, Ij)} c_{ij} = 1, \forall\, j$
\end{enumerate}

results in $\rho$--zCDP with 
$\rho = \frac{1}{\psi^2}$.
\end{theorem}

\begin{proof}
Applying Theorem 14 of \cite{canonne2020discrete} to queries $\{q_{ijkl}\}$ we want to find $\rho$ such that 
$$ 
\frac{1}{\psi^2} \sum_{i} \sum_{j} \sum_{k} \sum_{l} (c_{ij}d_{ijk}) (q_{ijkl}(x) - q_{ijkl}(x'))^2 \le 2 \rho.
$$
In the case of the TDA, it is easiest to split the proof into two cases. The first being that the change from $x$ to $x'$ does not affect the block count totals (that is, the record change does not change the given block). The second being that the change from $x$ to $x'$ does affect the block count totals (the record change does change the given block).


\noindent Case 1:\\
The change from $x$ to $x'$ will affect the queries in exactly one node of each geographic level because the record change did not change the geography of the record. Let $j^*_i$ be the node in level $i$ that is affected. Any queries outside of these nodes will have values of 0 for squared differences in the queries for $x$ and $x'$. Then, 

\begin{align*}
&\frac{1}{\psi^2} \sum_{i} \sum_{j} \sum_{k} \sum_{l} (c_{ij}d_{ijk}) (q_{ijkl}(x) - q_{ijkl}(x'))^2  \\
&=\frac{1}{\psi^2} \sum_{i}  \sum_{k} \sum_{l} (c_{ij^*_i}d_{ij^*_ik}) (q_{ij^*_ikl}(x) - q_{ij^*_ikl}(x'))^2  \\
&=\frac{1}{\psi^2} \sum_{i} c_{ij^*_i}  \sum_{k} (d_{ij^*_ik}) \sum_{l}  (q_{ij^*_ikl}(x) - q_{ij^*_ikl}(x'))^2 \\
& \le \frac{1}{\psi^2} \sum_{i} c_{ij^*_i}  \sum_{k} (d_{ij^*_ik}) (1^2 + (-1)^2)\\
&\text{[Since all queries are marginals, $q_{i j^*_i k l}(x)$ and $q_{i j^*_i k l}(x^\prime)$} \\
&\text{differ in at most 2 values of $l$] } \\
&=  \frac{1}{\psi^2} \sum_{i} c_{ij^*_i}  (1) (2)\\
&\text{[Since  $\forall \; i,j$ $\sum_{k = 1, \cdots, K_{ij} } d_{ijk} = 1$] } \\
&= \frac{2}{\psi^2} \\
&\text{[Since $\sum_{ij \in P(1, Ij)} c_{ij} = 1, \forall\, j$]. } 
\end{align*}
Therefore, $\rho = \frac{1}{\psi^2}$.\\

\noindent Case 2:\\
Using $i=I$ to denote the block geographic level (leaves of the hierarchy) and without loss of generality, let $j^*_I$ and $j^{**}_I$ be indices for the block nodes that were affected by the change from $x$ to $x'$, respectively.
By design, we know that $j^*_I \ne j^{**}_I$. Consider the generalization of this change such that $j^*_i$ is the node at level $i$ of the unit before the change from $x$ to $x'$ and $j^{**}_i$ is the node at level $i$ after the change. By definition, at the US level, $j^{*}_1 = j^{**}_1$, and it is possible that $j^{*}_i = j^{**}_i$ for other levels as well. For example, consider a change from $x$ to $x'$ that changes the block of a record but in a way such that the county is the same. Because of the tree hierarchy of geographic levels, if $j^{*}_i = j^{**}_i$, then $j^{*}_{i-1} = j^{**}_{i-1}$. Therefore, for all possible $x$ to $x'$ changes in our case, there exists $S \in \{1,\cdots,I-1\}$ such that for $j^{*}_i = j^{**}_i$ for all $i \le S$. Therefore,          

\begin{align*}
&\frac{1}{\psi^2} \sum_{i} \sum_{j} \sum_{k} \sum_{l} (c_{ij}d_{ijk}) (q_{ijkl}(x) - q_{ijkl}(x'))^2 \\
&= \frac{1}{\psi^2} \sum_{i=1}^{S} \sum_{j} \sum_{k} \sum_{l} (c_{ij}d_{ijk}) (q_{ijkl}(x) - q_{ijkl}(x'))^2 \\
&+ \frac{1}{\psi^2} \sum_{i=S+1}^{I}   \sum_{j} \sum_{k} \sum_{l} (c_{ij}d_{ijk}) (q_{ijkl}(x) - q_{ijkl}(x'))^2.
\end{align*}

Now consider the first part of the summation. Similar to Case 1, the change from $x$ to $x'$ doesn't change the geography of the record at the given level, only at most the cell. Therefore, at most two cells will be changed giving 
\begin{align*}
&\frac{1}{\psi^2} \sum_{i=1}^{S} \sum_{j} \sum_{k} \sum_{l} (c_{ij}d_{ijk}) (q_{ijkl}(x) - q_{ijkl}(x'))^2 \\
&= \frac{1}{\psi^2} \sum_{i=1}^{S}  \sum_{k} \sum_{l} (c_{ij^{*}_i}d_{ij^{*}_ik}) (q_{ij^{*}_ikl}(x) - q_{ij^{*}_ikl}(x'))^2 \\
&= \frac{1}{\psi^2} \sum_{i=1}^{S}  c_{ij^{*}_i} \sum_{k} \sum_{l} (d_{ij^{*}_ik}) (q_{ij^{*}_ikl}(x) - q_{ij^{*}_ikl}(x'))^2 \\
&\le \frac{1}{\psi^2} \sum_{i=1}^{S}  c_{ij^{*}_i} \sum_{k} (d_{ij^{*}_ik}) (1^2 +(-1)^2) \\
&= \frac{2}{\psi^2} \sum_{i=1}^{S}  c_{ij^{*}_i}  \\
&\text{[Since  $\forall \; i,j$ $\sum_{k = 1, \cdots K_{ij} } d_{ijk} = 1$] .}
\end{align*}

Now consider the second half of the summation. We know that for the affected nodes $j^*_i$ and $j^{**}_i, i > S$ all query marginals will be affected by exactly +1 or exactly -1 respectively. The queries for all other nodes will be unchanged. Then,   

\begin{align*}
&\frac{1}{\psi^2} \sum_{i=S+1}^{I}   \sum_{j} \sum_{k} \sum_{l} (c_{ij}d_{ijk}) (q_{ijkl}(x) - q_{ijkl}(x'))^2 \\
&= 0+ \frac{1}{\psi^2} \sum_{i=S+1}^{I}    c_{ij^{*}_{i}} \sum_{k} d_{ij^{*}_{i}k} \sum_{l} (q_{ij^*_{i}kl}(x) - q_{ij^*_{i}kl}(x'))^2 \\
&+ \frac{1}{\psi^2} \sum_{i=S+1}^{I}    c_{ij^{**}_{i}} \sum_{k} d_{ij^{**}_{i}k} \sum_{l}  (q_{ij^{**}_{i}kl}(x) - q_{ij^{**}_{i}kl}(x'))^2\\
& \le \frac{1}{\psi^2} \sum_{i=S+1}^{I}    c_{ij^{*}_{i}} \sum_{k} d_{ij^{*}_{i}k} (-1)^2
+ \frac{1}{\psi^2} \sum_{i=S+1}^{I}    c_{ij^{**}_{i}} \sum_{k} d_{ij^{**}_{i}k} (1)^2\\
&= \frac{1}{\psi^2} \sum_{i=S+1}^{I}    c_{ij^{*}_{i}}  
+ \frac{1}{\psi^2} \sum_{i=S+1}^{I}    c_{ij^{**}_{i}} \\
&\text{[Since  $\forall \; i,j$ $\sum_{k = 1, \cdots K_{ij} } d_{ijk} = 1$] .}
\end{align*}

Pairing this with the earlier result yields 

\begin{align*}
&\frac{1}{\psi^2} \sum_{i} \sum_{j} \sum_{k} \sum_{l} (c_{ij}d_{ijk}) (q_{ijkl}(x) - q_{ijkl}(x'))^2 \\
&\le  \frac{2}{\psi^2} \sum_{i=1}^{S}  c_{ij^{*}_i} + \frac{1}{\psi^2} \sum_{i=S+1}^{I}    c_{ij^{*}_{i}}  
+ \frac{1}{\psi^2} \sum_{i=S+1}^{I}    c_{ij^{**}_{i}} \\
&=  \frac{1}{\psi^2} \left(\sum_{i=1}^{S}  c_{ij^{*}_i} + \sum_{i=S+1}^{I}    c_{ij^{*}_{i}} \right) 
+ \frac{1}{\psi^2} \left( \sum_{i=1}^{S}  c_{ij^{*}_i} + \sum_{i=S+1}^{I} c_{ij^{**}_{i}} \right) \\
&=  \frac{1}{\psi^2} \left(\sum_{i=1}^{I}    c_{ij^{*}_{i}} \right) 
+ \frac{1}{\psi^2} \left(\sum_{i=1}^{I} c_{ij^{**}_{i}} \right) \\
& = \frac{2}{\psi^2} \\
&\text{[Since $\sum_{ij \in P(1, Ij)} c_{ij} = 1, \forall\, j$  ].}
\end{align*}

\noindent Therefore, $\rho = \frac{1}{\psi^2}$.
\end{proof}

\printbibliography

@inproceedings{DMNS06,
  author    = {Cynthia Dwork and
               Frank McSherry and
               Kobbi Nissim and
               Adam D. Smith},
  title     = {{Calibrating Noise to Sensitivity in Private Data Analysis}},
  booktitle={{Theory of Cryptography Conference}},
  year      = {2006},
  pages={265--284},
  organization={Springer},
  month = {12},
  day = {12}
}

@InProceedings{ashwin08:map,
  author =   {Ashwin Machanavajjhala and Daniel Kifer and John M. Abowd and Johannes Gehrke and Lars Vilhuber},
  title =    {Privacy: From Theory to Practice On the Map},
  booktitle =    {Proceedings of the IEEE International Conference on Data Engineering (ICDE)},
  year = {2008},
  pages = {277-286},
  note={\url{https://ieeexplore.ieee.org/document/4497436}}
}

@misc{onthemap,
  title = {On The Map: Longitudinal Employer-Household Dynamics},
  author = {{U.S. Census Bureau}},
  note={\url{https://onthemap.ces.census.gov/}},
  year = {2008},
}

@inproceedings{rappor,
 author = {Erlingsson, \'{U}lfar and Pihur, Vasyl and Korolova, Aleksandra},
 title = {RAPPOR: Randomized Aggregatable Privacy-Preserving Ordinal Response},
 booktitle = {Proceedings of the 2014 ACM SIGSAC Conference on Computer and Communications Security},
 series = {CCS '14},
 year = {2014},
 note={\url{https://arxiv.org/abs/1407.6981}}
}

@article{elasticsensitivity,
 author = {Johnson, Noah and Near, Joseph P. and Song, Dawn},
 title = {Towards Practical Differential Privacy for SQL Queries},
 journal = {Proc. VLDB Endow.},
 issue_date = {January 2018},
 volume = {11},
 number = {5},
 year = {2018},
}

@article{applediffp,
  title = {Learning with Privacy at Scale},
  author = {{Apple Differential Privacy Team}},
  journal = {Apple Machine Learning Journal},
  note={\url{https://machinelearning.apple.com/2017/12/06/learning-with-privacy-at-scale.html}},
  volume = {1},
  number = {8},
  year = {2017},
}

@misc{abowd18kdd,
    author = {Abowd, John M.},
    title = {{KDD 2018 Invited Lecture: The U.S. Census Bureau Adopts Differential Privacy}},
    year = {2018},
    month = {8},
    day = {23},
    journal = {KDD},
    note={\url{https://dl.acm.org/doi/10.1145/3219819.3226070}}
}

@inproceedings{prochlo,
    author = {Bittau, Andrea and Erlingsson, \'{U}lfar and Maniatis, Petros and Mironov, Ilya and Raghunathan, Ananth and Lie, David and Rudominer, Mitch and Kode, Ushasree and Tinnes, Julien and Seefeld, Bernhard},
    title = {Prochlo: Strong Privacy for Analytics in the Crowd},
    booktitle = {Proceedings of the 26th Symposium on Operating Systems Principles},
    series = {SOSP '17},
    year = {2017},
}

@inproceedings{DingKY17,
  author    = {Bolin Ding and
               Janardhan Kulkarni and
               Sergey Yekhanin},
  title     = {Collecting Telemetry Data Privately},
  booktitle = {Advances in Neural Information Processing Systems (NIPS)},
  year      = {2017},
}

@misc{pl94:2020,
  title = {{2020 Census Redistricting Data (Public Law (P.L.) 94-171) Summary File -- United States machine-readable data files/prepared by the U.S. Census Bureau, 2021.}},
  author = {{U.S. Census Bureau}},
  year={2021}
}

@misc{pl94:2020:codebase,
  title = {{DAS 2020 Redistricting Production Code Release}},
  author = {{U.S. Census Bureau}},
  note = {\url{https://github.com/uscensusbureau/DAS_2020_Redistricting_Production_Code}},
  year={2021}
}

@inproceedings{pinq,
 author = {McSherry, Frank},
 title = {Privacy Integrated Queries: An Extensible Platform for Privacy-preserving Data Analysis},
 booktitle = {Proceedings of the 2009 ACM SIGMOD International Conference on Management of Data},
 year = {2009},
 pages = {19-30},
 note={\url{https://doi.org/10.1145/1559845.1559850}}
 }

@article{Dwork:2014:AFD,
 author = {Dwork, Cynthia and Roth, Aaron},
 title = {The Algorithmic Foundations of Differential Privacy},
 journal = {Found. Trends Theor. Comput. Sci.},
 volume = {9},
 year = {2014},
 issn = {1551-305X},
 pages = {211--407},
 publisher = {Now Publishers Inc.},
 address = {Hanover, MA, USA},
}

@article{rubin1974,
    title={{Characterizing the Estimation of Parameters in Incomplete Data Problems}},
    author={Donald B. Rubin},
    journal={Journal of the American Statistical Association},
    volume={69},
    year={1974},
    pages={467-474},
    }

@article{rubin1976,
    title={{Inference and Missing Data}},
    author={Donald B. Rubin},
    journal={Biometrika},
    volume={63},
    year={1976},
    pages={581-592},
}

@misc{hansell2006aol,
  title={AOL {Removes} {Search} {Data} on {Group} of {Web} {Users} {New York Times}},
  author={Hansell, S},
  year={2006},
  note={\url{https://www.nytimes.com/2006/08/08/business/media/08aol.html}}
}

@article{narayanan2008robust,
  title={{Robust De-anonymization of Large Datasets (How to Break Anonymity of the Netflix Prize Dataset)}},
  author={Narayanan, Arvind and Shmatikov, Vitaly},
  journal={University of Texas at Austin},
  year={2008},
  note={\url{https://arxiv.org/pdf/cs/0610105.pdf}}
}

@article{barth2012re,
  title={{The `Re-identification' of Governor William Weld's Medical Information: a Critical Re-examination of Health Data Identification Risks and Privacy Protections}},
  author={Barth-Jones, Daniel},
  journal={Then and Now (July 2012)},
  year={2012}
}

@inproceedings{dwork2006our,
  title={{Our Data, Ourselves: Privacy Via Distributed Noise Generation}},
  author={Dwork, Cynthia and Kenthapadi, Krishnaram and McSherry, Frank and Mironov, Ilya and Naor, Moni},
  booktitle={Annual International Conference on the Theory and Applications of Cryptographic Techniques},
  pages={486--503},
  year={2006},
  organization={Springer}
}

@article{sweeney2002k,
  title={{k-anonymity: A Model for Protecting Privacy}},
  author={Sweeney, Latanya},
  journal={International Journal of Uncertainty, Fuzziness and Knowledge-Based Systems},
  volume={10},
  number={05},
  pages={557--570},
  year={2002},
  publisher={World Scientific}
}

@article{garfinkel2015identification,
  title={{De-identification of Personal Information}},
  author={Garfinkel, Simson L.},
  journal={National Institute of Standards and Technology},
  year={2015},
  note={\url{https://doi.org/10.6028/NIST.IR.8053}}
}

@article{cohen2018linear, 
    title={Linear Program Reconstruction in Practice}, 
    volume={10}, 
    note = {\url{https://journalprivacyconfidentiality.org/index.php/jpc/article/view/711}, 
    DOI={10.29012/jpc.711}}, 
    abstractNote={&amp;lt;p&amp;gt;We briefly report on a successful linear program reconstruction attack performed on a production statistical queries system and using a real dataset. The attack was deployed in test environment in the course of the Aircloak Challenge bug bounty program and is based on the reconstruction algorithm of Dwork, McSherry, and Talwar. We empirically evaluate the effectiveness of the algorithm and a related algorithm by Dinur and Nissim with various dataset sizes, error rates, and numbers of queries in a Gaussian noise setting.&amp;lt;/p&amp;gt;}, 
    number={1}, 
    journal={Journal of Privacy and Confidentiality}, 
    author={Cohen, Aloni and Nissim, Kobbi}, 
    year={2020}, 
    month={1} 
}

@inproceedings{wong2007minimality,
  title={{Minimality Attack in Privacy Preserving Data Publishing}},
  author={Wong, Raymond Chi-Wing and Fu, Ada Wai-Chee and Wang, Ke and Pei, Jian},
  booktitle={Proceedings of the 33rd international conference on Very large data bases},
  pages={543--554},
  year={2007}
}

@inproceedings{kifer2009attacks,
  title={{Attacks on Privacy and deFinetti's Theorem}},
  author={Kifer, Daniel},
  booktitle={Proceedings of the 2009 ACM SIGMOD International Conference on Management of data},
  pages={127--138},
  year={2009}
}

@article{warner1965randomized,
  title={{Randomized Response: A Survey Technique for Eliminating Evasive Answer Bias}},
  author={Warner, Stanley L},
  journal={Journal of the American Statistical Association},
  volume={60},
  number={309},
  pages={63--69},
  year={1965},
  publisher={Taylor \& Francis}
}

@techreport{canonne2020discrete,
      title={{The Discrete Gaussian for Differential Privacy}}, 
      author={Clément L. Canonne and Gautam Kamath and Thomas Steinke},
      year={2021},
      note={\url{https://arxiv.org/abs/2004.00010}},
      eprint={2004.00010},
      archivePrefix={arXiv},
      primaryClass={cs.DS}
}

@inproceedings{10.1007/978-3-662-53641-4_24,
    author = {Bun, Mark and Steinke, Thomas},
    title = {{Concentrated Differential Privacy: Simplifications, Extensions, and Lower Bounds}},
    year = {2016},
    isbn = {9783662536407},
    publisher = {Springer-Verlag},
    address = {Berlin, Heidelberg},
    note ={\url{https://doi.org/10.1007/978-3-662-53641-4_24}, 
    full text \url{https://link.springer.com/chapter/10.1007\%2F978-3-662-53641-4_24}},
    booktitle = {Proceedings, Part I, of the 14th International Conference on Theory of Cryptography - Volume 9985},
    pages = {635–658},
    numpages = {24}
}

@techreport{bun2016concentrated,
      title={{Concentrated Differential Privacy: Simplifications, Extensions, and Lower Bounds}}, 
      author={Mark Bun and Thomas Steinke},
      note={arXiv extended version \url{https://arxiv.org/abs/1605.02065}},
      year={2016},
      eprint={1605.02065},
      archivePrefix={arXiv},
      primaryClass={cs.CR}
}

@techreport{bell:schafer:2022,
    title={Simulation Studies to Investigate Variation in Census Counts and in Census Coverage Error Using 2010 {SF-1} Data and 2010 {CCM} Results},
    author={Bell, William and Schafer, Joseph},
    note={\url{https://www2.census.gov/adrm/CED/Papers/CY22/2022-01-simulation-studies.pdf}},
    year={2022}
}

@misc{DASteam:uncertainty:2022,
    title={Understanding Disclosure Avoidance-Related Variability in the 2020 Census Redistricting Data},
    author={{U.S. Census Bureau}},
    note={\url{https://www.census.gov/library/fact-sheets/2022/variability.html}},
    year={2022}
    
}

@inproceedings{NEURIPS2020_b53b3a3d,
 author = {Canonne, Cl\'{e}ment L and Kamath, Gautam and Steinke, Thomas},
 booktitle = {{Advances in Neural Information Processing Systems}},
 editor = {H. Larochelle and M. Ranzato and R. Hadsell and M. F. Balcan and H. Lin},
 pages = {15676--15688},
 publisher = {Curran Associates, Inc.},
 title = {{The Discrete Gaussian for Differential Privacy}},
 note = {\url{https://proceedings.neurips.cc/paper/2020/file/b53b3a3d6ab90ce0268229151c9bde11-Paper.pdf}},
 volume = {33},
 year = {2020}
}

@misc{CensusDataProductsCrosswalk,
      title={{2020 Census Data Products Crosswalk}}, 
      author={{U.S. Census Bureau}},
      year={2020},
      note={\url{https://www2.census.gov/programs-surveys/decennial/2020/program-management/data-product-planning/2010-demonstration-data-products/2020-census-data-products-planning-crosswalk.xlsx}}
}

@techreport{CensusFederalRegistryRedistrictingDataFile,
      title={{Final Content Design for the Prototype 2020 Census Redistricting Data File}}, 
      year={2018},
      author={{Federal Register}},
      eprint={83 FR 19042},
      note ={\url{https://www.federalregister.gov/documents/2018/05/01/2018-09189/final-content-design-for-the-prototype-2020-census-redistricting-data-file}}
}

@misc{DASMetricsOverviewApril28,
      title={{Data Metrics for 2020 Disclosure Avoidance: Update for the April 28, 2021 release}}, 
      year={2021},
      author={{U.S. Census Bureau}},
      note={\url{https://www2.census.gov/programs-surveys/decennial/2020/program-management/data-product-planning/2010-demonstration-data-products/ppmf20210428/2021-04-28-das-metrics-overview.pdf}}
}

@misc{DASMetricsOverviewNov16,
      title={{Data Metrics for 2020 Disclosure Avoidance: November 16, 2020 release}}, 
      year={2020},
      author={{U.S. Census Bureau}},
      note={\url{https://www2.census.gov/programs-surveys/decennial/2020/program-management/data-product-planning/2010-demonstration-data-products/ppmf20201116/2020-11-16-das-metrics-overview.pdf}}
}

@misc{AbowdCSAC2018,
      title={{Disclosure Avoidance for Block Level Data and Protection of Confidentiality in Public Tabulations}}, 
      year={2018},
      author={John M. Abowd},
      howpublished = {Census Scientific Advisory Committee Presentation},
      note={\url{https://www2.census.gov/cac/sac/meetings/2018-12/abowd-disclosure-avoidance.pdf}}
}

@article{abowdschmutte2019,
     author={John M. Abowd and Ian M. Schmutte},
     title={{An Economic Analysis of Privacy Protection and Statistical Accuracy as Social Choices}},
     journal={American Economic Review},
     volume={109},
     number={1},
     month={1},
     year={2019},
     pages={171-202}, 
     note={\url{https://doi.org/10.1257/aer.20170627}} 
}

@article{abowd:schmutte:2015,
    author = {John M. Abowd and Ian M. Schmutte},
    title = {{Economic Analysis and Statistical Disclosure Limitation}},
    journal = {Brookings Papers on Economic Activity},
    year = {2015},
    pages={221-267},
    note={\url{http://www.brookings.edu/~/media/Projects/BPEA/Spring-2015-Revised/AbowdText.pdf?la=en}}
}

@article{ghosh2012universally,
  title={{Universally Utility-Maximizing Privacy Mechanisms}},
  author={Ghosh, Arpita and Roughgarden, Tim and Sundararajan, Mukund},
  journal={SIAM Journal on Computing},
  volume={41},
  number={6},
  pages={1673--1693},
  year={2012},
  publisher={SIAM}
}

@article{fellegi:1972,
  title={{On the Question of Statistical Confidentiality}},
  author={Fellegi, Ivan P},
  journal={{Journal of the American Statistical Association}},
  volume={67},
  number={337},
  pages={7--18},
  year={1972},
  publisher={Taylor \& Francis Group}
}

@techreport{mckenna:2018,
  author={Laura McKenna},
  title={{Disclosure Avoidance Techniques Used for the 1970 through 2010 Decennial Censuses of Population and Housing}},
  year={2018},
  month={11},
  institution={Center for Enterprise Dissemination, U.S. Census Bureau},
  type={Working Papers},
  note={\url{https://ideas.repec.org/p/cen/wpaper/18-47.html}},
  number={18-47}
}

@TechReport{mckenna:2019,
  author={Laura McKenna},
  title={{Disclosure Avoidance Techniques Used for the 1960 Through 2010 Decennial Censuses of Population and Housing Public Use Microdata Samples}},
  year={2019},
  month={4},
  institution={Center for Enterprise Dissemination, U.S. Census Bureau},
  type={Working Papers},
  note={\url{https://www.census.gov/library/working-papers/2019/adrm/six-decennial-censuses-da.html}}
 }

@techreport{hotchkiss:phelan:2017,
    title={{Uses of Census Bureau Data in Federal Funds Distribution}}, 
    author={Marisa Hotchkiss and Jesica Phelan},
    year={2017},
    month={9},
    institution={U.S. Census Bureau},
    note={\url{https://www2.census.gov/programs-surveys/decennial/2020/program-management/working-papers/Uses-of-Census-Bureau-Data-in-Federal-Funds-Distribution.pdf}}
}

@techreport{reamer:2020,
    title={{Counting for Dollars 2020: The Role of the Decennial Census in the Geographic Distribution of Federal Funds}},
    author={Andrew Reamer},
    year={2020},
    month={4},
    institution={George Washington University},
    note ={\url{https://gwipp.gwu.edu/counting-dollars-2020-role-decennial-census-geographic-distribution-federal-funds} accessed on January 17, 2022}
}

@misc{2020_das_development,
    title = {{Developing the DAS: Demonstration Data and Progress Metrics}},
    author = {{U.S. Census Bureau}},
    note={\url{https://www.census.gov/programs-surveys/decennial-census/decade/2020/planning-management/process/disclosure-avoidance/2020-das-development.html
    }},
    note = {Accessed: 2022-01-15},
    year = {2021}
}

@inproceedings{cox1987research,
  title={{Research at the Census Bureau into Disclosure Avoidance Technique for Tabular Data}},
  author={Cox, Lawrence H and Fagan, James T and Greenberg, Brian and Hemmig, Robert},
  booktitle={Proceedings of the Section on Survey Research Methods, American Statistical Association},
  year={1987}
}

@article{DBLP:journals/corr/DworkR16,
  author    = {Cynthia Dwork and
               Guy N. Rothblum},
  title     = {Concentrated Differential Privacy},
  journal   = {CoRR},
  volume    = {abs/1603.01887},
  year      = {2016},
  note = {\url{http://arxiv.org/abs/1603.01887}},
  archivePrefix = {arXiv},
  eprint    = {1603.01887},
  bibsource = {dblp computer science bibliography, https://dblp.org}
}

@article{cox:1980,
 ISSN = {01621459},
 author = {Lawrence H. Cox},
 journal = {Journal of the American Statistical Association},
 number = {370},
 pages = {377--385},
 publisher = {[American Statistical Association, Taylor & Francis, Ltd.]},
 title = {Suppression Methodology and Statistical Disclosure Control},
 volume = {75},
 year = {1980},
  note = {\url{http://www.jstor.org/stable/2287463}}
}

@techreport{AIANspine,
    title = {{Geographic Spines in the 2020 Census Disclosure Avoidance System Topdown Algorithm}},
    author = {Abowd, John M. and Ashmead, Robert and Cumings-Menon, Ryan and Kifer, Daniel and Leclerc, Philip and Ocker, Jeffrey and Ratcliffe, Michael and Zhuravlev, Pavel},
    year = {2021},
    month = {9},
    note = {\url{https://www.census.gov/library/working-papers/2021/adrm/geographic-spines.html} accessed on January 17, 2022}
}

@article{anon:2021,  
    title={{An Uncertainty Principle is a Price of Privacy-Preserving Microdata}},
    author={Abowd, John M. and Ashmead, Robert and Cumings-Menon, Ryan and Garfinkel, Simson L. and Kifer, Daniel and Leclerc, Philip and Sexton, William and Simpson, Ashley and Task, Christine and Zhuravlev, Pavel},
    journal={Advances in Neural Information Processing Systems},
    volume={34},
    year={2021},
    note={\url{https://arxiv.org/abs/2110.13239}},
}

@misc{pl94:law,
    title={{Public Law 94-171}},
    author={{94th Congress}},
    year={1975},
    month={12},
    note={\url{https://www.govinfo.gov/content/pkg/STATUTE-89/pdf/STATUTE-89-Pg1023.pdf}}
}

@misc{pl111:law,
    title={{Public Law 111-117}},
    author={{111th Congress}},
    year={2009},
    month={12},
    note={\url{https://www.govinfo.gov/content/pkg/PLAW-111publ117/pdf/PLAW-111publ117.pdf} and see ``Notes'' \url{https://www.law.cornell.edu/uscode/text/13/5}}
}

@misc{ncsl:2021,
    title={{Redistricting Criteria}},
    author={{National Conference of State Legislatures}},
    year={2021},
    month={7},
    note={\url{https://www.ncsl.org/research/redistricting/redistricting-criteria.aspx} accessed on January 17, 2022}
}

@techreport{wright:irimata:2021,
    title={{Empirical Study of Two Aspects of the TopDown Algorithm Output for Redistricting: Reliability \& Variability}},
    author={Tommy Wright and Kyle Irimata},
    year={2021},
    note={\url{https://www.census.gov/library/working-papers/2021/adrm/SSS2021-01.html}, updated \url{https://www.census.gov/library/working-papers/2021/adrm/SSS2021-02.html}}
}

@misc{spd15,
    title={{Revisions to the Standards for the Classification of Federal Data on Race and Ethnicity}},
    author={{Office of Management and Budget}},
    year={1997},
    note={\url{https://www.govinfo.gov/content/pkg/FR-1997-10-30/pdf/97-28653.pdf}}
}

@inproceedings{garfinkel:leclerc:2020,
author = {Garfinkel, Simson L. and Leclerc, Philip},
title = {Randomness Concerns When Deploying Differential Privacy},
year = {2020},
isbn = {9781450380867},
publisher = {Association for Computing Machinery},
address = {New York, NY, USA},
note= {\url{https://doi.org/10.1145/3411497.3420211}},
%doi = {10.1145/3411497.3420211},
abstract = {The U.S. Census Bureau is using differential privacy (DP) to protect confidential respondent data collected for the 2020 Decennial Census of Population &amp; Housing. The Census Bureau's DP system is implemented in the Disclosure Avoidance System (DAS) and requires a source of random numbers. We estimate that the 2020 Census will require roughly 90TB of random bytes to protect the person and household tables. Although there are critical differences between cryptography and DP, they have similar requirements for randomness. We review the history of random number generation on deterministic computers\o{}mitt, including von Neumann's "middle-square'' method, Mersenne Twister (MT19937) (the default NumPy random number generator, which we conclude is unacceptable for use in production privacy-preserving systems), and the Linux /dev/urandom device. We also review hardware random number generator schemes, including the use of so-called "Lava Lamps'' and the Intel Secure Key RDRAND instruction. We finally present our plan for generating random bits in the Amazon Web Services (AWS) environment using AES-CTR-DRBG seeded by mixing bits from/dev/urandom and the Intel Secure Key RDSEED instruction, a compromise of our desire to rely on a trusted hardware implementation, the unease of our external reviewers in trusting a hardware-only implementation, and the need to generate so many random bits.},
booktitle = {Proceedings of the 19th Workshop on Privacy in the Electronic Society},
pages = {73–86},
numpages = {14},
keywords = {randomness, us census bureau, differential privacy, rdrand},
location = {Virtual Event, USA},
series = {WPES'20}
}

@article{wasserman:zhou:2010,
author = {Larry Wasserman and Shuheng Zhou},
title = {{A Statistical Framework for Differential Privacy}},
journal = {Journal of the American Statistical Association},
volume = {105},
number = {489},
pages = {375-389},
year  = {2010},
publisher = {Taylor & Francis},
%doi = {10.1198/jasa.2009.tm08651},
note={\url{https://doi.org/10.1198/jasa.2009.tm08651}},
%eprint = {https://doi.org/10.1198/jasa.2009.tm08651}
}

@article{hogan:et:al:2013,
    author={Howard Hogan and Patrick J. Cantwell and Jason Devine and Vincent T. Mule and Victoria Velkoff},
    title={Quality and the 2010 Census},
    journal={{Population Research Policy Review}},
    note={\url{https://doi.org/10.1007/s11113-013-9278-5}},
    year={2013},
    volume={32},
    pages={637–662}
}

@inproceedings{mironov:2012,
author = {Mironov, Ilya},
title = {On Significance of the Least Significant Bits for Differential Privacy},
year = {2012},
isbn = {9781450316514},
publisher = {Association for Computing Machinery},
address = {New York, NY, USA},
note = {\url{https://doi.org/10.1145/2382196.2382264}},
%doi = {10.1145/2382196.2382264},
abstract = {We describe a new type of vulnerability present in many implementations of differentially private mechanisms. In particular, all four publicly available general purpose systems for differentially private computations are susceptible to our attack.The vulnerability is based on irregularities of floating-point implementations of the privacy-preserving Laplacian mechanism. Unlike its mathematical abstraction, the textbook sampling procedure results in a porous distribution over double-precision numbers that allows one to breach differential privacy with just a few queries into the mechanism.We propose a mitigating strategy and prove that it satisfies differential privacy under some mild assumptions on available implementation of floating-point arithmetic.},
booktitle = {Proceedings of the 2012 ACM Conference on Computer and Communications Security},
pages = {650–661},
numpages = {12},
keywords = {differential privacy, floating point arithmetic},
location = {Raleigh, North Carolina, USA},
series = {CCS '12}
}

@article{Balcer:Vadhan:2019, 
   title={Differential Privacy on Finite Computers},
   volume={9},
   number={2},
   journal={Journal of Privacy and Confidentiality},
   author={Balcer, Victor and Vadhan, Salil},
   year={2019},
   month={Sep.}
}

@article{dobra2000bounds,
  title={{Bounds for Cell Entries in Contingency Tables Given Marginal Totals and Decomposable Graphs}},
  author={Dobra, Adrian and Fienberg, Stephen E},
  journal={Proceedings of the National Academy of Sciences},
  volume={97},
  number={22},
  pages={11885--11892},
  year={2000},
  publisher={National Academy of Sciences}
}

@article{fienberg2005preserving,
  title={{Preserving the Confidentiality of Categorical Statistical Databases when Releasing Information for Association Rules}},
  author={Fienberg, Stephen E and Slavkovic, Aleksandra B},
  journal={Data Mining and Knowledge Discovery},
  volume={11},
  number={2},
  pages={155--180},
  year={2005},
  publisher={Springer}
}

@misc{title13,
    title={{Title 13 - Census}},
    author={{United States Code}},
    year={1954},
    note={\url{https://www.govinfo.gov/content/pkg/USCODE-2007-title13/pdf/USCODE-2007-title13.pdf}}
}

@misc{calprivacyact,
    title={{An act to add Title 1.81.5 (commencing with Section 1798.100) to Part
4 of Division 3 of the Civil Code, relating to privacy }},
    author={{State of California}},
    year={2018},
    note={\url{https://leginfo.legislature.ca.gov/faces/billPdf.xhtml?bill_id=201720180AB375&version=20170AB37591CHP}}
}

@misc{EUgeneralDataProtection,
    title={{REGULATION (EU) 2016/679 OF THE EUROPEAN PARLIAMENT AND OF THE COUNCIL }},
    author={{Official Journal of the European Union}},
    year={2016},
    note={\url{https://eur-lex.europa.eu/legal-content/EN/TXT/PDF/?uri=CELEX:32016R0679}}
}

@inproceedings{mcsherry2007mechanism,
  title={{Mechanism Design via Differential Privacy}},
  author={McSherry, Frank and Talwar, Kunal},
  booktitle={48th Annual IEEE Symposium on Foundations of Computer Science (FOCS'07)},
  pages={94--103},
  year={2007},
  organization={IEEE},
  note={\url{https://doi.org/10.1109/FOCS.2007.66}}
}

@article{kasiviswanathan2014semantics,
  title={{On the `Semantics' of Differential Privacy: A Bayesian Formulation}},
  author={Kasiviswanathan, Shiva P and Smith, Adam},
  journal={Journal of Privacy and Confidentiality},
  volume={6},
  number={1},
  year={2014},
  note = {\url{https://doi.org/10.29012/jpc.v6i1.634}}
}

@inproceedings{balle2020hypothesis,
  title={{Hypothesis Testing Interpretations and Renyi Differential Privacy}},
  author={Balle, Borja and Barthe, Gilles and Gaboardi, Marco and Hsu, Justin and Sato, Tetsuya},
  booktitle={International Conference on Artificial Intelligence and Statistics},
  pages={2496--2506},
  year={2020},
  organization={PMLR}
}

@article{dong2021gaussian,
  title={{Gaussian Differential Privacy}},
  author={Dong, Jinshuo and Roth, Aaron and Su, Weijie},
  journal={Journal of the Royal Statistical Society Series B},
  year={2021},
  note={\url{https://arxiv.org/abs/1905.02383}}
}

@inproceedings{kairouz2015composition,
  title={{The Composition Theorem for Differential Privacy}},
  author={Kairouz, Peter and Oh, Sewoong and Viswanath, Pramod},
  booktitle={International conference on machine learning},
  pages={1376--1385},
  year={2015},
  organization={PMLR}
}

@inproceedings{seeman2020private,
  title={{Private Posterior Inference Consistent with Public Information: A Case Study in Small Area Estimation from Synthetic Census Data}},
  author={Seeman, Jeremy and Slavkovic, Aleksandra and Reimherr, Matthew},
  booktitle={International Conference on Privacy in Statistical Databases},
  pages={323--336},
  year={2020},
  organization={Springer}
}

@book{blau1977inequality,
  title={{Inequality and Heterogeneity: A Primitive Theory of Social Structure}},
  author={Blau, Peter Michael},
  volume={7},
  year={1977},
  publisher={Free Press New York}
}

@inproceedings{Dinur:Nissim:2003,
 author = {Dinur, Irit and Nissim, Kobbi},
 title = {Revealing Information While Preserving Privacy},
 booktitle = {Proceedings of the Twenty-second ACM SIGMOD-SIGACT-SIGART Symposium on Principles of Database Systems},
 year = {2003},
}

@article{homer:etal:2008,
    doi = {10.1371/journal.pgen.1000167},
    author = {Homer, Nils AND Szelinger, Szabolcs AND Redman, Margot AND Duggan, David AND Tembe, Waibhav AND Muehling, Jill AND Pearson, John V. AND Stephan, Dietrich A. AND Nelson, Stanley F. AND Craig, David W.},
    journal = {PLOS Genetics},
    publisher = {Public Library of Science},
    title = {Resolving Individuals Contributing Trace Amounts of DNA to Highly Complex Mixtures Using High-Density SNP Genotyping Microarrays},
    year = {2008},
    month = {08},
    volume = {4},
    url = {https://doi.org/10.1371/journal.pgen.1000167},
    pages = {1-9},
    number = {8},
}

@article{dwork:etal:2017,
author = {Dwork, Cynthia and Smith, Adam and Steinke, Thomas and Ullman, Jonathan},
title = {Exposed! A Survey of Attacks on Private Data},
journal = {Annual Review of Statistics and Its Application},
volume = {4},
number = {1},
pages = {61-84},
year = {2017},
doi = {10.1146/annurev-statistics-060116-054123},

URL = { 
        https://doi.org/10.1146/annurev-statistics-060116-054123
    
},
eprint = { 
        https://doi.org/10.1146/annurev-statistics-060116-054123
    
}
}

@InProceedings{murtagh:vadhan:10.1007/978-3-662-49096-9_7,
author="Murtagh, Jack
and Vadhan, Salil",
editor="Kushilevitz, Eyal
and Malkin, Tal",
title="The Complexity of Computing the Optimal Composition of Differential Privacy",
booktitle="Theory of Cryptography",
year="2016",
publisher="Springer Berlin Heidelberg",
address="Berlin, Heidelberg",
pages="157--175",
isbn="978-3-662-49096-9"
}

@techreport{JASON:2020,
    author={{The JASON Group}},
    title={{Formal Privacy Methods for the 2020 Census}},
    note={\url{https://www.census.gov/programs-surveys/decennial-census/decade/2020/planning-management/plan/planning-docs/privacy-methods-2020-census.html}},
    year=2020
}

@inproceedings{asoodeh:etal:2020:9b52f4c81b6f4c50a766fe9675b81066,
title = "A Better Bound Gives a Hundred Rounds: Enhanced Privacy Guarantees via f-Divergences",
author = "Shahab Asoodeh and Jiachun Liao and Calmon, {Flavio P.} and Oliver Kosut and Lalitha Sankar",
note = "Funding Information: This work was supported in part by NSF under grants CIF 1922971, 1815361, 1742836, 1900750, and CIF CAREER 1845852. Publisher Copyright: {\textcopyright} 2020 IEEE.; 2020 IEEE International Symposium on Information Theory, ISIT 2020 ; Conference date: 21-07-2020 Through 26-07-2020",
year = "2020",
month = jun,
doi = "10.1109/ISIT44484.2020.9174015",
language = "English (US)",
series = "IEEE International Symposium on Information Theory - Proceedings",
publisher = "Institute of Electrical and Electronics Engineers Inc.",
pages = "920--925",
booktitle = "2020 IEEE International Symposium on Information Theory, ISIT 2020 - Proceedings",
}

@book{wolsey:1999,
    title={{Integer and Combinatorial Optimization}},
    author={Nemhauser, George L. and Wolsey, Laurence A.},
    year={1999},
    publisher={Wiley Interscience},
    isbn={13: 978-0471359432}
}

\end{document}